\documentclass[11pt]{article}
\usepackage{fullpage}
\usepackage[bookmarks]{hyperref}
\usepackage{amssymb}
\usepackage{amsmath}
\usepackage{amsthm}
\usepackage{graphicx,color,colordvi}
\usepackage{bbm}
\usepackage{cleveref}
\usepackage{stmaryrd}
\usepackage[utf8]{inputenc}
\usepackage[blocks]{authblk}
\usepackage{dsfont}
\usepackage{mathtools}
\usepackage{authblk}
\newcommand{\norm}[1]{\left\| #1 \right\|}

\usepackage{pgfplots}

\newcommand{\ssubset}{\subset\joinrel\subset}

\def\openone{\leavevmode\hbox{\small1\kern-3.8pt\normalsize1}}

\def\CC{\mathbb{C}}
\def\RR{\mathbb{R}}
\def\ZZ{\mathbb{Z}}

\def\NN{\mathbb{N}}
\def\LL{\mathbb{L}}

\def\11{\mathbb{I}}
\def\LL{\mathcal{L}}

\def\RRR{\mathcal{R}}

\newtheorem{theorem}{Theorem}
\newtheorem{lemma}{Lemma}
\newtheorem{proposition}{Proposition}
\newtheorem{corollary}{Corollary}
\newtheorem{step}{Step}

\theoremstyle{definition}
\newtheorem{definition}{Definition}

\newtheorem{step2}{Step}

\newcommand{\dist}{{\operatorname{dist}}}

\def\eps{\varepsilon}

\newcommand{\vertiii}[1]{{\left\vert\kern-0.25ex\left\vert\kern-0.25ex\left\vert #1 
		\right\vert\kern-0.25ex\right\vert\kern-0.25ex\right\vert}}
\newcommand{\qty}[1]{\left\lbrace #1   \right\rbrace}
\newcommand{\supp}{\mathop{\rm supp}\nolimits}

\newcommand{\tr}{\mathop{\rm Tr}\nolimits}

\newcommand{\ds}{\displaystyle}
\newcommand{\ccc}{{\operatorname{c}}}

\usepackage{pst-node}
\usepackage{tikz-cd}

\newcommand{\ket}[1]{|#1\rangle}

\newcommand{\cA}{{\cal A}}
\newcommand{\cB}{{\cal B}}
\newcommand{\cD}{{\cal D}}
\newcommand{\cE}{{\cal E}}
\newcommand{\cF}{{\mathcal{F}}}

\newcommand{\cN}{{\cal N}}
\newcommand{\cT}{{\cal T}}
\newcommand{\cH}{{\cal H}}
\newcommand{\cK}{{\cal K}}
\newcommand{\cJ}{{\cal J}}

\newcommand{\cS}{\mathcal{S}}
\newcommand{\cP}{\mathcal{P}}

\newcommand{\cL}{{\cal L}}

\newcommand{\cM}{{\mathcal{M}}}
\newcommand{\cO}{{\cal O}}

\newcommand{\Id}{{\mathds{1}}}

\newcommand{\R}{{\mathbb{R}}}

\def\e{\mathrm{e}}
\newcommand{\ketbra}[2]{|#1\rangle\langle #2|}

\usepackage{graphicx}
\usepackage{setspace}
\usepackage{verbatim}
\usepackage{subfig}

\theoremstyle{definition}

\theoremstyle{remark}
\newtheorem{remark}{Remark}
\newtheorem{condition}{Condition}
\numberwithin{equation}{section}

\newcommand{\abs}[1]{\lvert#1\rvert}

\DeclareRobustCommand\openone{\leavevmode\hbox{\small1\normalsize\kern-.33em1}}

\newcommand{\id}{{\rm{id}}}
\newcommand{\be}{\begin{equation}}
	\newcommand{\ee}{\end{equation}}
\newcommand{\bea}{\begin{eqnarray}}
	\newcommand{\eea}{\end{eqnarray}}
\newcommand{\beas}{\begin{eqnarray*}}
	\newcommand{\eeas}{\end{eqnarray*}}

\title{The modified logarithmic Sobolev inequality for quantum spin systems: classical and commuting nearest neighbour interactions}

\begin{document}

\author[1,2]{{\'A}ngela Capel}
\author[1,2]{Cambyse Rouz\'{e}}
\affil[1]{Department of Mathematics, Technische Universit\"at M\"unchen, 85748 Garching, Germany}
\affil[2]{Munich Center for Quantum Science and Technology (MCQST), M\"unchen, Germany}
\author[3]{Daniel Stilck Fran\c ca}
\affil[3]{QMATH, Department of Mathematical Sciences, University of Copenhagen, Universitetsparken 5, 2100 Copenhagen, Denmark}
\date{}

\maketitle

\begin{abstract}
Given a uniform, frustration-free family of local Lindbladians defined on a quantum lattice spin system in any spatial dimension, we prove a strong exponential convergence in relative entropy of the system to equilibrium under a condition of spatial mixing of the stationary Gibbs states and the rapid decay of the relative entropy on finite-size blocks. Our result leads to the first examples of the positivity of the modified logarithmic Sobolev inequality for quantum lattice spin systems independently of the system size. Moreover, we show that our notion of spatial mixing is a consequence of the recent quantum generalization of Dobrushin and Shlosman's complete analyticity of the free-energy at equilibrium. The latter typically holds above a critical temperature $T_c$. 

Our results have wide-ranging applications in quantum information. As an illustration, we discuss four of them: first, using techniques of quantum optimal transport, we show that a quantum annealer subject to a finite range classical noise will output an energy close to that of the fixed point after constant annealing time. Second, we prove Gaussian concentration inequalities for Lipschitz observables and show that the eigenstate thermalization hypothesis holds for certain high-temperture Gibbs states. Third, we prove a finite blocklength refinement of the quantum Stein lemma for the task of asymmetric discrimination of two Gibbs states of commuting Hamiltonians satisfying our conditions. Fourth, in the same setting, our results imply the existence of a local quantum circuit of logarithmic depth to prepare Gibbs states of a class of commuting Hamiltonians.

In order to prove our main result, we introduce the concept of \textit{peeling}, which refers to the decomposition of the analysis of the evolution into two steps: first, we study the rapidity at which the initial state becomes indistinguishable from the stationary state on some finite-size cubes which tile the lattice. This first step requires the existence of the complete modified logarithmic Sobolev inequality on finite lattice subregions. Then, we show that the convergence of these previously partially ``peeled'' states towards the global Gibbs state rapidly occurs under our condition of spatial mixing. The proof of this last statement requires a newly derived approximate tensorization of the relative entropy between the partially peeled state and the fixed point.

	\end{abstract}

\newpage
\tableofcontents
\newpage
\section{Introduction}
In any realistic setting, a quantum system undergoes unavoidable interactions with its environment. These interactions lead to alterations of the information initially contained in the system. Within the current context of emerging quantum information-processing devices, a proposed solution to the problem of decoherence is to encode the quantum logical information into a highly entangled many-body state in order to protect it from the action of local noise \cite{shor1995scheme,dennis2002topological}.  Such a state will typically belong to the ground space of a Hamiltonian modeling the noiseless, unitary evolution of the system in the absence of an environment. When the environmental noise can be modeled by a Markovian evolution and below some critical temperature, the resulting self-correcting quantum memory should survive for a time which scales at least polynomially with the size of the system. Conversely, faster decoherence was recently used as a viable method for the preparation and control of relevant phases of matter \cite{verstraete2009quantum,diehl2008quantum,kraus2008preparation,syassen2008strong,witthaut2009dissipation}, as well as to estimate the run-time of algorithms based on the efficient preparation of a Gibbs state \cite{brandao2017quantum}. The variety of the aforementioned applications indicates the importance of finding easy criteria for the study of the speed at which quantum lattice spin systems thermalize.

 Since the seminal works of Dobrushin-Shlosman and Stroock-Zegarlinski, equilibrium and out-of-equilibrium properties of classical lattice spin systems are known to be closely related: in their attempt to answer the problem of the analytical dependence of a Gibbs measure to its corresponding potential, Dobrushin and Shlosman introduced twelve equivalent statements, one of which we refer to as the condition of \textit{exponential decay of correlations} (sometimes also referred to as \textit{clustering of correlations}): the correlations between two separated regions $A$ and $B$ of a lattice spin system decay exponentially in the distance separating $A$ from $B$. On the other hand, given a potential, one can construct a  Markov process, called \textit{Glauber dynamics}, whose stationary state coincides with the Gibbs state for the given potential. For these dynamics, Holley and Stroock \cite{holley1976,holley1989uniform} made the key observation that systems thermalizing in times scaling logarithmically in the system size, a property known as \textit{rapid mixing}, satisfy exponential decay of correlations at equilibrium.
The converse implication, namely that exponential decay of correlations implies rapid mixing, was investigated later on in a series of articles \cite{STROOCK1992299,stroock1992b,stroock1992c} by Zegarlinski and Stroock, who proved the stronger condition of an exponential entropic decay of the dynamics towards the limiting Gibbs measure, also known as \textit{logarithmic Sobolev inequality}. Exponential decay of correlations was also proven to be equivalent to the non-closure of the spectral gap of the generator of the Glauber dynamics. Finally, since all these conditions occur above some critical temperature, their equivalence rigorously establishes the equivalence between dynamical and static phase transitions.

Functional inequalities like the logarithmic Sobolev inequality are by now one of the most powerful tools available in the study of classical spin systems \cite{Martinelli1999,zegarlinski1990log,zegarlinski1990logb,ZEGARLINSKI199277,martinelli1994approacha,martinelli1994approach}, and are still the subject of active research \cite{chen2020optimal,cryan2019modified,lee2018stochastic}. They have also found numerous applications in optimization, information theory and probability theory \cite{raginsky2012concentration,villani2008optimal,bakry2013analysis,boucheron2013concentration}, just to name a few. Functional inequalities can be described as differential versions of strong contraction properties of various distance measures under the action of a semigroup. For instance, the \textit{Poincaré inequality} provides an estimate on the Lindbladian's spectral gap and can be understood as quantifying how fast the variance of observables decays under the semigroup. Significantly faster convergence can be shown via the existence of a (modified) logarithmic Sobolev inequality, which implies exponential convergence in relative entropy of any initial state evolving towards equilibrium \cite{[OZ99],[KT13]}, with a rate that is referred to as the \textit{(modified) logarithmic Sobolev constant}. Unlike the spectral gap, this convergence can further be used to provide tight estimates on various capacities of the semigroup \cite{bardet2019group}.

The extension of the above unifying theory to quantum spin systems is still far from being well understood despite a large body of literature devoted to the subject. From the static point of view, the theory of Dobrushin and Shlosman was recently almost completely generalized to the quantum setting in \cite{harrow2020classical} (see also \cite{araki1969gibbs,kliesch2014locality,kuwahara2020clustering,kato2019quantum}), whereas various notions of exponential decay of correlations and area laws were derived under the existence of a functional inequality in \cite{Kastoryano2013a,brandao2015area}. In the low-temperature regime, Temme \cite{Temme2014} proved a lower bound on the spectral gap of the Davies generator corresponding to a stabilizer Hamiltonian in terms of the energy barrier of the corresponding code, hence rigorously connecting the latter to the memory’s lifetime.  On the other hand, based on the previous work of \cite{majewski1995quantum} (see also \cite{majewski1996quantum,[MOZ98],Zegarlinski2000}), Temme and Kastoryano recently showed that, above a critical temperature, any heat-bath dynamics associated with a commuting Hamiltonian satisfies the rapid mixing property \cite{temme2015fast}. Previously, the uniform positivity of the spectral gap for these Markov processes was shown in \cite{[BK16]} to be equivalent to a stronger condition of clustering of the correlations in the Gibbs state between separated regions of the lattice. More recently, exponential clustering of correlations of a Gibbs state was proved to imply its efficient preparation on a quantum \cite{brandao2019finite} or classical \cite{harrow2020classical} computer. In other words, the transition in the phase of a quantum system is also accompanied by a transition in the hardness of approximation \cite{sly2010computational}. 

In spite of these advances in the understanding of quantum Gibbs states, only very few many-body quantum systems are known to satisfy a modified logarithmic Sobolev inequality (see \cite{[KT13],[TPK14],beigi2018quantum,capel2018quantum} for non-interacting systems, and \cite{hyperspectral} for Fermionic systems), and establishing it in generic situations has been an open problem for decades. 
One reason behind this can be explained by the fact that the presence of entanglement  poses significant technical challenges, as most proofs for classical systems rely on concepts that do not generalize to the quantum settings, such as conditioning on the boundary or coupling. The goal of this paper is precisely to find ways around these issues in order to fill in this missing gap.

 \begin{figure}[h!]
	\centering
	\includegraphics[width=1\linewidth]{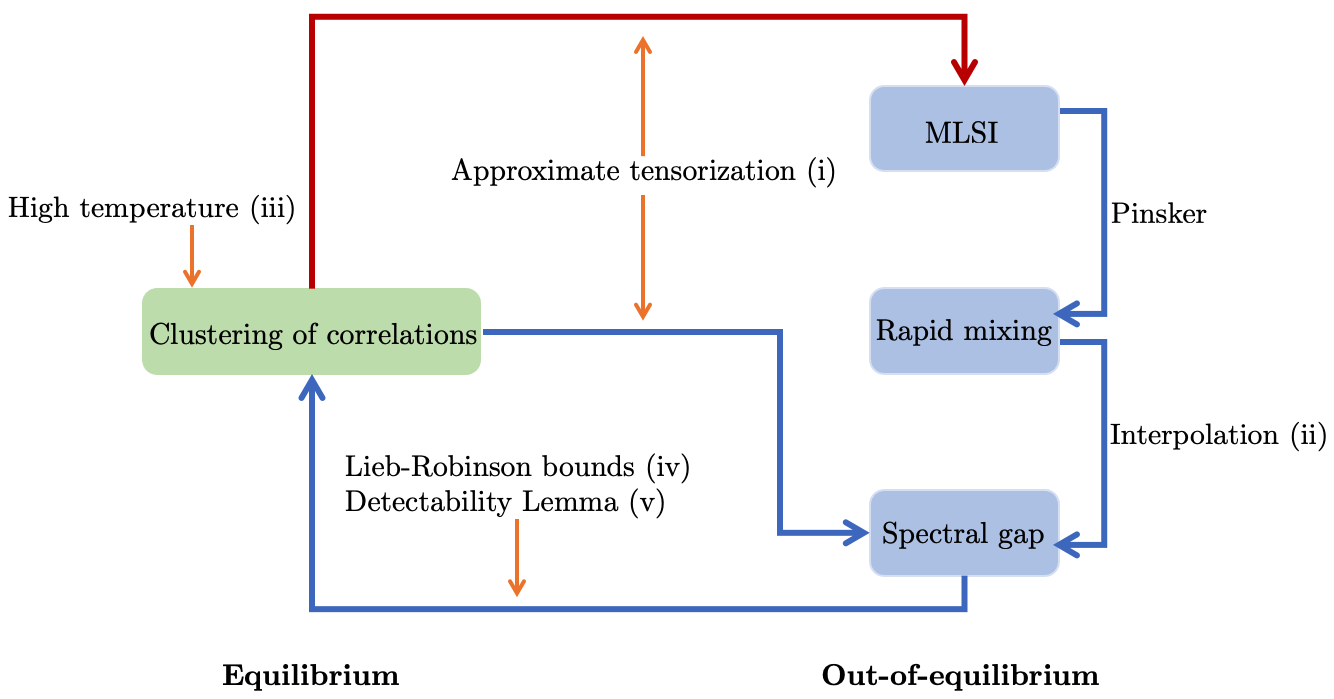}
	\caption{Relations between static and dynamical properties of quantum Gibbs states. The main result of this paper is depicted by the red arrow connecting the notion of clustering of correlations to the existence of a modified logarithmic Sobolev constant. (i) The approximate tensorization of the variance was proved in \cite{[BK16]} to lead to the non-closure of the gap, and an entropic strengthening of that statement is the subject of Theorem \ref{ATAC}. The interpolation (ii) was essentially proven in \cite[Lemma 6]{temme2015fast}.
	We derive the clustering of correlations at high enough temperature and show (iii) by adapting techniques recently pioneered in~\cite{harrow2020classical}, whereas Lieb-Robinson bounds were employed in \cite{Kastoryano2013a} to derive it from the non-closure of the gap (iv). Finally, the detectability lemma was employed in \cite{[BK16]} to derive a stronger notion of clustering of correlation from the gap condition (v).  }
\end{figure}

\paragraph{Main results and proof strategy} 
From a mathematical point of view, our main result constitutes the first complete proof of the existence of the modified logarithmic Sobolev inequality for interacting quantum spin systems independently of the lattice size under exponential clustering of correlations:

\begin{theorem}[MLSI for quantum lattice spin systems (informal)]\label{TheoMLSI} Given the Gibbs state $\sigma_\Lambda$ of a local commuting Hamiltonian $H_\Lambda$ on the $d$-dimensional lattice $\Lambda\ssubset \mathbb{Z}^d$, there exists a local quantum Markov semigroup $(\e^{t\mathcal{L}_{\Lambda*}})_{t\ge 0}$ converging to $\sigma_\Lambda$ exponentially fast in relative entropy distance if $\sigma_\Lambda$ satisfies exponential decay of correlations and any of the three conditions below is satisfied:
\begin{itemize}
\item[$\operatorname{(i)}$] $H_\Lambda$ is classical;
\item[$\operatorname{(ii)}$] $H_\Lambda$ is a nearest neighbour Hamiltonian;
\item[$\operatorname{(iii)}$] $\Lambda$ is a one-dimensional spin chain.
\end{itemize}
More precisely, for every initial state $\rho$
 \begin{align*}
     D(e^{t{\cL}_{\Lambda *}}(\rho)\|\sigma_\Lambda)\leq e^{-\alpha t}D(\rho\|\sigma_\Lambda)\,,
 \end{align*}
 for a constant $\alpha>0$ independent of system size $|\Lambda|$, where $D(\rho\|\sigma_\Lambda):=\tr[\rho\,(\ln\rho-\ln\sigma_\Lambda)].$ Moreover, the notion of decay of correlations that we use holds at any inverse temperature $\beta<\beta_c$, where $\beta_c:=(5 \e g h\kappa)^{-1}$ depends on the locality $\kappa$, the interaction strength $h$ and the growth constant $g$ of $H_\Lambda$. In the case of a classical Hamiltonian, we further prove the equivalence between $\operatorname{(a)}$ the existence of a modified logarithmic Sobolev inequality, $\operatorname{(b)}$ the exponential decay of correlations in $\sigma_\Lambda$, $\operatorname{(c)}$ the uniform positivity of the spectral gap and $\operatorname{(d)}$ rapid mixing. 
\end{theorem}

 We emphasize that proving such a result for quantum systems is nontrivial, even in the case of systems thermalizing to a classical state. This is because the initial state could be highly entangled, and it is a-priori not clear whether entanglement could be used as a resource to substantially slow down the thermalization. Our analysis rigorously proves that this is not the case. Our proof of Theorem \ref{TheoMLSI} is adapted from a modern strategy by \cite{cesi2001quasi}. It splits into three parts:
\newline 

\textit{Strengthened exponential decay of correlations} (Section \ref{clustercorrela}): First, we prove a strengthened exponential decay of correlations below the critical inverse temperature $\beta_c$. For a classical Gibbs state, this condition is precisely the one of Dobrushin-Shlosman. We provide an extension to the commuting, nearest neighbour setting. Our construction of the conditional expectations involved in the result relies on a Schmidt decomposition of the local interactions, which was already used in the study of the local Hamiltonian problem in \cite{bravyi2003commutative}. We refer to Section \ref{clustercorrela} for more details.
\begin{theorem}[Conditioned $\mathbb{L}_1- \mathbb{L}_\infty$ exponential decay of correlations (informal)] 
Let $\sigma_\Lambda$ be the Gibbs state of a commuting nearest neighbour Hamiltonian $H_\Lambda$ at inverse temperature $\beta\le \beta_c$. Then, for any two overlapping regions $C,D\subset \Lambda$, any boundary condition $\omega\equiv \omega_{\partial{C\cup D}}$ and any observable $X^{\omega}$ conditioned on the boundary of $C\cup D$, 
\begin{align*}
    \langle (E^{\omega}_C-E_{C\cup D}^{\omega})(X^{\omega}),\, (E^{\omega}_D-E_{C\cup D}^{\omega})(X^{\omega})\rangle_{\sigma^{\omega}}\le \,c\,|C\cup D|\,\e^{-\operatorname{dist}(D\backslash C,C\backslash D)/\xi}\,\|X^\omega\|_\infty\,\|X^\omega\|_{\mathbb{L}_1(\sigma^\omega)}\,,
\end{align*}
where $\{E^{\omega}_A\}$, $A\in\{C,D,C\cup D\}$, is a family of conditional expectations with respect to $\sigma_\Lambda$, and where $\sigma^\omega$ is the local Gibbs state conditioned on the boundary of $C\cup D$. 
\end{theorem}
Our result extends on the recent quantum generalization of Dobrushin-Shlosman's conditions \cite{harrow2020classical} in two ways: First, we get a bound in terms of the product of an $\mathbb{L}_1$ and an $\mathbb{L}_\infty$ norm, as opposed to the standard albeit weaker $\mathbb{L}_\infty-\mathbb{L}_\infty$ bound. Secondly, our construction in this specific $2$-local setting allows for a local bound in any subregion $C\cup D$ conditioned on its boundary, as opposed to the global bounds found in \cite{harrow2020classical}. This local refinement is crucial to our subsequent proof of the modified logarithmic Sobolev inequality. It is also one of the reasons why we have to restrict our analysis to nearest-neighbours, except for the case  of $1D$ systems, as for those systems a simple coarse-graining argument allows us to reduce the analysis to the nearest neighbour case.
\newline

\textit{Approximate tensorization of the relative entropy} (Section \ref{sec:main}): 
 At the turn of the millennium, a new strategy to prove the modified logarithmic Sobolev inequality for classical spin systems, based on the approximate tensorization of the relative entropy, was provided \cite{cesi2001quasi,[D02]}, which arguably simplifies the classical result of Stroock and Zegarlinski. This strategy's core insight was to realize that Dobrushin and Shlosman's $\mathbb{L}_1-\mathbb{L}_\infty$ exponential decay of correlations could be used to prove the following generalization of the strong subadditivity (SSA) of the entropy, here written with quantum notations for simplicity: for any classical state $\rho$, 
\begin{align}\tag{$\ast$}\label{ast}
    D(\rho\| E_{C\cup D*}^\omega(\rho))\le \,(1+c\,|C\cup D|\,\e^{-\operatorname{dist}(D\backslash C,C\backslash D)/\xi})\,\big(\,D(\rho\| E_{C*}^\omega(\rho))+D(\rho\|E_{D*}^\omega(\rho)) \big)\,.
\end{align}
Indeed, when $\sigma_\Lambda$ is the maximally mixed state, i.e. at $\beta=0$, the (dual) conditional expectation $E_{A*}$ reduces to the partial trace $\tr_A$ on any region $A$, and $c=0$ so that \eqref{ast} reduces to the celebrated SSA \cite{lieb1973proof}: taking $C$ and $D$ non-overlapping, and $B=\Lambda\backslash CD$, $S(BCD)_\rho+S(B)_\rho\le S(BC)_\rho+S(BD)_\rho$. Using the multivariate trace inequalities recently derived in \cite{Sutter2017}, two of the authors extended the result of Cesi to quantum states in \cite{bardet2020approximate}, informally stated below in a more general von Neumann algebraic setting.

\begin{theorem}[Approximate tensorization of the quantum relative entropy \cite{bardet2020approximate} (informal)] \label{theoremapproxtens}
Let $\mathcal{M}\subset \mathcal{N}_1,\mathcal{N}_2\subset\mathcal{N}$ be finite-dimensional von Neumann algebras, with corresponding conditional expectations $E_\mathcal{M}$, $E_1$ and $E_2$. Under a condition of $\mathbb{L}_1-\mathbb{L}_\infty$ clustering of correlations, the following inequality holds: there exists a constant $c$ depending on the clustering, such that for all quantum states $\rho$,
\[
    D(\rho\|E_{\mathcal{M}*}(\rho))\le c\,\big(D(\rho\|E_{1*}(\rho))+D(\rho\|E_{2*}(\rho))  \big)+d(\rho)\,,
\]
where $d(\rho)$ is a $\rho$-dependent additive error term that measures the deviation of $\rho$ from being diagonal in the block decomposition of the matrix algebra $\mathcal{M}$.
\end{theorem}

\textit{Removing additive errors by peeling} (Section \ref{sec:main}): Theorem \ref{TheoMLSI} states the existence of a constant rate $\alpha>0$, independent of the size of $\Lambda$, such that for any initial state $\rho$ evolving according to the semigroup, $    D(\e^{t\mathcal{L}_{\Lambda*}}(\rho)\|\sigma_\Lambda)\le \,\e^{-\alpha \,t}
    D(\rho\|\sigma_\Lambda)$. It turns out that this exponential convergence is equivalent to its derivative with respect to $t$ at $t=0$. The resulting inequality turns out to be the \textit{modified logarithmic Sobolev inequality} (MLSI) that we already mentioned: for any state $\rho$,
\begin{align}\tag{MLSI}
    \alpha\,D(\rho\|\sigma_\Lambda)\le \,-\left.\frac{d}{dt}\right|_{t=0}\,D(\e^{t\mathcal{L}_{\Lambda*}}(\rho)\|\sigma_\Lambda)=\operatorname{EP}_{\mathcal{L}_{\Lambda*}}(\rho)\,.
\end{align}
The right-hand side of the MLSI has the useful property of being linear in the generator $\mathcal{L}_{\Lambda*}$. Moreover, under the $\mathbb{L}_1-\mathbb{L}_\infty$ clustering property of $\sigma_\Lambda$, the approximate tensorization of the relative entropy for classical spins can be used to prove that the left-hand side of MLSI is approximately sub-additive. These two crucial properties led Cesi to formulate the idea of decomposing the problem into regions of a small fixed size, where the MLSI constant $\alpha$ is known to exist. However, the non-vanishing of the constant $d$ on quantum states found in Theorem \ref{theoremapproxtens} is responsible for the failure of Cesi's argument in the quantum regime. In Section \ref{sec:main}, we devise an original argument, which we refer to as \textit{peeling}, in order to manage our way around this issue. From a high-level perspective, our idea consists in proving that any initial state $\rho$ will very quickly converge into a state $\gamma$ whose constant $d(\gamma)$ vanishes on appropriately chosen regions $C\cup D$. For these states, we recover \eqref{ast}, which allows us to conclude our proof. Here again, the restriction to nearest-neighbour interactions appears to be difficult to relax. 
\newline

\textit{Applications} (Section \ref{sec:applications}) We then apply our results to four settings in which the relative entropy decay estimates given by the modified logarithmic Sobolev inequality are crucial. First, we show that the output energy of an Ising quantum annealer subject to finite range classical thermal noise at high enough temperature outputs a state whose energy is close to that of the thermal state of the noise after an annealing time that is constant in system-size. Although the results of~\cite{2009.05532} also allow us to make a similar analysis based on our modified logarithmic Sobolev inequality, here we take a new approach by exploiting quantum optimal transport techniques~\cite{Rouz__2019,gao2018fisher,Carlen20171810}, showcasing the potential of such techniques in the analysis of noisy quantum computation.
Secondly, we apply the results of~\cite{Rouz__2019} to get Gaussian concentration bounds for the outcome distributions of Lipschitz observables on Gibbs states, while also showing how to deduce that the eigenstate thermalization hypothesis holds with use of a MLSI. In both of these applications, using our methods expands the set of observables the results apply to when compared with state-of-the-art~\cite{brandao_equivalence_2015,anshu_concentration_2016,kuwahara2020gaussian,kuwahara_connecting_2016}, albeit for a smaller set of Gibbs states.
Thirdly, we apply our results to quantum asymmetric hypothesis testing. There we show a decay estimate on the type II error for two Gibbs states corresponding to commuting potentials in the finite blocklength regime.
Finally, we also apply our main result to obtain efficient quantum Gibbs samplers for certain Gibbs states corresponding to commuting potentials. Our methods only require the implementation of a circuit of local quantum channels of logarithmic depth, in contrast to previous results~\cite{brandao2019finite} that required quasi-local quantum channels.

\paragraph{Outline of the paper}

In Section \ref{sec2}, we introduce some necessary notation as well as the main concepts of this work, namely quantum Gibbs states, Gibbs samplers and functional inequalities. 

Section \ref{clustercorrela} is devoted to a thorough recap and extension of the various notions of clustering of correlations. We start this section by reviewing Dobrushin-Shlosman's mixing condition and its connection to some functional inequalities and strong ergodicity for classical systems. Then we introduce our key notion of $\mathbb{L}_1 - \mathbb{L}_\infty$ clustering of correlations and establish it above a critical temperature. We hope that this will serve as a useful reference for future work on quantum many-body systems at finite temperature. 

In Section \ref{sec:main}, we expose our main results regarding the relations between $\mathbb{L}_1 - \mathbb{L}_\infty$ clustering and the modified logarithmic Sobolev inequality. We begin by showing that the embedded Glauber dynamics with an additional dephasing satisfy a MLSI. This simple example serves as a toy model for the rest of the section. Afterwards, we implement a tiling of the $d$-dimensional lattice and devise a geometric construction based on grained sets over that tiling, both of which constitute some of the main ingredients for our main result. We present our main result of the existence of MLSI for $d$-dimensional systems, whose proof we leave to Appendix \ref{appendix:nD} and conclude the section by showing the simplified version of this result in one and two dimensions, employing a geometric construction based on rhomboidal regions. 

As a consequence of these results, we present four applications of our result to the field of quantum information theory and information processing in Section \ref{sec:applications}.

\section{Notations and definitions}\label{sec2}

\subsection{Basic notations}\label{subsec:notations}

We denote a finite-dimensional Hilbert space of dimension $d_\cH$ by$(\cH,\langle .|.\rangle)$, the algebra of bounded operators on $\cH$ by $\cB(\cH)$,  by $\cB(\cH)_{\operatorname{sa}}$ the subspace of self-adjoint operators on $\cH$, i.e. $\cB(\cH)_{\operatorname{sa}}:=\left\{X\in\cB(\cH);\ X=X^\dagger\right\}$, where the adjoint of an operator $Y$ is written as $Y^\dagger$, and by $\cB(\cH)_+$ the cone of positive semidefinite operators on $\cH$. We will also use the same notations $\cN_{\operatorname{sa}}$ and $\cN_+$ in the case of a von Neumann subalgebra $\cN$ of $\cB(\cH)$. The identity operator on $\cN$ is denoted by $\Id_\cN$, and we drop the index $\cN$ when it is unnecessary.

Given a map $\Psi:\cB(\cH)\to\cB(\cH)$, we denote its dual with respect to the Hilbert-Schmidt inner product as $\Psi_*$. We also denote by $\id_{\cB(\cH)}$, or simply $\id$, the identity superoperator on $\cB(\cH)$. We further denote by $\mathcal{D}(\cH)$ the set of positive semidefinite, trace one operators on $\cH$, also known as density matrices, and by $\cD(\cH)_+$ the subset of full-rank density operators. In the following, we will often identify a density matrix $\rho\in\mathcal{D}(\cH)$ and the state it defines, that is the positive linear functional $\cB(\cH)\ni X\mapsto\tr[\rho \,X]$.

Given two states $\rho,\sigma\in\cD(\cH)$ with $\supp(\rho)\subseteq\supp(\sigma)$, the \textit{relative entropy} between $\rho$ and $\sigma$ is defined as
\begin{align*}
D(\rho\|\sigma):=\tr\big[\rho\,\big( \ln\rho-\ln\sigma \big)  \big]\,.
\end{align*}
Next, given a state $\sigma\in\cD(\cH)_+$ and a $*$-subalgebra $\cN\subset \cB(\cH)$, a linear map $E:\cB(\cH)\to\cN$ is called a \textit{conditional expectation} with respect to $\sigma$ onto $\cN$ if the following conditions are satisfied  \cite{Aspects2003}: 
\begin{itemize}
    \item[(i)] For all $X\in\cB(\cH)$, $\|E[X]\|\le \|X\|$ .
    \item[(ii)] For all $X\in \cN$, $E[X]=X$ .
    \item[(iii)] For all $X\in\cB(\cH)$, $\tr[\sigma E[X]]=\tr[\sigma X]$ .
\end{itemize}
Given any state $\rho\in\cD(\cH)$ and any state $\sigma=E_{*}(\sigma)$, the following \textit{chain rule} holds true (see for instance Lemma 3.4 in \cite{junge2019stability}):
\begin{align}\label{chainrulerelatent}
D(\rho\|\sigma)=D(\rho\|E_*(\rho))+D(E_*(\rho)\|\sigma)\,.
\end{align}

Moreover, given $p\ge 1$ and a full-rank state $\sigma\in\cD(\cH_\Lambda)$, we also define the following \textit{weighted $\mathbb{L}_p$-norms} on $\cB(\cH)$:
\begin{align*}
\|X\|_{\mathbb{L}_p(\sigma)}:=\Big(\tr\Big[ \big|\sigma^{\frac{1}{2p}}X\sigma^{\frac{1}{2p}}\big|^p\Big]\Big)^{\frac{1}{p}}\,,
\end{align*}
For $p=2$, these norms provide $\cB(\cH)$ with a Hilbert space structure with inner product $\langle X,Y\rangle_\sigma^{\operatorname{KMS}}:=\tr[\sigma^{\frac{1}{2}} X^\dagger \sigma^{\frac{1}{2}}Y]$. We write the resulting \textit{covariance} as 
\begin{align*}
\operatorname{Cov}_\sigma(X,Y):=\langle X-\tr[\sigma X]\Id,\, Y-\tr[\sigma Y]\Id\rangle^{\operatorname{KMS}}_\sigma\,,
\end{align*}
 We also use another notion of covariance based on the so-called GNS inner product $\langle X,Y\rangle_\sigma^{\operatorname{GNS}}:=\tr[\sigma X^\dagger Y] $ (see \cite{kliesch2014locality}):
\begin{align*}
\operatorname{Cov}^{(0)}_\sigma(X,Y):=\langle X-\tr[\sigma X]\Id,\, Y-\tr[\sigma Y]\Id\rangle^{\operatorname{GNS}}_\sigma\,.
\end{align*}
Furthermore, we denote by $\Delta_{\sigma}:X\mapsto \sigma X\sigma^{-1}$ the \textit{modular operator} corresponding to a full-rank state $\sigma$, and by $(\Delta_{\sigma}^{it}:=e^{i\sigma t}.e^{-it\sigma})_{t\in\RR}$ its modular group. Finally, we define the maps $\Gamma_\sigma:X\mapsto \sigma^{\frac{1}{2}}X\sigma^{\frac{1}{2}}$, whose action is to embed $\mathbb{L}_1(\sigma)$ onto the space $\cT_1(\cH)$ of trace-class operators.

\subsection{Quantum Hamiltonians and Gibbs states}

Given a finite region $\Lambda\ssubset \ZZ^d$, we denote by $|\Lambda|$ the number of its sites. The Hilbert space of the system is denoted by $\cH_\Lambda=\bigotimes_{x\in\Lambda}\cH_x$, where $\cH_x$ is a copy of the Hilbert space $\cH$ corresponding to a particle at site $x\in\Lambda$. Given a region $A\subseteq \Lambda$, we denote by $A^c$ the complement of $A$ in $\ZZ^d$, and by $\Lambda\backslash A$ the complement of $A$ in $\Lambda$. The distance between two sites $i,j\in\ZZ^d$ is denoted by $\operatorname{dist}(i,j)$, the distance of a site $i$ to a set $A$ by $\operatorname{dist}(i,A)$, and the distance between two sets $A$ and $B$ by $\operatorname{dist}(A,B)$. We adopt similar notations for graphs.

Let $\{\Phi(X)\}_{X\ssubset \ZZ^d}$ be an $r$-local potential, i.e. for any $X\ssubset \ZZ^d$, $\Phi(X)$ is self-adjoint and supported in a ball of radius $r$ around $X$. We assume further that $\| \Phi(X) \|\le h$ for all $X\ssubset \ZZ^d$, and some constant $h<\infty$. The potential $\Phi$ is said to be a \textit{commuting potential} if for any $X,Z\ssubset \ZZ^d$, $[\Phi(X),\Phi(Y)]=0$.  Given such a local potential, the \textit{Hamiltonian} on a finite region $ \Lambda\ssubset \ZZ^d$ is defined as
\begin{align}
H_{\Lambda}=\sum_{X\subseteq \Lambda}\,\Phi(X)\,.
\end{align}
The Hamiltonian is called $(\kappa,R)$\textit{-local}, or simply \textit{geometrically local}, if  there exist parameters $\kappa,R> 0$ such that $\Phi(X)=0$ whenever the diameter $\operatorname{diam}(X)>R$ or $|X|>\kappa$. Moreover, the \textit{growth constant} $g$ of a geometrically-local Hamiltonian is defined such that
\begin{equation*}
    \left| \sum_{X:x_0 \in X} \Phi (X) \right| \leq g h \, ,
\end{equation*}
for all sites $x_0 \in \Lambda$. The Gibbs state corresponding to the region $A$ and at inverse temperature $\beta$ is defined as
\begin{align}
\sigma^A:=\frac{\e^{-\beta H_A}}{\tr[\e^{-\beta H_A}]}\,.
\end{align}
Note that this is in general not equal to the state $\tr_{A^c}[\sigma_\Lambda]$. Moreover, given $A\subset \Lambda$, we define the \textit{boundary} of $A$ by
\begin{equation*}
    \partial A := \{ x \in \Lambda \setminus A \, : \, \text{dist}(x,A) < \kappa \} \, ,
\end{equation*}
and denote by $A\partial$ the union of $A$ and its boundary. Note that $H_A$, and thus $\sigma^A$, have support on $A \partial$.

\subsection{Uniform families of Lindbladians}\label{Markov}

 We consider the basic model for the evolution of an open system in the Markovian regime given by a quantum Markov semigroup (or QMS) $(\cP_t)_{t\ge0}$ acting on the algebra $\cB(\cH)$ of bounded operators over a finite-dimensional Hilbert space $
 \cH$. Such a semigroup is characterised by the Lindbladian $\LL$, its generator, which is defined on $\cB(\cH)$ by $\cL(X)={\lim}_{t\to 0}\,\frac{1}{t}\,(\cP_t(X)-X)$ for all $X\in\cB(\cH)$. Recall that by the GKLS Theorem \cite{Lind,[GKS76]}, $\cL$ takes the following form: for all $X\in\cB(\cH)$,
\begin{equation}\label{eqlindblad}
\cL(X)=i[H,X]+\frac{1}{2}\sum_{k=1}^l{\left[2\,L_k^\dagger XL_k-\left(L_k^\dagger L_k\,X+X\,L_k^\dagger L_k\right)\right]}\,,
\end{equation}
where $H\in\cB(\cH)_{\operatorname{sa}}$, the sum runs over a finite number of \textit{Lindblad operators} $L_k\in\cB(\cH)$, and $[\cdot,\cdot]$ denotes the commutator defined as $[X,Y]:=XY-YX$, $\forall X,Y\in\cB(\cH)$. The QMS is said to be \textit{faithful} if it admits a full-rank invariant state $\sigma$, and when the state $\sigma$ is the unique invariant state, the semigroup is called \textit{primitive}. Moreover, we say that the Lindbladian is KMS-\textit{reversible} with respect to $\sigma$ if it is self-adjoint with respect to the KMS inner product defined in Section \ref{subsec:notations}. Whenever this condition holds, there exists a conditional expectation $E\equiv E_{\cF}$ onto the kernel of the generator
$\operatorname{Ker}(\cL):=\{X\in\cB(\cH):\,\cL(X)=0\}$, which by a slight abuse of notation we frequently call fixed-point subalgebra, and denote by $\cF(\cL)$, such that
\begin{align*}
\cP_t(X)\underset{t\to\infty}{\to}E[X]\,.
\end{align*}
Similarly, we say that the semigroup satisfies the \textit{detailed balance condition}, or is GNS-reversible with respect to $\sigma$ if it is self-adjoint with respect to the GNS inner product defined in  Section \ref{subsec:notations}. Under the assumption of GNS symmetry, it is also possible to write the generators in the following normal form:
\begin{theorem}[\cite{maas2011gradient} Theorem 3.1]\label{thm:normalformCM}Let $\cL$ be the generator of a quantum Markov semigroup on $\cB(\cH)$, where $\cH$ is an $N$-dimensional Hilbert space, with full-rank stationary state $\sigma$. Suppose that the generator $\cL $ is self-adjoint with respect to  $\langle \cdot,\cdot\rangle_{\sigma}^{\operatorname{GNS}}$. Then the generator $\LL$ has the following form: $\forall X\in \cB(\cH)$, 
					\begin{align}\label{LLDBC}
						\LL(X)&=\sum_{j\in \mathcal{J}}c_j\left( \e^{-\omega_j/2}\tilde{L}_j^*[X,\tilde{L}_j]+\e^{\omega_j/2}[\tilde{L}_j,X]\tilde{L}_j^*\right)\\
						&=\sum_{j\in \mathcal{J}}c_j\e^{-\omega_j/2}\left( \tilde{L}_j^*[X,\tilde{L}_j]+[\tilde{L}_j^*,X]\tilde{L}_j\right),
					\end{align}
					where $\cJ$ is a finite set of cardinality $|\cJ|\le N^2-1$, $\omega_j\in\RR$ and $c_j>0$ for all $j\in\mathcal{J}$, and $\{\tilde{L}_j\}_{j\in\mathcal{J}}$ is a set of operators in $\cB(\cH)$ with the properties:
					\begin{itemize}
						\item[1] $\{\tilde{L}_j\}_{j\in\mathcal{J}}=\{\tilde{L}_j^*\}_{j\in\mathcal{J}}$;
						\item[2] $\{\tilde{L}_j\}_{j\in\mathcal{J}}$ consists of eigenvectors of the modular operator $\Delta_\sigma$ with
						\begin{align}\label{eigenD}
							\Delta_\sigma(\tilde{L}_j)=\e^{-\omega_j}\tilde{L}_j;
						\end{align}
						\item[3] $\frac{1}{\dim(\cH)}\tr(\tilde{L}_j^*\tilde{L}_k)=\delta_{k,j}$ for all $j,k\in\mathcal{J}$;
						\item[4] $\tr(\tilde{L}_j)=0$ for all $j\in\mathcal{J}$.
						\end{itemize}

\end{theorem}	
This normal form will later be important to obtain bounds on the performance of noisy quantum annealers from quantum transport inequalities in Section~\ref{sec:applications}.

Next, we introduce the notion of a uniform family of quantum Markov semigroups defined on subregions of the lattice $\ZZ^d$. Our setup and notations are taken from \cite{cubitt2015stability}. 
\begin{definition}[Uniform family of Lindbladians]
Let $J\ge 0$ and $f:\NN\to \RR_+ $. Then, a family $\cL:=\{\cL_\Lambda,\cL_{\partial \Lambda}\}_{\Lambda\ssubset \ZZ^d}$ composed of \textit{bulk} $\cL_\Lambda$ and \textit{boundary} $\cL_{ \partial \Lambda}$ Lindbladians of strength $(J,f)$, both indexed on the finite subsets of $\ZZ^d$, is called \textit{uniform} whenever the following conditions hold:
\begin{itemize}
\item[(i)] \textit{Bulk Lindbladians:} For all $\Lambda\ssubset \ZZ^d$, $$\cL_\Lambda=\sum_{u\in\Lambda}\sum_{r\in\NN}\,\cL_{u,r}\,,~~~\text{ where }~~~ \supp(\cL_{u,r})=B_u(r)\,,$$ 
\end{itemize}
where $B_u(r)$ denotes the ball in $\ZZ^d$ centred at $u$ and of radius $r$. Moreover, 
\begin{align*}
J:=\sup_{u,r}\|\cL_{u,r*}\|_{1\to 1,\,\operatorname{cb}}<\infty ~~~\text{ and }~~~ f(r):=\sup_u\,\frac{\|\cL_{u,r*}\|_{1\to 1,\,\operatorname{cb}}}{J}\,,
\end{align*}
where $\|\Psi\|_{1\to 1,\operatorname{cb}}$ denotes the completely bounded $1\to 1$ norm of the superoperator $\Psi$. 
\item[(ii)] \textit{Boundary Lindbladian:} For all $\Lambda\ssubset \ZZ^d$, $$\cL_{\partial \Lambda}:=\sum_{k\in\NN}\,\cM_k\,,$$ 
where $$\|\cM_{k*}\|_{1\to 1,\operatorname{cb}}\le J\,|\partial_k^\text{in} \Lambda|\,f(k)$$
with 
$$\partial^\text{in}_k\Lambda:=\big\{  x\in\Lambda;\, \dist(x,\Lambda^c)\le k\big\}\,,~~~~~\operatorname{supp}(\cM_k)\subset \partial^\text{in}_k\Lambda\,.$$
The \textit{closed boundary} Lindbladians $\{\overline{\cL}^{\Lambda}\}_{\Lambda\ssubset\ZZ^d}$ are then defined as the sum of the bulk Lindbladians and of the boundary conditions:
\begin{align*}
\overline{\cL}_\Lambda:=\cL_\Lambda+\cL_{\partial \Lambda}\,.
\end{align*}
The uniform family $\cL$ is said to be \textit{$\kappa$-local}, $\kappa\in\NN$, if $f(r)=0$ for all $r > \kappa-1$. Moreover, $\cL$ is said to have a unique stationary state if there exists a family of quantum states $\{\sigma^{\Lambda}\}_{\Lambda\ssubset \ZZ^d}$ such that, for all $\Lambda\ssubset \ZZ^d$,  $\sigma^\Lambda$ is the unique stationary state of $\overline{\cL}_\Lambda$. Furthermore, the family $\cL$ is said to be \textit{primitive} if the states $\sigma^\Lambda$ are full-rank. $\cL$ is said to be \textit{locally reversible} if  $\overline{\cL}_\Lambda$ as well as $\cL_A$, for all $A\subseteq \Lambda$, are KMS-symmetric with respect to $\sigma^\Lambda$. Finally, $\cL$ is said to be \textit{frustration-free} if for all $A\subseteq B\ssubset \ZZ^d$, $\rho$ is a stationary state of $\cL_A$ whenever it is a stationary state of $\overline{\cL}_B$. In other words, we have $E_{A}\circ E_B=E_B\circ E_A=E_B$, where for a region $X\ssubset\ZZ^d$ we denote by $E_X:=\lim_{t\to\infty}\e^{t\cL_X}$ the conditional expectation onto the fixed-point subalgebra of $\cL_X$. 
\end{definition}

Given a primitive and reversible uniform family of Lindbladians $\cL$ and a finite region $A\ssubset \ZZ^d$, we decompose the fixed-point algebra $\cF(\cL_A)$ as
$$\cF(\cL_A):=\bigoplus_{i\in I_{\partial A}}\,\cB(\cH^A_i)\otimes \Id_{\cK^A_i}\,,~~~~~\text{ where }~~~~~\cH_\Lambda:=\bigoplus_{i\in I_{\partial A}}\cH^A_i\otimes \cK^A_i\,.$$ 
Then the conditional expectation $E_{A*}$ is expressed in the Schrödinger picture by
\begin{align}\label{Daviescond}
E_{A*}(\rho):=\lim_{t\to \infty}\e^{t\cL_{A*}}(\rho)\equiv \sum_{i\in I_{\partial A}}\tr_{\cK_i}\big[P^A_i  \rho P^A_i\big]\otimes \tau^A_i  \,.
\end{align}
Above, $\{P^A_i\}_{i\in I_A}$ are the central projections of $\cF(\cL_A)$, and $\tau_i^A$ are full-rank states supported on the space $\cK^A_i$. From now on, we often omit the dependence of the above spaces and algebras on the set $A$ for sake of simplicity, and only use the sum over the ``boundary conditions'' $i\in I_{\partial A}$ in order to remind the reader of the region being considered. When the states $\{\sigma^{\Lambda}\}_{\Lambda \ssubset \ZZ^d}$ are derived from a potential $\{\Phi(X)\}_{X\ssubset \ZZ^d}$, the maps $E_A$ and the family $\cL$ will be respectively referred to as the \textit{local specifications} and the \textit{quantum Gibbs sampler} corresponding to that potential. 

Assuming that the family $\cL$ is frustration-free, we have that for all $A\subset B\subset\Lambda \ssubset \ZZ^d$, the blocks $P^B_i\cB(\cH_\Lambda)P^B_i$ are preserved by the conditional expectation $E_A$. Moreover, on each of these blocks, $E_A$ only acts non-trivially on the factor $\cB(\cK^B_i)$, i. e. there exists a family of conditional expectations $\{E_{A}^{(i)}\in\cB(\cB(\cH_{\cK^B_i}))\}_{i\in I_{\partial B}}$ such that for each boundary condition $i\in I_{\partial B}$,  
\begin{align}\label{EAi}
E_A|_{P^B_i\cB(\cH_\Lambda)P^B_i}:=\id_{\cB(\cH^B_i)}\otimes E^{(i)}_{A}\,,~~~\text{ with }~~~
   E^{(i)}_{A*}(\rho):=\sum_{j\in I^{i}_{\partial A}}\,\tr(P^{i,A}_j\,\rho\,P^{i,A}_j)\otimes \tau^{i,A}_j\,. 
\end{align}
This is a consequence of \Cref{lemm1}.

\subsection{Examples of Gibbs samplers}\label{sec-examplessemigroups}

In this subsection, we introduce the Gibbs samplers which we will consider in the rest of the paper, what we call Schmidt generators and the more traditional embedded Glauber dynamics. Other works along similar lines~\cite{Temme2014,[BK16],bardet2020approximate} mostly consider the Davies and Heat-bath generators and we refer to~\cite{[SL78],majewski1995quantum,[BK16]} for more details on these other families of Gibbs samplers. These families of Gibbs samplers, like the Schmidt generators we will introduce,  enjoy many desirable properties. They are locally reversible and have $\sigma_{\Lambda}$ as their unique invariant state. However, the main reason we do not work with Davies or Heat-bath Gibbs samplers in this work is that the conditional expectation associated to them do not admit an explicit enough characterization, in contrast to the Schmidt generators we now introduce:

\paragraph{Schmidt generators with nearest neighbour interactions:}

In this section, we construct a more tractable family of conditional expectations, and a corresponding uniform family of Lindbladians, stabilizing the Gibbs state of a commuting Hamiltonian. This family is inspired by  a decomposition one can find in the proof of Lemma 8 in \cite{bravyi2003commutative} (see also \cite{johnson2017exact}). Here, we will restrict ourselves to nearest neighbour interactions, namely $2$-local interactions (i.e. Hamiltonians defined on graphs). Let $G=(V,E)$ be a graph, where each vertex $j\in V$ corresponds to a system with Hilbert space $\cH_j$, so that $\cH_V:=\bigotimes_{j\in V}\cH_j$, and define a Gibbs state on $\cH_V$ corresponding to the commuting Hamiltonian $H_V:=\sum_{(j,k)\in E}H_{jk}$, where $H_{jk}$ acts nontrivially on $\cH_j\otimes \cH_k$. Consider now the operator $e^{-H_{jk}}$ and perform a Schmidt decomposition over $\cH_j\otimes \cH_k$ ($\beta$ is taken to be equal to $1$ for sake of simplicity):
\begin{align*}
e^{-H_{jk}}=\sum_l X_j^l\otimes Y_k^l \, ,
\end{align*}
where $\{X_j^l\in\cB(\cH_j)\}$, resp. $\{Y_k^l\in\cB(\cH_{k})\}$, are independent. Given this decomposition for each interaction, we define the $C^*$-algebra $\mathcal{A}_{jk}\subset\cB(\cH_j)$ generated by the operators $X_j^l$. We now fix a vertex $j\in V$, and consider all the vertices $k_j$ such that $(j,k_j)\in E$, i.e. the neighborhood of $j$. We denote the neighborhood of $j$ by $V_j$. Since the vector spaces generated by $Y_{k}^l\otimes \Id_{k^c}$ and $Y_{k'}^l\otimes \Id_{k'^c}$, for $k\ne k'\in V_j$, are independent and $H_{jk},H_{jk'}$ commute, the algebras $\cA_{jk}$, $k\in V_j$, commute as well. Therefore, they can be jointly block decomposed:
\begin{align*}
    \cH_j:=\bigoplus_{\alpha_j}\,\bigotimes_{k\in V_j}\cH_{jk}^{\alpha_j}\otimes \cH_{jj}^{\alpha_j}\,,
\end{align*}
such that
\begin{align*}
    \cA_{jk}:=\bigoplus_{\alpha_j}\cB(\cH_{jk}^{\alpha_j})\otimes\bigotimes_{k'\in V_j\backslash \{k\}} \Id_{\cH_{jk'}^{\alpha_j}}\otimes \Id_{\cH_{jj}^{\alpha_j}}\,.
\end{align*}
According to this decomposition, the operator $\e^{-H_{jk}}$ for $(j,k)\in E$ on $\cH_{j}\otimes \cH_k$ can be decomposed as
\begin{align*}
    \e^{-H_{jk}}:=\bigoplus_{\alpha_j\alpha_k}(\e^{-H_{jk}})^{\alpha_j\alpha_k}\,,
\end{align*}
where each block $(\e^{-H_{jk}})^{\alpha_j\alpha_k}$ acts on $\cH^{\alpha_j}_{jk}\otimes \cH_{kj}^{\alpha_k}$. Therefore:
\begin{align}
   \sigma^V\simeq  \e^{-H_V}=\prod_{(j,k)\in E}\e^{-H_{jk}}=\bigoplus_{\alpha}\,\bigotimes_{(j,k)\in E}\,(\e^{-H_{jk}})^{\alpha_j\alpha_k}\,\otimes\bigotimes_{j'\in V}\Id_{\cH_{j'j'}^{\alpha_{j'}}}\,,\label{sigmavdecomp}
\end{align}
where the decomposition is over $\alpha=\{\alpha_k\}_{k\in V}$.

With these concepts at hand, we are now ready to define the conditional expectation $E_A^S$ corresponding to a subset of vertices $A\subset V$. 

First, given a set $A$, we define $\mathcal{A}_{A,\operatorname{out}}$ (see Figure \ref{tryschmidt}) to be

 \begin{align}
 \mathcal{A}_{A,\operatorname{out}}=\bigotimes_{j\in\partial A}\,\bigotimes_{\substack{k\in V\backslash A\\(j,k)\in E}}\,\bigoplus_{\alpha_j}\Id_{\cH_{jj}^{\alpha_j}}\,\otimes\,\cB(\cH_{jk}^{\alpha_j})\,
  \otimes
  \bigotimes_{\substack{k'\in A\\(j,k')\in E}} \Id_{ \cH_{jk'}^{\alpha_{j}}}.\label{equ:decomp}
 \end{align}
 	\begin{figure}[!ht]
	\centering
	\includegraphics[width=0.4\linewidth]{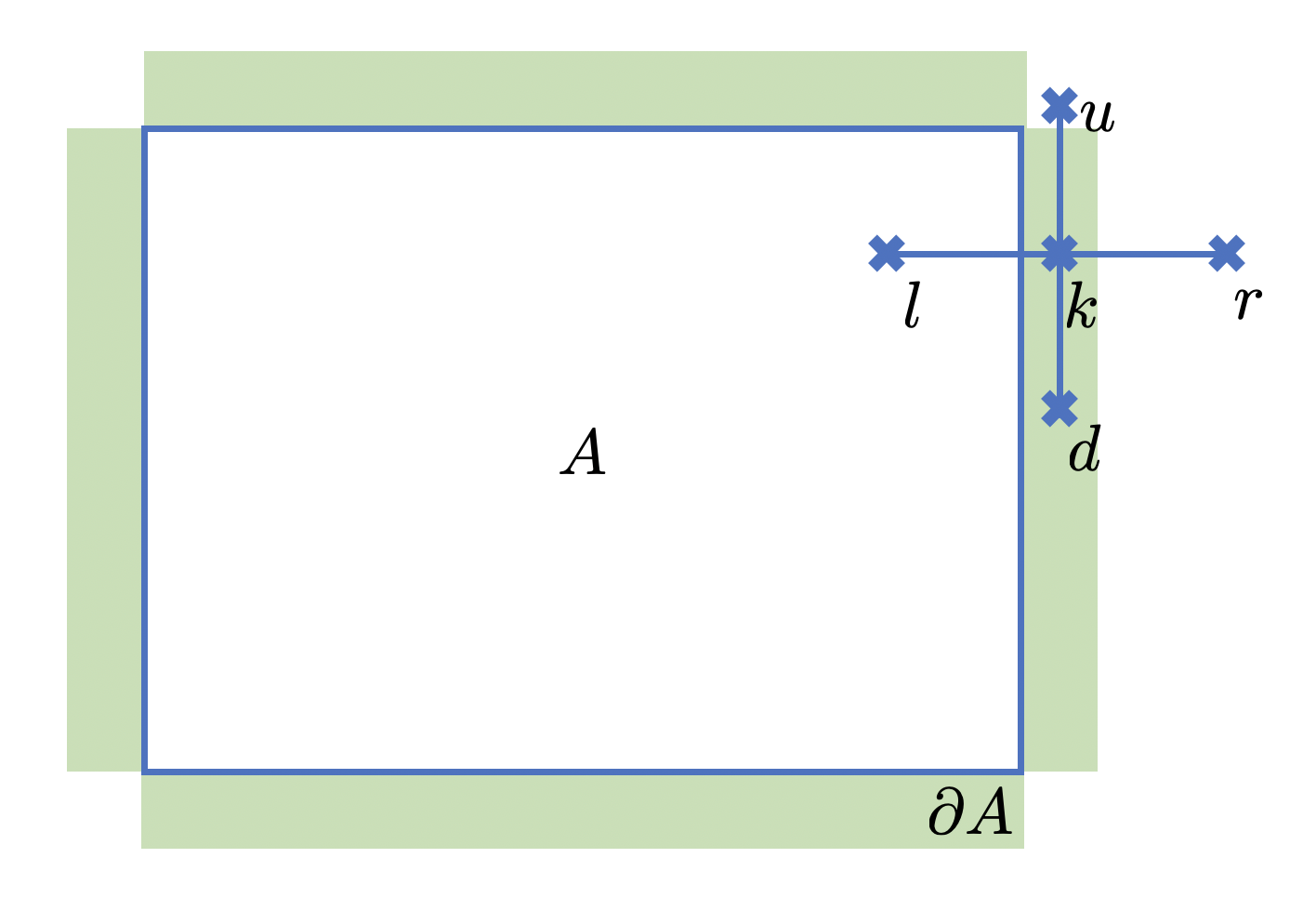}
	\caption{For the construction of $\mathcal{A}_{A, \text{out}}$, given a site $k$ in the boundary of $A$, the algebra acts trivially on the site $k$ itself, as well as on the edge $l-k$, whereas it acts non-trivially in the other three edges.}
	\label{tryschmidt}
	\end{figure}
Note that we are constructing this algebra over the boundary of $A$ by acting non-trivially only on the edges between a vertex of the boundary and a vertex in $V \setminus A$ (possibly in the boundary too).  We then finally define $E_A^{S}$ to be the so-called \textit{Schmidt conditional expectation} with respect to the Gibbs state $\sigma^{V}$ onto

\begin{align*}
    \Id_{\cH_A}\otimes \mathcal{A}_{A,\operatorname{out}}\otimes\,\bigotimes_{k\in V\backslash A\partial}\cB(\cH_k) \,.
\end{align*}
The existence and uniqueness of $E_A^S$ is guaranteed by Takesaki's theorem (see \Cref{propositioncondexp}) and the decomposition \eqref{sigmavdecomp} of $\sigma^V$. The main advantage of this $2$-local setting is that central projections of $\cA_{A,\operatorname{out}}$ are tensor products of $1$-local projections in each of the sites in $\partial A$. In analogy with the classical setting, we will call each vector $\alpha:=(\alpha_1,...,\alpha_{|\partial A|})$ defining a central projection in $\cA_{A,\operatorname{out}}$ a \textit{configuration}. 

We are now ready to define the \textit{Schmidt generator} as
\begin{equation*}
    \overline{\mathcal{L}}^S_V (X)  := \underset{k \in V}{\sum} E_{k}^S [X] - X \, .
\end{equation*}
Note that $\overline{\mathcal{L}}^S_V $ is such that $\sigma^V$ is the unique invariant state of $\overline{\mathcal{L}}^S_{V*} $, it is frustration-free and locally reversible. Thus, these generators still retain the desirable properties of the Heat-bath and Davies generators. However, in contrast to them, we see that the structure of the underlying
conditional expectations does not depend on system-bath couplings and is thus simpler to analyse.

 \begin{lemma}\label{lemmaSchmidt}
 	Let $\{\Phi(X)\}_{X\ssubset V}$ be a $2$-local commuting potential on a graph $G=(V,E)$. Then, the corresponding family $\cL^S= \{ \cL^S_A,\cL^S_{\partial A}\}_{A\ssubset V}$ of Schmidt generators introduced above satisfies the following: for all $A\subset V$,
 	\begin{align*}
 	\operatorname{Ker}(\cL^S_A)=\bigcap_{k\in A}\cF(E_k^S)=\cF(E_{A}^S)\,.
 	\end{align*}
 \end{lemma}
\begin{proof}
\begin{comment}
	\begin{figure}[!ht]
	\centering
	\includegraphics[width=0.5\linewidth]{tryschmidt.png}
	\caption{\textcolor{red}{Figure}}
	\label{tryschmidt}
\end{figure}

\end{comment}

	The first identity is a consequence of the following well-known symmetrization trick. For any $X\in\cB(\cH_{\Lambda})_{\operatorname{sa}}$, 
	\begin{align*}
	-\langle X,\,\cL_A^S(X)\rangle_{\sigma^\Lambda}=\sum_{k\in A}\,\langle X,\,X-E_k^S[X]\rangle_{\sigma^\Lambda}=\frac{1}{2}\,\sum_{k\in A}\langle X-E_k^S[X],\,X-E_k^S[X]\rangle_{\sigma^\Lambda}\,,
	\end{align*}
	where the last identity follows from the KMS-orthogonality of $E_k^S[X]$ and $X-E_k^S[X]$. From this, we directly see that $\cL_A^S(X)=0$ if and only if $X=E_k^S[X]$ for all $k\in A$. 
	
	The second identity can be shown invoking the decomposition given in~\eqref{equ:decomp}. We will proceed by induction on $|A|$. The claim is obvious for $|A|=1$. Assume it is true for all $|A|\leq m$ and consider $B=A\cup\{k\}$ for some $k\in V\backslash A$. By our induction hypothesis and the previous discussion, we have that $\operatorname{Ker}(\cL^S_{B})=\operatorname{Ker}(\cL^S_{A})\cap\operatorname{Ker}(\cL^S_{k})=\cF(E_{A}^S)\cap\cF(E_{k}^S)$. Let us now compare the two algebras $\cF(E_{A}^S)$ and $\cF(E_{k}^S)$. First, they clearly agree on $\left(A\partial\cup\{k\}\partial\right)^c$, as in that region both conditional expectations act trivially. In $\left(A\partial\cup\{k\}\partial\right)$, the elements of $\cF(E_{k}^S)$ will only act nontrivially on the Hilbert spaces $\cH_{ij}^{\alpha_i}$ with $j\not=k$  and $i\in\partial\{k\}$ such that $(i,j)\in E$. Similarly, the elements of $\cF(E_{A}^S)$ will only act nontrivially on $\cH_{i'j'}^{\alpha_{i'}}$ such that $i'\in\partial A$ and $j'\not\in A$. Thus, we conclude that the operators in the intersection of the two algebras will only act nontrivially on $\cH_{ij}^{\alpha_i}$ where $i\in\partial{A}\cup\partial{k}$ and $j\not\in\{k\}\cup A$. As these are exactly the Hilbert spaces in $\left(A\partial\cup\{k\}\partial\right)$ on which the elements of $\cF(E_{k\cup A}^S)$ act nontrivially, this concludes the proof.

\end{proof}

\paragraph{Embedded Glauber dynamics:}
 
  The situation becomes even simpler when $\sigma^\Lambda$ is diagonal in the computational basis: fix local bases $\{|\eta_x\rangle\}_{\eta=1}^{d_\cH}$ for $x\in\ZZ^d$, and denote the tensor products of the local basis elements by $\{|\eta^\Lambda\rangle\equiv \otimes_{x\in\Lambda}|\eta_x\rangle\}$, where $\eta^\Lambda:=(\eta_1,\dots,\eta_\Lambda)\in\{1,\dots, d_\cH\}^{\Lambda}$. Next, we assume the existence of a Gibbs measure $\mu^{\Lambda}$ on the configuration space $\Omega_\Lambda:=\{1,\dots d_\cH\}^{\Lambda}$ such that $$\sigma^\Lambda:=\sum_{\eta^\Lambda\in\Omega_\Lambda}\mu^\Lambda(\eta^\Lambda)\,|\eta^\Lambda\rangle\langle \eta^\Lambda|\,.$$

One can easily verify that the resulting Heat-bath dynamics leaves the computational basis invariant. Moreover, when restricted to that diagonal, it acts as the Glauber dynamics as defined for instance in Section 5.1 of \cite{Guionnet2003}: for any function $f:\Omega_\Lambda\to \RR$, and $A\subset \Lambda$,
\begin{align*}
L_A^{G}(f)(\eta^\Lambda)=\sum_{j\in A}\, \mathbb{E}_j^{\eta^{\{j\}^c}}\big[ f \big]-f(\eta^\Lambda)\,,
\end{align*}
where $\mathbb{E}^{\eta^{A^c}}_A$ is the classical conditional expectation associated with the local Gibbs measure $\mu^A$ conditioned on the boundary configuration $\eta^{A^c}$. The noncommutative conditional expectation $E^G_A$ takes then the form
\begin{align}\label{EAcondexp}
E^G_{A*}(\rho):=\sum_{\eta_{\partial A}\in\Omega_{\partial A}}\,\tr_A\big[\langle \eta_{\partial A}|\,\rho\,|\eta_{\partial A}\rangle\big]\otimes |\eta_{\partial A}\rangle\langle \eta_{\partial A}|\otimes \tau_A^{\eta_{\partial A}}\,.
\end{align}
Moreover, denoting by $\mathcal{C}_{A\partial }$ the local Pinching map onto the commutative algebra of local functions $f:\Omega_{A\partial }\to\CC$,  
\begin{align}\label{EAclassEA}
E_A^{G} = \mathbb{E}_A \circ \mathcal{C}_{A\partial }\,.
\end{align}

\subsection{Functional inequalities and rapid mixing}\label{FIs}

 Given the generator $\cL$ of a quantum Markov semigroup over $\cB(\cH)$ which we assume KMS-reversible with respect to an invariant state $\sigma$, the \textit{entropy production} of $\cL$ is defined for any other state $\rho\in\cD(\cH)$ by
\begin{align*}
\operatorname{EP}_{\cL}(\rho):=-\left.\frac{dD(\e^{t\cL_*}(\rho)\|E_*(\rho))}{dt}\right|_{t=0}\,.
\end{align*}
The entropy production is always non-negative, by the monotonicity of the relative entropy under quantum channels. Moreover, it satisfies the following useful property:
\begin{lemma}\label{EPexpress}
	Let $(\cP_t)_{t\ge 0}$ be a faithful, reversible quantum Markov semigroup of generator $\cL$. Then, for any full-rank invariant state $\omega$:
	\begin{align*}
	\operatorname{EP}_{\cL}(\rho)=-\tr\big[\cL_*(\rho)\big( \ln(\rho)-\ln(\omega)  \big)\big]\,.
	\end{align*} 
\end{lemma}
\begin{proof}
	This simply follows from the fact that the difference of the logarithms of the two invariant states $\ln(\omega)-\ln(E_*[\rho])$ belongs to the fixed point algebra $\cF(\cL)$ (see for instance the structure of invariant states in Equation (2.10) of \cite{bardet2018hypercontractivity}). Therefore
	\begin{align*}
	\tr\big[\cL_*(\rho)\big(\ln(E[\rho])-\ln\omega\big)\big]=\tr\big[\rho\cL\big(\ln(E[\rho])-\ln\omega\big)\big]=0\,.
	\end{align*}
\end{proof}

\begin{definition}[\cite{[KT13],CM15, BarEID17,gao2018fisher}]
	The quantum Markov semigroup $(\cP_t)_{t\ge 0}$ is said to satisfy a (non-primitive) \textit{modified logarithmic Sobolev inequality} (MLSI) if there exists a constant $\alpha>0$ such that, for all $\rho\in\cD(\cH)$:
	\begin{align}\tag{MLSI}\label{MLSI}
	4\alpha\, D(\rho\|E_*(\rho))\le \operatorname{EP}_\cL(\rho)\,.
	\end{align}
	The best constant satisfying \eqref{MLSI} is called the \textit{modified logarithmic Sobolev constant} and denoted by $\alpha(\cL)$. Moreover, the semigroup satisfies a \textit{complete modified logarithmic Sobolev inequality} (CMLSI) if, for any reference system $\cH_R$, the semigroup $(\e^{t\cL}\otimes \id_R)_{t\ge 0}$ satisfies a modified logarithmic Sobolev inequality with a constant $\alpha$ independent of $R$. In this case, the best constant satisfying CMLSI is called the \textit{complete modified logarithmic Sobolev constant} and is denoted by $\alpha_{\operatorname{c}}(\cL)$.
\end{definition}
	 The reason for the introduction of the complete modified logarithmic Sobolev constant is due to its tensorization property:
	 \begin{lemma}[\cite{gao2018fisher}]
	     Let $\cL$ and $\cK$ be two generators of $\operatorname{KMS}$-symmetric quantum Markov semigroups, and denote by $E_\cL$, resp. by $E_\cK$, their corresponding conditional expectations. Moreover, assume that $[E_\cL,E_\cK]=0$. Then
	     \begin{align*}
	        \alpha_{\operatorname{c}}(\cL+\cK)\ge \min\{ \alpha_{\operatorname{c}}(\cL),\,\alpha_{\operatorname{c}}(\cK)\}\,.
	     \end{align*}
	 \end{lemma}

By Gr\"{o}nwall's inequality, the (complete) modified logarithmic Sobolev inequality is directly related to the exponential convergence of the evolution towards its equilibrium, as measured in relative entropy:
\begin{align*}
	D\big(\e^{t\cL_*}(\rho)\big\|E_*(\rho) \big)\le \e^{-4\alpha (\mathcal{L}) t}\,D(\rho\|E_*(\rho))\,.
	\end{align*}

The problem of determining whether a quantum Markov semigroup satisfies a MLSI has been addressed in various settings in the last years. Some examples appear in \cite{muller2016relative, muller2016entropy}, where it was shown that the MLSI constant of the depolarizing channel can be lower bounded by $1/2$. This was subsequently extended to the generalized depolarizing channel in \cite{beigi2018quantum, capel2018quantum}. \cite{BardetCapelLuciaPerezGarciaRouze-HeatBath1DMLSI-2019} constitutes the first attempt to prove the inequality in the setting of spin systems, where the Heat-bath  generator in 1D was shown to satisfy a MLSI under two conditions of decay of correlations on the Gibbs state.
	
Similarly to the case for the MLSI, the spectral gap of $\cL$ provides a weaker notion of convergence with respect to the variance:
\begin{definition}[\cite{BarEID17}]
The \textit{spectral gap} of the semigroup $(\e^{t\cL})_{t\ge 0}$ with an invariant state $\sigma$ is given by the largest constant $\lambda$ which satisfies the following (non-primitive) \textit{Poincar\'{e} inequality}: for all $X\in\cB(\cH)_{\operatorname{sa}}$,
	\begin{align}\tag{PI}\label{poincare}
\lambda	\operatorname{Var}_E(X)\le \cE_{\cL}(X)\, ,
	\end{align}
	where the variance is given by $\operatorname{Var}_E (X):= \| X - E[ X]  \|^2_{\mathbb{L}_2(\sigma)}$ and $ \cE_{\cL}(X)$ is the \textit{Dirichlet form} of $X$:
\begin{align*}
\cE_{\cL}(X):=-\left.\frac{d}{dt}\right|_{t=0}\,\operatorname{Var}_E\big(\e^{t\cL}(X)\big) \, ,
\end{align*}
The spectral gap is denoted by $\lambda(\cL)$.
\end{definition}
From the previous definition we notice that a notion of complete spectral gap would be redundant, since it would provide the same information than the usual spectral gap. Moreover,  \eqref{poincare} is equivalent to the exponential decay of the variance: For all $X\in\cB(\cH)_{\operatorname{sa}}$,
	\begin{align*}
	\operatorname{Var}_E\big(\e^{t\cL}(X)\big)\le \e^{-\lambda(\mathcal{L}) t}\,\operatorname{Var}_E(X)\,.
	\end{align*}
	 The spectral gap of the Heat-bath and Davies generators was studied in \cite{[BK16]} and its positivity independently of the system size was proven to be equivalent to a strong form of clustering of correlations in the Gibbs state, which we discuss in \Cref{clustercorrela}. 
	
	Moreover, the positivity of the spectral gap can be used to show the existence of a CMLSI, as shown in~\cite{gao_spectral_2021}:
\begin{theorem}[Complete modified logarithmic Sobolev inequality~\cite{gao_spectral_2021},Theorem 4.3]\label{CMLSIholds}
	Let $\cL$ be a generator of a $\operatorname{GNS}$-symmetric quantum Markov semigroup acting on a $\operatorname{\dim}(\cH)$-dimensional Hilbert space with invariant state $\sigma>0$. Then $\cL$ satisfies:
\begin{align*}
\alpha_{\operatorname{c}}(\cL)\geq \frac{\lambda(\cL)\|\sigma^{-1}\|^{-1}}{(\operatorname{\dim}\cH)^2}\,.	
\end{align*}
\end{theorem}
For our purposes the result above will be helpful to ensure that terms of the generator acting on a bounded region always satisfy a CMLSI of constant order. The strict positivitity of the CMLSI was also proved in~\cite{gao_spectral_2021}, although without an explicit lower-bound expression.	

	In the classical locally finite setting a variant of the modified logarithmic Sobolev inequality -which predates it- is more naturally considered:

	\begin{definition}[\cite{[KT13]}]
Assume that the quantum Markov semigroup $(\cP_t)_{t\ge 0}$ is primitive with $\sigma$ as unique fixed point. It is said to satisfy a \textit{logarithmic Sobolev inequality} (LSI) if there exists a constant $\alpha_2>0$ such that, for all $\rho\in\cD(\cH)$:
	\begin{align}\tag{LSI}\label{LSI}
\alpha_2 \,	 D(\rho\| \sigma) \le \cE_{\cL}(\Gamma_\sigma^{-1/2}(\sqrt{\rho}))\, .	\end{align}
		The best constant $\alpha_2$ satisfying \eqref{LSI} is called the \textit{logarithmic Sobolev constant} and denoted by $\alpha_2(\cL)$.
\end{definition}

 The reason why the analysis of the LSI constant has attracted more attention in the classical literature is due to its connection to the useful property of hypercontractivity of the semigroup. Moreover, LSI implies MLSI, as least for GNS-symmetric semigroups:
	
\begin{proposition}[\cite{[KT13], CM15,BarEID17}]
Let $\cL$ be the generator of a primitive, $\operatorname{KMS}$-symmetric quantum Markov semigroup over a finite-dimensional Hilbert space. Then, \eqref{MLSI} $\Rightarrow$ \eqref{poincare}. Moreover, if the semigroup is $\operatorname{GNS}$-symmetric, then \eqref{LSI} $\Rightarrow $ \eqref{MLSI}.
\end{proposition}
In the case of locally unbounded classical evolutions, the existence of a positive uniform lower bound on the LSI constant of a family of generators is a strictly stronger condition than its analogue for the MLSI \cite{[D02]}. In the quantum case, the situation is even worse, since the existence of a complete LSI as introduced in \cite{[BK16a]}, or even of an LSI constant for a non-primitive semigroup, was proved to always fail \cite{bardet2018hypercontractivity}. This fact justifies the focus of the current article on the MLSI constant, since our proof heavily relies on the notion of a complete functional inequality.

Regarding the speed of convergence of a uniform family of quantum Markov evolutions to their corresponding equilibrium, we introduce the notion of a rapidly mixing uniform family of Lindbladians as follows:

\begin{definition}[Rapid mixing \cite{cubitt2015stability}]
	A primitive uniform family $\cL:=\{\cL_\Lambda,\cL_{\partial \Lambda}\}_{\Lambda\ssubset \ZZ^d}$ of Lindbladians with corresponding invariant states $\sigma^\Lambda$ is said to be \textit{rapidly mixing} if there exist positive constants $c,\gamma,\delta$ such that, for all $\Lambda\ssubset \ZZ^d$
\begin{align}\tag{RM}\label{RM}
\|\e^{t\cL_{\Lambda*}}(\rho)-\sigma^{\Lambda}\|_1\le  c\,\ln^\delta(\dim(\cH_\Lambda))\,\e^{-t\gamma}\,.
\end{align}
\end{definition}

This property has profound implications for the system, such as stability against external perturbations \cite{cubitt2015stability} and an area law in the mutual information for its fixed points \cite{brandao2015area}. To conclude, we recall that, by means of Pinsker's inequality, a positive uniform lower bound in the MLSI constant of a family of generators is a sufficient condition for a quantum system to satisfy rapid mixing. This serves as a motivation for our main result. From now on, we define the MLSI constant of a uniform family $\cL$ of Lindbladians as
$$\alpha (\cL) := \underset{\Lambda \nearrow \mathbb{Z}^d}{\text{lim inf}} \; \alpha (\mathcal{L}_\Lambda)\,.$$
We immediately have:
\begin{lemma}[\cite{[KT13]}]
Let $\cL:=\{\cL_\Lambda\}_{\Lambda\ssubset \ZZ^d}$ be a primitive uniform family of Lindbladians. If $\alpha(\cL)>0$, then $\cL$ is rapidly mixing.
\end{lemma}

\section{Clustering of correlations}\label{clustercorrela}
In this section we discuss various relevant notions of clustering of correlations, both for classical and quantum systems, and their relation to logarithmic Sobolev inequalities. Given the zoo of different notions of clustering present in the literature and the notation and language barriers arising from the different communities working on this subject, we start with a thorough review of the main concepts. But in~\Cref{clustering} we also show new connections between the recently introduced notion of analyticity after measurement~\cite{harrow2020classical} and strengthenings of  the standard notions of clustering. Furthermore, for the special case of the previously introduced Schmidt semigroup, we also derive a version of clustering needed in the proof of recent approximate tensorization results for the relative entropy~\cite{bardet2020approximate}.

\subsection{Dobrushin and Shlosman's mixing condition}

As mentioned in the introduction, the classical Glauber dynamics over a classical system is known to satisfy a logarithmic Sobolev inequality with constant independent of the lattice size if and only if correlations between two regions, as measured in the Gibbs equilibrium state, decay exponentially fast with the distance separating them. This general notion of \textit{decay of correlations} has many equivalent formulations in the classical setting. These were first put forward in the seventies with the ground-breaking works of Dobrushin and Shlosman.

Although refined results about unicity and mixing properties of Gibbs states are often model dependent, Dobrushin's original introduction \cite{dobrushin1974markov} of a widely applicable criterion, nowadays known as the \textit{Dobrushin uniqueness condition}, opened the door to the possibility of a global analysis of the equilibrium theory of spin systems. Given a potential $\Phi$, this criterion ensures the uniqueness of the Gibbs state in the thermodynamic limit whose local specifications on region $A$ given boundary condition $\omega_{A^c}$ correspond to $\omega_A\mapsto Z^{-1}{\e^{-\beta H(\omega_A,\omega_{A^c})}}$. This criterion was shown to hold at high enough temperature for a large class of models, including translation invariant, finite range interactions. Gross showed in \cite{gross1981absence} that Dobrushin's original condition implies that the mapping taking a potential to its associated Gibbs measure is twice differentiable. Later, Dobrushin and Shlosman \cite{Dobrushin1985,Dobrushin1985d} introduced a multi-site generalization of the Dobrushin uniqueness condition, known as \textit{Dobrushin-Shlosman uniqueness condition}, which also implies the uniqueness of the Gibbs state. However, none of these conditions imply the analytical dependence of the Gibbs measure to its corresponding potential. In their attempt to answer this problem, Dobrushin and Shlosman introduced twelve statements equivalent to analyticity, one of which being usually referred to as \textit{Dobrushin-Shlosman's mixing condition} \cite{Dobrushin1985,dobrushin1987completely}: There exists $\gamma\in (0,\infty)$ such that, for any $\emptyset \ne A\subset \Lambda\ssubset \ZZ^d$, there exists a constant $C(A)\in [0,\infty)$ such that, for any function $f$ supported in $A$ and $k\in \partial_r\Lambda$
	\begin{align}\label{DSM}\tag{DSM}
	\sup_{\substack{\omega,\eta\\\omega_j=\eta_j\forall j\ne k}}\big| \mathbb{E}^\eta_\Lambda[f]-\mathbb{E}^\omega_\Lambda[f]  \big|\le C(A)\,\vertiii f\,\e^{-\gamma \operatorname{dist}(k,A)}\,.
	\end{align}
	for some constant $C(A)$. Here, $\vertiii f :=\sum_{k\in \ZZ^d}\| f-\nu_k(f)\|$, where $\nu_k$ denotes the uniform measure at site $k$.

\subsection{Decay of correlations: the dynamical theory}

Since the 90's, Gibbs states have also attracted a lot of attention from the point of view of their dynamical properties \cite{liggett2012interacting}. Given a potential, one can construct a  Markov process, usually called \textit{Glauber dynamics}, whose reversing states coincide with the set of Gibbs states for the given potential (cf. \Cref{sec-examplessemigroups}). In particular, primitivity of the Glauber dynamics ensures the uniqueness of the Gibbs measure. In this case, Holley and Stroock \cite{holley1976,holley1989uniform} made the key observation that rapid uniform convergence of the evolution further ensures the decay of correlations at equilibrium. Their proof relied on finite propagation speed arguments. Roughly speaking, the probability that two distant regions correlate during a finite time interval is exponentially small in the distance separating them. Then, rapid convergence of the dynamics enables to transfer this property to the Gibbs state in the limit of large times.

The program of showing logarithmic Sobolev inequalities for non-trivial Gibbs measures was initiated by the work of Carlen and Stroock \cite{10.1007/BFb0075726} (see also \cite{deuschel1990hypercontractivity}). However, their techniques could only handle very special models at high temperatures. Later, Zegarlinski took a different approach \cite{ZEGARLINSKI199277,zegarlinski1990log,zegarlinski1990logb} with the goal of relating the existence of a logarithmic Sobolev inequality to the equilibrium theory of Dobrushin and Shlosman. This program was completed in a series of articles \cite{STROOCK1992299,stroock1992b,stroock1992c} where the authors showed the equivalence of the LSI and Dobrushin and Shlosman's mixing condition. Essentially, Stroock and Zegarlinski proved the equivalence between the following four notions:
\begin{itemize}
	\item[(i)] \textit{Dobrushin-Shlosman mixing condition}: Condition \eqref{DSM} holds for some constant $\gamma\in (0,\infty)$.
	\item[(ii)] \textit{Logarithmic Sobolev inequality}: the logarithmic Sobolev constant is lower bounded away from $0$ uniformly in any finite subset $\Lambda\ssubset \ZZ^d$ as well as in the boundary conditions $\omega_{\Lambda^c}$ chosen.
	\item[(iii)] \textit{Strong ergodicity}: There exist constants $\eps>0$ and $K<\infty$ such that for all function $f$ of the configurations, any subset $\Lambda\ssubset \ZZ^d$ and any boundary conditions $\omega_{\Lambda^c}$, 
	\begin{align}\label{SE}\tag{SE}
	\| \e^{t L_\Lambda^{\omega_{\Lambda^c}}}(f)-\mathbb{E}^{\omega_{\Lambda^c}}_{\Lambda}(f)\|\le K\,\vertiii f \,\e^{-\eps t},~~~ t\in(0,\infty)\,.
	\end{align}
	\item[(iv)] \textit{Spectral gap}: the spectral gap is lower bounded away from $0$ uniformly in the finite subset $\Lambda\ssubset \ZZ^d$ as well as in the boundary conditions $\omega_{\Lambda^c}$ chosen.
\end{itemize}
Moreover, since \eqref{DSM} typically holds above a threshold temperature, the above equivalence establishes a dynamical phase transition between low and high temperature regimes: indeed, the spectral gap estimate implies a mixing time in $\mathcal{O}(\operatorname{poly}(|\Lambda|))$ for fixed inverse temperature $\beta$. On the other hand, strong ergodicity implies a mixing time in $\mathcal{O}(\operatorname{polylog}(|\Lambda|))$, i.e. rapid mixing, since for $f$ supported on $\Lambda\ssubset \ZZ^d$,
\begin{align*}
\vertiii f=\sum_{k\in\ZZ^d}\|f-\nu_k(f)\|\overset{(1)}{=}\sum_{k\in\Lambda}\|f-\nu_k(f)\|\le 2 |\Lambda| \|f\|
\end{align*}
where $(1)$ simply follows from the fact that $f$ only acts on configurations in $\Lambda$, so that $f=\nu_k(f)$ for $k\in\Lambda^c$. Hence \eqref{SE}$\Rightarrow$\eqref{RM}. Moreover, it is a standard exercise to show that (iii)$\Rightarrow $(iv) by simply showing the stronger implication \eqref{RM}$\Rightarrow$(iv) via interpolation of $\mathbb{L}_p$ spaces (see e.g. Lemma 6 in \cite{temme2015fast}). This fact, which directly extends to quantum lattice spin systems, also establishes that the mixing time scales either logarithmically with system size, or at least exponentially.

Later, the assumption of the existence of any of the above statements uniformly for any finite set $\Lambda$ was relaxed to that for regular volumes (i.e. volumes which are unions of translations of a sufficiently large given cube) by Lu and Yau \cite{lu1993}, as well as Martinelli and Oliveri \cite{martinelli1994approach,martinelli1994approacha}. These weakened assumptions, referred to as \textit{strong mixing} in \cite{martinelli1994approach,martinelli1994approacha}, permitted to extend the domain of validity of logarithmic Sobolev inequalities to a larger class of potentials. For more information on this, we point the interested reader to the excellent lecture notes \cite{Guionnet2003,Martinelli1999}.

More recently, new and arguably simpler proofs of (i)$\Rightarrow $(ii) and (i)$\Rightarrow$ (iv) (or analogously their relaxations to regular volumes), based on an approximate tensorization of the variance and relative entropy, appeared in \cite{cesi2001quasi,bertini2002spectral,[D02]}. In \cite{cesi2001quasi} for instance, it is shown that \eqref{DSM} implies the existence of  constants $c$ and $\xi$ such that, for any two intersecting sets $C,D\ssubset \ZZ^d$, with $\Lambda:=C\cup D$, and all boundary condition $\omega_{\Lambda^c}$,
\begin{align}\tag{$\mathbb{L}_1\to\mathbb{L}_\infty$}\label{l1linfty}
\|\mathbb{E}_D\circ \mathbb{E}_C-\mathbb{E}^{\omega_{\Lambda^c}}_{\Lambda}[f]\|_{\mathbb{L}_\infty(\mu^{\Lambda,\omega_{\Lambda^c}})}\le c\,\e^{-\frac{\operatorname{dist}(\Lambda\backslash C,\Lambda\backslash D)}{\xi}}\,\|\mathbb{E}_C[f]\|_{\mathbb{L}_1(\mu^{\Lambda,\omega_{\Lambda^c}})}\,,
\end{align}
where $\mu^{\Lambda,\omega_{\Lambda^c}}$ denotes the Gibbs measure restricted to region $\Lambda$, with boundary condition $\omega_{\Lambda^c}$. This condition was then used to retrieve an approximate tensorization of the relative entropy, which in turn allows for an iterative procedure in order to prove the logarithmic Sobolev inequality by reduction to smaller regions.

\subsection{Quantum clustering of correlations}\label{clustering}

In the recent years, the classification of quantum lattice spin systems in and out of equilibrium has been the subject of active research within the community of mathematical physicists and that of quantum information theorists. In \cite{Kastoryano2013a}, a decay of correlations similar in spirit to \eqref{DSM} was found under the condition of positivity of the spectral gap independent of the system size, based on Lieb-Robinson bounds. Under the stronger assumption of a positive logarithmic Sobolev constant, \cite{Kastoryano2013a,brandao2015area} derived a stronger clustering in mutual information, leading to area laws implying an efficient classical approximate description as matrix product operators. Similar techniques were also used to prove the stability of rapidly mixing local quantum Markov semigroups against polynomially decaying error terms in the generator in \cite{cubitt2015stability}. More recently, the equivalence (i)$\Leftrightarrow$(iv) in the quantum setting was addressed by Kastoryano and Brand\~{a}o \cite{[BK16]}. There, the authors showed the equivalence between the positivity of the spectral gap independently of the lattice size and the following analogue of  \eqref{DSM} for frustration-free conditional expectations: for any $\Lambda \ssubset \ZZ^d$ a family of local specifications $\{E_A\}_{A\subset \Lambda}$ corresponding to the Gibbs state $\sigma^{\Lambda}\in \cD(\cH_\Lambda)$ satisfies a \textit{strong $\mathbb{L}_2$ clustering of correlations} if for any $A,B\subseteq\Lambda$ with $A\cap B\ne\emptyset$ and $A\cup B=\Lambda$, there exist constants $K,\gamma>0$ such that for any observable $X\in\cB(\cH_\Lambda)_{\operatorname{sa}}$:
\begin{align}\label{L2L2}\tag{$\operatorname{sq}\mathbb{L}_2$}
\|E_{B}\circ  E_A[X]-E_{\Lambda}[X]\|_{\mathbb{L}_2(\sigma^\Lambda)}\le K\, \e^{-\gamma \operatorname{dist}(B\backslash A, A\backslash B)}\,\|X\|_{\mathbb{L}_2(\sigma^{\Lambda})}\,.
 \end{align} 
Simple equivalence of $\mathbb{L}_p$ norms arguments can be used to show that for classical systems, the condition of strong $\mathbb{L}_2$ clustering is implied by \eqref{l1linfty}. Surprisingly enough, the equivalence between (ii) and (iv) above also provides the opposite implication. The direction \eqref{L2L2} implies spectral gap was shown by extending the classical proof of \cite{bertini2002spectral}, whereas the opposite implication is a consequence of the detectability lemma (see \cite{Aharonov_2009}). In \cite{[BK16]}, \eqref{L2L2} was also shown to hold for one dimensional systems, and for any lattice system at high enough temperature. Let us however stress the importance of not considering a Hamiltonian interaction part in the generator of the semigroup: previous work found examples of semigroups with vanishing gap even at infinite temperature due to the presence of internal interactions \cite{cai2013algebraic}. 

Before moving to the definition of mixing that we will use to derive our main result, let us briefly mention some interesting related work on the fast preparation of quantum Gibbs states: In \cite{majewski1995quantum,majewski1996quantum}, Majewski and Zegarlinski found similar conditions as those of Dobrushin and Shlosman under which the quantum Heat-bath generator is rapidly mixing. These conditions were typically shown to hold at high temperature by Kastoryano and Temme in \cite{temme2015fast}. More recently, rigorous connections between the analyticity of the partition function of the Gibbs state, its estimation by means of a classical algorithm, and decay of correlations, were found in \cite{harrow2020classical}. These results can be interpreted as the first quantum extensions of the seminal work of Dobrushin and Shlosman beyond the 1D case \cite{araki1969gibbs} or the high temperature regime \cite{kliesch2014locality}.

Brand\~{a}o and Kastoryano \cite{brandao2019finite} derived an efficient quantum dissipative algorithm for the preparation of quantum Gibbs states of a possibly non-commuting potential satisfying a uniform approximate Markov property, under a condition of \textit{uniform clustering of correlations} (see \Cref{clustering}). These two conditions were shown to hold at high enough temperature \cite{kuwahara2020clustering,kliesch2014locality}, hence proving the existence of efficient Gibbs samplers in that regime.

Although the algorithm of \cite{brandao2019finite} has constant depth, it employs log-size gates. On the other hand, as rightfully pointed by the authors of that paper, proving the logarithmic Sobolev inequality for a local Gibbs sampler would provide an algorithm which would converge with time scaling logarithmically with the system size, with local Lindblad operators. Therefore,  our main result can also be turned into an algorithm that efficiently prepares the Gibbs state of a commuting potential with local channels only and logarithmic depth.

In this section, we extend some of the equivalent  mixing conditions of  Dobrushin and Shlosman to the quantum realm. In particular, we show that the following notions of quantum clustering all follow from the recently introduced notion of \textit{analyticity after measurement} \cite{harrow2020classical}. Similar statements can be found in \cite{kliesch2014locality,[BK16],harrow2020classical}. In what follows, given the Gibbs state $\sigma^\Lambda$ with ``open boundary conditions'' over region $\Lambda$ and a test $0\le P_A\le \Id$ supported on sub-region $A$, we denote by 
\begin{align*}
\sigma^{\Lambda, P_A}:=\frac{\sqrt{P_A}\sigma^{\Lambda}\sqrt{P_A}}{\tr\big[ P_A\sigma^\Lambda\big]}\,
\end{align*}
the post-selected state after test $P_A$ has occurred. This definition generalizes the concept of a closed boundary condition to the quantum setting. Then,
\begin{definition}[Clustering of correlations]
A potential $\{\Phi(X)\}_{X\ssubset \ZZ^d}$ satisfies  
\begin{itemize}
	\item[(i)] the \textit{uniform $\mathbb{L}_\infty$ clustering of correlations} if there exist constants $c\ge 0$ and $\xi>0$ such that, for any $A,B\subset \Gamma\subset  \Lambda$, and all $X_A\in\cB(\cH_{A})_{\operatorname{sa}}$, $Y_B\in\cB(\cH_{B})_{\operatorname{sa}}$ and tensor product of local tests $P_{\partial \Gamma}\in \cB(\cH_{\partial \Gamma})_+$:
\begin{align}\tag{q$\mathbb{L}_\infty$}\label{Linfty}
\operatorname{Cov}_{\sigma^{\Lambda,P_{\partial\Gamma}}}(X_A,Y_B),\,\operatorname{Cov}^{(0)}_{\sigma^{\Lambda,P_{\partial \Gamma}}}(X_A,Y_B)\le c\,|\Gamma|\,\|X_A\|_{\infty}\,\|Y_B\|_{\infty}\,\e^{-\frac{\dist(A ,B)}{\xi}}\,.
\end{align}
\item[(ii)] the \textit{uniform $\mathbb{L}_2$ clustering of correlations} if there exist constants $c\ge 0$ and $\xi>0$ such that, for any $A,B\subset \Gamma\subset \Lambda$, and all $X_A\in\cB(\cH_{A})_{\operatorname{sa}}$, $Y_B\in\cB(\cH_{B})_{\operatorname{sa}}$ and tensor product of local tests $P_{\partial \Gamma}\in \cB(\cH_{\partial \Gamma})_+$:
\begin{align}\tag{q$\mathbb{L}_2$}\label{qL2}
\operatorname{Cov}_{\sigma^{\Lambda.P_{\partial\Gamma}}}(X_A,Y_B)\le c\,|\Gamma|\,\|X_A\|_{\mathbb{L}_2(\sigma^{\Lambda,P_{\partial \Gamma}})}\,\|Y_B\|_{\mathbb{L}_2(\sigma^{\Lambda, P_{\partial \Gamma}})}\,\e^{-\frac{\dist(A,B)}{\xi}}\,.
\end{align}
\item[(iii)] the \textit{uniform $\mathbb{L}_2^{(0)}$ clustering of correlations} if there exist constants $c\ge 0$ and $\xi>0$ such that, for any $A,B\subset \Lambda$, and all $X_A\in\cB(\cH_{A})_{\operatorname{sa}}$, $Y_B\in\cB(\cH_{B})_{\operatorname{sa}}$ and tensor product of local tests $P_{\partial \Gamma}\in \cB(\cH_{\partial \Gamma})_+$:
\begin{align}\tag{q$\mathbb{L}_2^{(0)}$}\label{qL20}
\operatorname{Cov}^{(0)}_{\sigma^{\Lambda,P_{\partial\Gamma }}}(X_A,Y_B)\le c\,|\Gamma|\,\Big(\tr\big[\sigma^{\Lambda,P_{\partial\Gamma}} X_A^\dagger X_A\big]\Big)^{\frac{1}{2}}\,\Big(\tr\big[\sigma^{\Lambda,P_{\partial\Gamma}} Y_B^\dagger Y_B\big]\Big)^{\frac{1}{2}} \,\e^{-\frac{\dist(A,B)}{\xi}}\,.
\end{align}
\item[(iv)] the \textit{quantum Dobrushin-Shlosman condition} (qIIId) if there exist constants $c\ge 0$ and $\xi>0$ such that, for any $A,B\subset\Gamma\subset  \Lambda$, any $N_A\in\cB(\cH_A)_+$, $\|N_A\|_{\mathbb{L}_1(\sigma^{\Lambda,P_{\partial\Gamma}})}\le 1$, $P_B,P_B'\in\cB(\cH_D)_+$ and tensor product of local tests $P_{\partial \Gamma}\in \cB(\cH_{\partial \Gamma})_+$, $0\le P_B,P_B',P_{\partial\Gamma}\le \Id$:
\begin{align}\label{qIIId}\tag{qIIId}
\big|  \tr[\sigma^{\Lambda, P_BP_{\partial\Gamma}}N_A]  -\tr[\sigma^{\Lambda,P_B^{\prime}P_{\partial\Gamma}}N_A]\,\big|\le  c\,|\Gamma|\,\e^{-\frac{\dist(A,B)}{\xi}}\tr\big[  \sigma^{\Lambda,P_B^{\prime}P_{\partial\Gamma}}N_A \big]\,,
\end{align}

In particular, taking the supremum over tests $0\le N_A\le\Id$, we have the following \textit{local indistinguishability}:
\begin{align}\tag{$\operatorname{qIIIc}$}
\big\|  \tr_{A^C}(\sigma^{\Lambda, P_BP_{\partial\Gamma}})  -\tr_{A^C}(\sigma^{\Lambda,P_B^{\prime}P_{\partial\Gamma}})\,\big\|_1\le  c\,|\Gamma|\,\e^{-\frac{\dist(A,B)}{\xi}}\,.
\end{align}

\end{itemize}

\end{definition}

\begin{remark}\label{remarkboundarycond}
Observe that our quantum Dobrushin-Shlosman conditions are slightly stronger than the ones enunciated for instance in \cite{dobrushin1987completely}, since the ``boundary conditions'' $P_D$ and $P_D'$ are allowed to differ on more than one site when $|D|>1$. 
\end{remark}

 \begin{figure}[h!]
	\centering
	\includegraphics[width=1\linewidth]{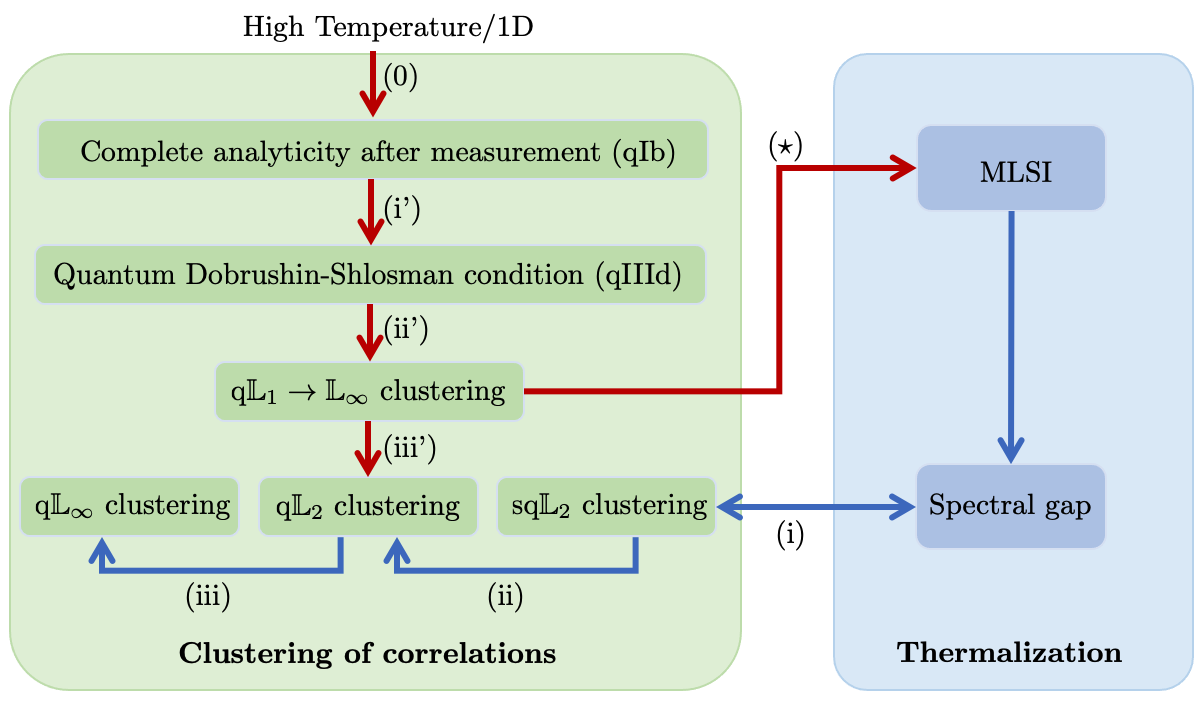}
	\caption{Relations between different notions of clustering of correlations and their link to thermalization times. (0) is proved in Theorem \ref{thm:analyticity-high-temperature} in the high temperature regime. (i) is the main result of \cite{[BK16]}. (i') and (iii) are proved in Theorem \ref{thm:differentclusterings} whereas (ii) is proved in \cite{[BK16]}. (ii') is proved in Proposition \ref{propDScorrdecay}. Finally ($\star$) is the subject of Section \ref{sec:main}.   } 
	\label{figurezoo}
\end{figure}

Uniform $\mathbb{L}_\infty$ clustering is the condition usually considered in the literature \cite{kliesch2014locality,brandao2019finite,harrow2020classical}, whereas the strong $\mathbb{L}_2$ clustering \eqref{L2L2} was shown in \cite{[BK16]} to be equivalent to the above uniform $\mathbb{L}_2$ clustering when $P_{\partial\Gamma}=\Id$ for commuting 1D Hamiltonians. While the fact that (IIId)$\Rightarrow$\eqref{Linfty} follows relatively easily, Dobrushin and Shlosman could also prove the opposite direction for translation invariant interactions by the introduction of an equivalent condition of \textit{complete analyticity} (conditions (Ia)-(Ic) in \cite{dobrushin1987completely}). Recently, the condition of uniform $\mathbb{L}_\infty$ clustering was shown to be a consequence of the following quantum generalization of Dobrushin and Shlosman's complete analyticity condition in \cite{harrow2020classical}, where the authors also extended the reverse direction to the case of possibly non-translation invariant classical interactions. 
	\begin{definition}[Analyticity after measurement, see Condition 1 in \cite{harrow2020classical}]\label{def:analyticity_after_measurement}  Given a geometrically-local Hamiltonian $H$, its free energy is said to be $\delta$\textit{-analytic} for all $\beta\in [0,\beta_c)$ if it is analytic in the open ball of radius $\delta$ around $\beta$ and if there exists a constant $c$ such that, for any operator $N\ge 0$ with $\|N\|_\infty=1$,
		\begin{align}\tag{$\operatorname{qIb}$}\label{qCA}
		\Big|\ln\Big(\tr\Big[\e^{-\sum_{X\subseteq \Lambda}z_X \Phi(X)}N\Big]\Big)\Big|\le c\,|\Lambda|
		~~~~\forall z_X\in\CC\,, |z_X-\beta|\le \delta \,.
		\end{align}
	\end{definition}
	
	This property can be shown to hold above a critical temperature, with a similar proof to the analogous fact for the partition function  in \cite{harrow2020classical}. We leave the proof of the following result to Appendix \ref{ap:Analyticity}.

	\begin{theorem}\label{thm:analyticity-high-temperature}
Let $H$ be a geometrically-local Hamiltonian with range $\kappa$, growth constant $g$ and local interactions with norm at most $h$.  Given $0< \delta <\frac{1}{5 \operatorname{e} g h \kappa }$, we denote $\beta_c=\frac{1}{5 \operatorname{e} g h \kappa } - \delta$. Then, for all $\beta \in [0,\beta_c)$ and $N \geq 0$ with $\|N\|_\infty=1$, the function $\mathbf{z}\mapsto  \ln\left( \tr \left[\operatorname{e}^{- \sum_{X \subset \Lambda} z_X \Phi(X)} N \right] \right)$ is analytic and  bounded in modulus by $( \operatorname{e}^2 g h (\beta + \delta) + \ln(d) ) \abs{\Lambda}$. 
\end{theorem}
	
	Furthermore, the proof of Theorem 31 of \cite{harrow2020classical} can be readily adapted to prove the implication \eqref{qCA}$\Rightarrow$\eqref{qIIId}, which is a part of the following result.

\begin{theorem}\label{thm:differentclusterings}
	For a given local commuting potential $\{\Phi(X)\}_{X\ssubset\ZZ^d}$, the following chain of implications holds:
	\begin{align*}
\eqref{qCA}\Longrightarrow	\eqref{qIIId}\Longrightarrow\eqref{qL2}\Longleftrightarrow\eqref{qL20}\Longrightarrow \eqref{Linfty}\,. 
	\end{align*}
In the case of a local but non-commuting potential, the same conclusion can be reached after replacing $|\Gamma|$ by $|\Lambda|$ in the upper bounds in the definitions of \eqref{qIIId}, \eqref{qL2}, \eqref{qL20} and \eqref{Linfty}.	Finally, the reverse implication \eqref{Linfty}$\Rightarrow$\eqref{qCA} holds in the case of a classical potential.
\end{theorem}
Before proving the above theorem, we recall the crucial Theorem 4.1 from \cite{Dobrushin1985} (see also Lemma 28 in \cite{harrow2020classical}):
\begin{lemma}\label{blackmagic}
    Let $f:\CC^m\to\CC$ be analytic on a connected open set $\Omega\subset \CC^m$ such that $|f(\mathbf{z})|\le M$ for all $\mathbf{z}\in\Omega$. Let moreover $k_1,...,k_m$ be non-negative integers summing up to $K$, and suppose that there exists $\mathbf{z}_0\in \Omega$ such that $f(\mathbf{z}_0)=0$ as well as the partial derivatives $\frac{\partial^K}{\partial^{k_1}z_1 ...\partial^{k_m}z_m}f(\mathbf{z}_0)=0$ unless we take derivatives with respect to at least $L$ distinct variables $z_i$. Then, for all $\mathbf{z}\in\Omega$, there exist $c_1\equiv c_1(\mathbf{z}), \, c_2\equiv c_2(\mathbf{z})>0$ such that $|f(\mathbf{z})|\le M c_1\,\e^{-c_2L}$.
\end{lemma}
Let us emphasize that the Lemma above ensures that the constants $c_1$ and $c_2$ only depend on $f$ through its domain of analiticity and on $\mathbf{z}$. Moreover, a close inspection of the results of~\cite{Dobrushin1985,harrow2020classical}  shows that for the special case in which we take $\Omega$ to be the region where~\eqref{qCA} holds, then they do not depend on $m$, only on $\|\mathbf{z}\|_{\infty}$.
\begin{proof}[Proof of \Cref{thm:differentclusterings}]

\eqref{qCA}$\Rightarrow$\eqref{qIIId}: follows from a refinement of the proof of Theorem 31 in \cite{harrow2020classical}: First, we observe that \eqref{qCA} holds for any local, non-zero, positive semidefinite operator $N$ such that $\|N\|_{\mathbb{L}_1(\sigma^{\Lambda,P_{\partial\Gamma}})}\le 1$. Indeed, for all $z_X\in\CC,\,|z_X-\beta|\le \delta$: 
	\begin{align*}
\Big|\ln\tr\Big[\e^{-\sum_{X\subset \Lambda}z_X \Phi(X)}N\Big]\Big|&\le \Big|\ln\tr\Big[\e^{-\sum_{X\subset \Lambda}z_X \Phi(X)}\frac{N}{\|N\|_\infty}\Big]\Big|+\big|\ln(\|N\|_\infty)\big|\\
&\overset{(1)}{\le} c\,|\Lambda|+|\ln\|(\sigma^{\Lambda,P_{\partial\Gamma}})^{-1}\|_\infty|\\
&\overset{(2)}{\le} \big(c+\ln(d_{\cH})\big)\,|\Lambda|\\
&\equiv c'\,|\Lambda|\,.
\end{align*}
In $(1)$ above, we used the equivalence relation between $\mathbb{L}_1$ and $\mathbb{L}_\infty$ norms, whereas $(2)$ comes from the following bounds:
\begin{align*}
    1\le \frac{\tr\big[\sqrt{P_{\partial\Gamma}}\sigma^{\Lambda}\sqrt{P_{\partial\Gamma}}\big]}{\|\sqrt{P_{\partial\Gamma}}\sigma^{\Lambda}\sqrt{P_{\partial\Gamma}}\|_\infty}\le d_{\cH}^{|\Lambda|}\,, 
\end{align*}
where the last inequality simply follows from the equivalence constant between the Schatten $1$ and $\infty$ norms. Therefore,  $|\ln\|(\sigma^{\Lambda,P_{\partial\Gamma}})^{-1}\|_\infty|\le |\Lambda|\,\ln(d_{\cH})$. The rest of the proof follows very similarly to that of Equation (57) in \cite{harrow2020classical}: Let us define  the complex perturbed Gibbs state as
\begin{align*}
    \sigma^{\Lambda,P}_{\mathbf{z}}:=\frac{\sqrt{P}\e^{-\sum_{X\subseteq\Lambda}z_X \Phi(X)}\sqrt{P}}{\tr\big[P\, \e^{-\sum_{X\subseteq\Lambda}z_X \Phi(X)}\big]}\,.
\end{align*}
Then, we need to prove that 
\begin{align}\label{toprove57}
|f(\textbf{z})|:= \left|    \ln\left( \frac{\tr\big[\sigma_{\textbf{z}}^{\Lambda,P_{\partial\Gamma} P_B}N_A\big]}{\tr\big[\sigma_{\textbf{z}}^{\Lambda,P_{\partial\Gamma} P_B'}N_A\big]}  \right)\right|\le c\,\e^{-\dist(A,B)/\xi}\,
\end{align}
for any $N_A, P_B$ and $P_{B}'$ as in the statement of \eqref{qIIId}.  \Cref{toprove57} will directly follow after showing that the function $f$ of the complex vector $\mathbf{z}$ satisfies the requirements of \Cref{blackmagic} with $L$ proportional to $\dist(A,B)$. First of all, we have by \eqref{qCA} that $f$ is a sum of analytic functions, and therefore is analytic itself. Moreover, denoting $$e_{\mathbf{z}}^{\Delta}:=\exp\left(-\sum_{\substack{X\subseteq\Lambda\\X\cap \Delta\ne 0}}z_X\Phi(X)\right)\,,$$
 we have that, for all $\mathbf{z}$ satisfying the condition \eqref{qCA}: 
\begin{align}\nonumber
    |f(\mathbf{z})|&\le \Big|\ln\tr\big[e_{\mathbf{z}}^\Lambda P_{\partial\Gamma} {P}_B {N}_A \big]\Big|+\Big|\ln\tr\big[e_{\mathbf{z}}^\Lambda\,P_{\partial\Gamma }{P}_B' \big]\Big|+\Big|\ln\tr\big[e^\Lambda_{\mathbf{z}}\,P_{\partial \Gamma}{P}_B \big]\Big|+\Big|\ln\tr\big[e^\Lambda_{\mathbf{z}}\, {P}_B' P_{\partial\Gamma}{N}_A\big]\Big|\\
    &\le 4c'|\Lambda|\label{boundnoncommuting}\,.
\end{align}
We are left with proving that all the derivatives of $f$ at $\mathbf{z}_0=0$ involving less than $\dist(A,B)$ distinct variables $z_X$ vanish. For this, we denote by $G$ the region $G:=\bigcup_{X\subseteq \Lambda,\,k_X\ge 1}X$, where $k_{X}$ denotes the degree of the partial derivative with respect to the variable $z_X$. That is, $G$ is the union of the support of terms we are taking derivatives of.

In general, $G=\bigcup G_i$ is a disjoint union of connected components $G_i$. We first consider the case where there is no connected path connecting $A$ to $B$ through unions of $G_i$'s or sites in $\partial\Gamma$. In this situation, the boundary $\partial\Gamma$ can be partitioned into sites $\partial\Gamma_A$ which are connected to $A$ through a union of regions $G_{A-\Gamma}:=\cup G_i$ or other sites in $\partial\Gamma$, those $\partial\Gamma_B$ connected to $B$ through another union of regions $G_{B-\Gamma}=\cup_j G_j$, disconnected from $G_{A-\Gamma}$, or other sites in $\partial\Gamma$, and the remaining sites $\partial\Gamma_C$ which are neither connected to $A$, nor to $B$. Finally, we define $G_A$, resp. $G_B$, as the union of  $G_{A-\Gamma}$ and regions $G_k$ intersecting $A$, resp. that of $G_{G-\Gamma}$ and regions $G_{k'}$ intersecting $B$. The union of the remaining regions constituting $G$ which either intersect $\partial\Gamma_C$, or do not intersect $A\cup B\cup \partial\Gamma_C$, is denoted by $G_C$. Then,
\begin{align*}
  \left.  \frac{\partial^K}{\prod_{X\subseteq \Lambda}\partial^{k_X}z_X}\right|_{\mathbf{z}=0}
  \,f(\mathbf{z})
 &=  \left.  \frac{\partial^K}{\prod_{X\subseteq \Lambda}\partial^{k_X}z_X}\right|_{\mathbf{z}=0} \left( \ln\tr\big[e_{\mathbf{z}}^G P_{\partial\Gamma} {P}_B {N}_A \big]+\ln\tr\big[e_{\mathbf{z}}^G\,P_{\partial\Gamma }{P}_B' \big] \right.\\
 & \phantom{sadasdasdsadasdasda}\left. -\ln\tr\big[e^G_{\mathbf{z}}\,P_{\partial \Gamma}{P}_B \big]-\ln\tr\big[e^G_{\mathbf{z}}\, {P}_B' P_{\partial\Gamma}{N}_A\big]\right)\\
  &= \left.  \frac{\partial^K}{\prod_{X\subseteq \Lambda}\partial^{k_X}z_X}\right|_{\mathbf{z}=0}\ln\Big(\tr\big[ \e_{\mathbf{z}}^{G_A}N_AP_{\partial\Gamma_A}\big]\tr\big[\e_{\mathbf{z}}^{G_B}P_BP_{\partial\Gamma_B} \big]\tr\big[\e_{\mathbf{z}}^{G_C}P_{\partial\Gamma_C} \big]\Big)\\
  &\; \; \; + \left.  \frac{\partial^K}{\prod_{X\subseteq \Lambda}\partial^{k_X}z_X}\right|_{\mathbf{z}=0}\ln\Big(\tr\big[ \e_{\mathbf{z}}^{G_A}P_{\partial \Gamma_A}\big]\tr\big[\e_{\mathbf{z}}^{G_B}P'_BP_{\partial\Gamma_B} \big]\tr\big[\e_{\mathbf{z}}^{G_C}P_{\partial\Gamma_C} \big]\Big)\\
  & \; \; \;  - \left.  \frac{\partial^K}{\prod_{X\subseteq \Lambda}\partial^{k_X}z_X}\right|_{\mathbf{z}=0}\ln\Big(\tr\big[ \e_{\mathbf{z}}^{G_A}P_{\partial \Gamma_A}\big]\tr\big[\e_{\mathbf{z}}^{G_B}P_BP_{\partial\Gamma_B} \big]\tr\big[\e_{\mathbf{z}}^{G_C}P_{\partial\Gamma_C} \big]\Big)\\
 & \; \; \;  - \left.  \frac{\partial^K}{\prod_{X\subseteq \Lambda}\partial^{k_X}z_X}\right|_{\mathbf{z}=0}\ln\Big(\tr\big[ \e_{\mathbf{z}}^{G_A}N_AP_{\partial \Gamma_A}\big]\tr\big[\e_{\mathbf{z}}^{G_B}P'_BP_{\partial\Gamma_B} \big]\tr\big[\e_{\mathbf{z}}^{G_C}P_{\partial\Gamma_C} \big]\Big)\\
& =0\,.
\end{align*}
 It remains to consider the case when regions $A$ and $B$ can be connected through a path constituted of regions $G_i$ and sites in $\partial\Gamma$. By the locality of the test $P_{\partial\Gamma }$ as well as that of the potential $\Phi$, this can only happen if $|G|$ scales linearly with $\dist(A,B)$. Finally, if $\{\Phi(X)\}_{X\ssubset \ZZ^d}$ is assumed to be commuting, the bound in Equation \eqref{boundnoncommuting} can be refined as follows, thus leading to the desired claim:
\begin{align*}
    |f(\mathbf{z})|&\le \Big|\ln\tr\big[e_{\mathbf{z}}^\Gamma P_{\partial\Gamma} {P}_B {N}_A \big]\Big|+\Big|\ln\tr\big[e_{\mathbf{z}}^\Gamma\,P_{\partial\Gamma }{P}_B' \big]\Big|+\Big|\ln\tr\big[e^\Gamma_{\mathbf{z}}\,P_{\partial \Gamma}{P}_B \big]\Big|+\Big|\ln\tr\big[e^\Gamma_{\mathbf{z}}\, {P}_B' P_{\partial\Gamma}{N}_A\big]\Big|\\
    &\le 4c'|\Gamma|\,.
\end{align*}

\eqref{qIIId}$\Rightarrow$ \eqref{qL20}$\Leftrightarrow$\eqref{qL2}:  The equivalence was already proved in Proposition 17 of \cite{[BK16]}. Hence, it is enough to show the first implication. For this, we first reduce the problem to proving the bound for $X_A, Y_B\ge 0$: Indeed, decomposing $X_A:=X_A^{+}-X_A^{-}$ and $Y_B:=Y_B^{+}-Y_B^{-}$ into their positive and negative parts, we have
\begin{align}
\operatorname{Cov}^{(0)}_{\sigma^{\Lambda, P_{\partial \Gamma}}}(X_A,Y_B)&=\big|\tr\big[\sigma^{\Lambda, P_{\partial \Gamma}}{X}_A{Y}_B \big]-\tr\big[ \sigma^{\Lambda, P_{\partial \Gamma}} X_A\big]\,\tr\big[\sigma^{\Lambda, P_{\partial \Gamma}} Y_B\big]\big|\nonumber\\
&\le \sum_{\gamma,\alpha\in\{\pm\}}\big|\tr\big[\sigma^{\Lambda, P_{\partial \Gamma}}{X}^\gamma_A{Y}^\alpha_B  \big]-\tr\big[  \sigma^{\Lambda, P_{\partial \Gamma}} X_A^\gamma \big]\,\tr\big[\sigma^{\Lambda, P_{\partial \Gamma}} Y_B^\alpha\big]\,\big| \, . \label{eq}
\end{align}
Next, choosing $N_A=X_A^\gamma$, $P_B=Y_B^\alpha/\|Y_B^\alpha \|_\infty $ and $P_B^{\prime}=\Id$ in \eqref{qIIId}, we have that 
\begin{align*}
&\big|  \tr[\sigma^{\Lambda, P_{\partial \Gamma}} Y_B^{\alpha}X_A^{\gamma}]  -\tr[\sigma^{\Lambda, P_{\partial \Gamma}}X_A^{\gamma}]\,\tr\big[\sigma^{\Lambda, P_{\partial \Gamma}} Y_B^{\alpha}\big]\big| \\
& \phantom{asdasdasdasdsdasdasdadsasdasdasda} \leq  c\, | \Gamma |\,\e^{-\frac{\dist(A ,B)}{\xi}}\,\tr\big[\sigma^{\Lambda, P_{\partial \Gamma}} Y_B^{\alpha}\big]\tr[\sigma^{\Lambda, P_{\partial \Gamma}}X_A^{\gamma}]\\
&\phantom{asdasdasdasdsdasdasdadsasdasdasda} \le c\,| \Gamma |\,\e^{-\frac{\dist(A,B)}{\xi}}\,\Big(\tr\big[\sigma^{\Lambda, P_{\partial \Gamma}} (Y_B^{\alpha})^2\big]\Big)^{\frac{1}{2}}\,\Big(\tr[\sigma^{\Lambda, P_{\partial \Gamma}}(X_A^{\gamma})^2]\Big)^{\frac{1}{2}}\\
&\phantom{asdasdasdasdsdasdasdadsasdasdasda} \le c\,| \Gamma |\,\e^{-\frac{\dist(A,B)}{\xi}}\,\Big(\tr\big[\sigma^{\Lambda, P_{\partial \Gamma}} (Y_B)^2\big]\Big)^{\frac{1}{2}}\,\Big(\tr[\sigma^{\Lambda, P_{\partial \Gamma}}(X_A)^2]\Big)^{\frac{1}{2}}\,,
\end{align*}
where in the last inequality, we used that e.g. $X_A^2=(X_A^+)^2+(X_A^-)^2\ge (X_A^+)^2,(X_A^{-})^2$. Inserting the last bound into Equation \eqref{eq} we have
\begin{align*}
\operatorname{Cov}^{(0)}_{\sigma^{\Lambda, P_{\partial \Gamma}}}(X_A,Y_B)\le 4c\,| \Gamma |\,\e^{-\frac{\dist(A,B)}{\xi}}\,\Big(\tr\big[\sigma^{\Lambda, P_{\partial \Gamma}} (Y_B)^2\big]\Big)^{\frac{1}{2}}\,\Big(\tr[\sigma^{\Lambda, P_{\partial \Gamma}}(X_A)^2]\Big)^{\frac{1}{2}}\,.
\end{align*}

\eqref{qL2}$\Rightarrow$\eqref{Linfty}: follows directly from the fact that $\|X\|_{\operatorname{L}_2(\sigma^{\Lambda, P_{\partial \Gamma}})},\Big(\tr\big[\sigma^{\Lambda, P_{\partial \Gamma}} X^2  \big]\Big)^{\frac{1}{2}}\le \|X\|_\infty$. 

\end{proof}

In order to prove the modified logarithmic Sobolev inequality, we need a condition introduced in \cite{bardet2020approximate}.

	\begin{definition}[$\mathbb{L}_1\to\mathbb{L}_\infty$ clustering of correlations]
		Let $\cL:=\big\{ \cL_\Lambda \big\}_{\Lambda\ssubset\ZZ^d}$ be a uniform family of primitive, reversible and frustration-free Lindbladians with corresponding unique fixed points $\{\sigma^{\Lambda}\}_{\Lambda\ssubset \ZZ^d}$. The family $\cL$ satisfies the \textit{$\mathbb{L}_1\to\mathbb{L}_\infty$} clustering of correlations if there exist constants $c\ge0$ and $\xi>0$ such that for any intersecting $C,D\ssubset \ZZ^d$, 
		\begin{align}\label{LiLinftyaa}\tag{q$\mathbb{L}_1\to\mathbb{L}_\infty$}
		&\max_{i\in I_{\partial( C\cup D)}}\,\big\|   E^{(i)}_C\circ E^{(i)}_D-E_{C\cup D}^{(i)}:\,\mathbb{L}_1(\tau_i^{C\cup D})_{\operatorname{sa}} \to \cB(\cK^{C\cup D}_i )_{\operatorname{sa}} \big\|\le c\,|C\cup D|\,\e^{-\frac{\dist(C\backslash D,D\backslash C)}{\xi}}\,,
		\end{align}
		where the maps $E^{(i)}_G$, $G\in\{C,D,C\cup D\}$, are defined as in \Cref{EAi}. 
\end{definition}

 In the next proposition, we show that the condition \eqref{LiLinftyaa} is a consequence of \eqref{qIIId} for classical states as well as for the Schmidt Gibbs sampler with nearest neighbour interactions. The proof is inspired by that of Lemma 4.2 in \cite{cesi2001quasi}.
 
\begin{proposition}\label{propDScorrdecay}
	Let $\cL$ be either $\operatorname{(a)}$ the embedded Glauber dynamics, or $\operatorname{(b)}$ the Schmidt Gibbs sampler with nearest neighbour interactions. Assume moreover that the potential satisfies \eqref{qIIId}. Then $\cL^S$ satisfies \eqref{LiLinftyaa}.
\end{proposition}

\begin{proof}
	
The case (a) is a simple consequence of the reasoning in the proof of Lemma 4.2 in \cite{cesi2001quasi} as well as \Cref{EAclassEA}. In fact, its proof is even more direct than the one of Cesi, where a summation over the boundary is performed. This is due to the fact that our definition for \eqref{qIIId} already allows for different boundary conditions over finite regions $D$ (cf. \Cref{remarkboundarycond}). 

	\begin{figure}[!ht]
		\centering
		\includegraphics[width=0.7\linewidth]{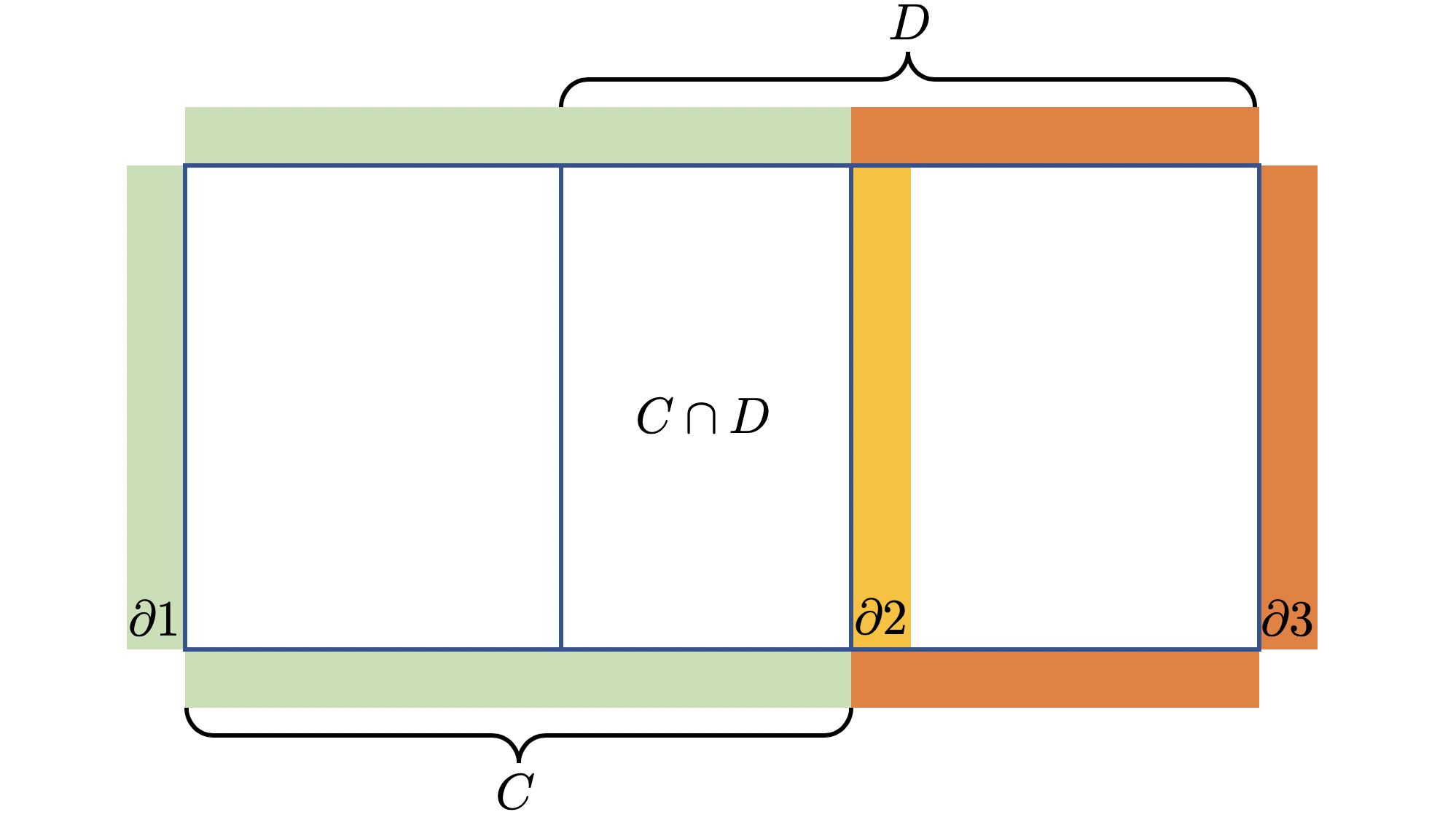}
		\caption{Decomposition of the boundaries of $C$ and $C \cup D$ for Proposition \ref{propDScorrdecay}. Here $\partial 1:=\partial C\backslash D$, $\partial 2:=\partial C \cap D$ and $\partial 3:=\partial D\backslash C\partial $. Therefore $\partial 1\cup \partial 2=\partial C$ whereas $\partial 1\cup \partial 3=\partial(C\cup D)$. Here, we considered nearest neighbour interactions.}
		\label{figuretilingCD}
	\end{figure}

We now turn our attention to the proof of (b). We recall that, in this case, a central projection $P^{C\cup D}_{\alpha}$ is labeled by a configuration $i\equiv \alpha:=(\alpha_1,...,\alpha_{|C\cup D|})$ in $I_{\partial (C\cup D)}$. Moreover, it can be decomposed into a product of projections onto each site of the boundary $\partial (C\cup D) $:
$$P^{C\cup D}_{\alpha}:=\bigotimes_{j\in\partial (C\cup D)}P_j^{\alpha_j}\,.$$
 Next, decomposing the boundary of $C$ into $\partial 1:=\partial C\backslash D$ and $\partial 2:=\partial C\cap D$ as in \Cref{figuretilingCD}, we choose a configuration $\beta:=(\alpha^{\partial 1},\beta^{\partial 2})\in I_{\partial C}$ which coincides with $\alpha$ in $\partial 1$, and denote $P_\beta^C:=P_{\partial1}^{\alpha^{\partial 1}}\otimes P_{\partial 2}^{\beta^{\partial 2}}$. Next, we let $X\in\cB(\cK_\alpha^{C\cup D})_+$ with $\|X\|_{\mathbb{L}_1(\tau_{\alpha}^{C\cup D})}=1$. We define $N_{C\backslash D}:= E_D^{(\alpha)}[X]\in\cB(\cH_{C\backslash D})$,
	$P_{\partial \Gamma}:=P_\alpha^{C\cup D}$, $P_{D\backslash C}':=\Id_{D\backslash C}$ and $P_{D\backslash C}:= P^{\alpha,C}_{\beta^{\partial 2}}\in\cB(\cH_{D\backslash C})$. Then, by \eqref{qIIId}:
	\begin{align*}
	\big|  \tr[\sigma^{\Lambda, P_\alpha^{C\cup D}P^{\alpha,C}_{\beta^{\partial 2}}}N_{C\backslash D}]  -\tr[\sigma^{\Lambda,P_\alpha^{C\cup D}}N_{C\backslash D}]\,\big|\le  c\,|C\cup D|\,\e^{-\frac{\dist(C\backslash D,D\backslash C)}{\xi}}\, \tr[\sigma^{\Lambda,P_\alpha^{C\cup D}}N_{C\backslash D}]\,.
	\end{align*}
	Moreover, by construction, we have that  $\sigma^{\Lambda,P_\alpha^{C\cup D}}\equiv \tau_\alpha^{C\cup D}$ and $\sigma^{\Lambda, P_\alpha^{C\cup D} P^{\alpha,C}_{\beta^{\partial 2}}}\equiv \tau^{\alpha,C}_{\beta^{\partial 2}}$ with the notations of \eqref{EAi}.

	Therefore:
	\begin{align*}
&	\|E_C^{(\alpha)}\circ E_D^{(\alpha)}[X]-E_{C\cup D}^{(\alpha)}[X]\|_\infty\\
&  \phantom{asdasdasdasdasdad}=\|E_C^{(\alpha)}[N_{C\backslash D}]-E_{C\cup D}^{(\alpha)}[N_{C\backslash D}]\|_\infty\\
	&\phantom{asdasdasdasdasdad}=\Big\|\sum_{\beta^{\partial 2}\in I^{\alpha}_{\partial C}}\Big(\tr[P^{\alpha,C}_{\beta^{\partial 2}}N_{C\backslash D}P^{\alpha,C}_{\beta^{\partial 2}}\tau_{\beta^{\partial 2}}^{\alpha,C}]-\tr[P_\alpha^{C\cup D}N_{C\backslash D}P_\alpha^{C\cup D}\,\tau_\alpha^{C\cup D}]\Big)\,\Id_{\beta^{\partial 2}}^{\alpha ,C}\Big\|_\infty\\
		&\phantom{asdasdasdasdasdad}=\max_{\beta^{\partial 2}\in I^\alpha_{\partial C}}\,\big|     \tr[\sigma^{\Lambda,P_\alpha^{C\cup D}}\,N_{C\backslash D}]-\tr[\sigma^{\Lambda, P_\alpha^{C\cup D} P^{\alpha,C}_{\beta^{\partial 2}}}\,N_{C\backslash D}]\,  \big|_\infty\\
	&\phantom{asdasdasdasdasdad}\le c\,|C\cup D|\,\e^{-\frac{\operatorname{dist}( C\backslash D,D\backslash C)}{\xi}}\,\tr\big[ \tau_\alpha^{C\cup D}X\big]\\
	&\phantom{asdasdasdasdasdad}=c\,|C\cup D|\,\e^{-\frac{\operatorname{dist}( C\backslash D,D\backslash C)}{\xi}}\,, 
	\end{align*}
	where we are using the explicit form for the conditional expectations from Equation \eqref{EAi} as well as \eqref{qIIId}.
	\end{proof}

In the case of the embedded Glauber dynamics, one can easily relate \Cref{LiLinftyaa} to Dobrushin and Shlosman's complete analyticity:

\begin{proposition}\label{classclust}
	Let $\cL^{G}$ be the uniform family of Glauber dynamics introduced in \Cref{Markov}. Then \Cref{LiLinftyaa} holds whenever \eqref{l1linfty} holds. In particular, \eqref{LiLinftyaa} holds for 1D systems, as well as for any dimensions above the critical temperature.
\end{proposition}
\begin{proof}
	This is simply a consequence of  the embedding in \eqref{EAclassEA}.
\end{proof}

\section{Clustering of correlations implies MLSI}\label{sec:main}
In this section, we prove the positivity of the MLSI constant for  generators defined over lattice spin systems under the \eqref{LiLinftyaa} clustering of correlations defined in \Cref{clustering}. 
But before we do this, we discuss two warm-up examples to give some intuition on our approach by proving the MLSI in two settings: embedded Glauber dynamics with dephasing and embedded Glauber dynamics in $1D$.
These two steps, although interesting in their own right, are simple and illustrate the key ideas of our approach.
\subsection{Intuitive outline of proof strategy: MLSI with dephasing and the importance of CMLSI}
Let us start by showing how to derive a MLSI directly by adding an additional dephasing to the generator of the embedded Glauber dynamics. Let $\overline{\cL}_{\Lambda}^{{G}}$ be the generator of an embedded Glauber dynamics and consider the family of generators $\overline{\cL}_{\Lambda}'=\overline{\cL}_{\Lambda}^{{G}}+(\mathcal{C}_{\Lambda}-\text{id})$, where $\mathcal{C}_{\Lambda}$ is the pinching with respect to the computational basis. That is, the classical dynamics with additional (global) dephasing. 
The chain rule \eqref{chainrulerelatent} for the relative entropy implies that for all states $\rho\in\cD(\cH_{\Lambda})$:
\begin{align*}
D(e^{t\overline{\cL}_{\Lambda *}'}(\rho)\|\sigma)=D(e^{t\overline{\cL}_{\Lambda*}'}(\rho)\|\mathcal{C}_\Lambda(\e^{t\overline{\cL}_{\Lambda*}'}\rho))+D(\mathcal{C}_\Lambda(e^{t\overline{\cL}_{\Lambda*}'}(\rho))\|\sigma)\,.
\end{align*}
Now, note that the generators $\overline{\cL}_{\Lambda}'$ and $\overline{\cL}_{\Lambda}^{G}$ commute. This immediately yields
 \begin{align}\label{equ:commutelind}
 D(e^{t\overline{\cL}_{\Lambda*}'}(\rho)\|\sigma)
  =D(e^{t(\mathcal{C}_\Lambda-\text{id})}(e^{t\overline{\cL}_{\Lambda*}^{G}}(\rho))
   \|\,\mathcal{C}_\Lambda(\e^{t\overline{\cL}_{\Lambda*}^{G}}(\rho)))
   +D(e^{t\overline{\cL}_{\Lambda*}^{G}}(\mathcal{C}_\Lambda(\rho))\|\sigma).
 \end{align}
Note that the second term on the r.h.s. of Equation~\eqref{equ:commutelind} is the relative entropy between two classical states and, thus, if the classical Glauber dynamics satisfies a MLSI with constant $\alpha_1$, we have 
$$D(e^{t\overline{\cL}_{\Lambda*}^{G}}(\mathcal{C}_\Lambda(\rho))\|\sigma)\leq e^{-4\alpha_1 t}D(\mathcal{C}_\Lambda(\rho)\|\sigma)\,.$$ 
To control the first term on the r.h.s. of \eqref{equ:commutelind}, note that the dephasing semigroup satisfies a MLSI with constant $\tfrac{1}{4}$, and so 
$$D(e^{t\overline{\cL}_{\Lambda*}'}(\rho)\|\mathcal{C}_\Lambda(\e^{t\overline{\cL}_{\Lambda*}^{{G}}}(\rho)))\leq e^{-t }D(e^{t\overline{\cL}_{\Lambda*}^{{G}}}(\rho)\|\mathcal{C}_\Lambda(e^{t\overline{\cL}_{\Lambda*}^{{G}}}(\rho)))\,.$$ 
Thus, by the data processing inequality:
\begin{align*}
D(e^{t\overline{\cL}_{\Lambda*}'}(\rho)\|\sigma)&\leq \max\{e^{-4\alpha_1 t},e^{-t}\}(D(\rho\|\mathcal{C}_\Lambda(\rho))+D(\mathcal{C}_\Lambda(\rho)\|\sigma))\\
&=\max\{e^{-4\alpha_1 t},e^{-t}\}D(\rho\|\sigma)\,,
\end{align*}
by another application of the chain rule. We conclude that with the extra dephasing semigroup on top of the classical Glauber dynamics, the semigroup $(\e^{t\overline{\cL}_\Lambda'})_{t\ge 0}$ satisfies a MLSI with a constant that is given by the minimal of the dephasing rate and the constant for the classical dynamics, $\min\{\alpha_1,\tfrac{1}{4}\}$. The same argument would also apply if instead we added local dephasing noise on each site, since the dephasing semigroup satisfies CMLSI with the same constant $\frac{1}{4}$. The lesson to be learned from the example above is that the application of the chain rule allowed us to handle the dynamics in the computational basis and the dephasing separately. This will be crucial for our analysis later and will motivate the introduction of the \emph{pinched} MLSI in Definition~\ref{def:pinchedmlsi}. 

In order to prove our main result and get rid of the additional dephasing assumed above, we will also need to resort to complete MLSI inequalities. Indeed, by Theorem~\ref{CMLSIholds} we have that all local terms of the generator satisfy CMLSI.

This result is crucial to generalize the argument of MLSI with extra dephasing given before. It is instructive to shortly consider the implication of the complete MLSI for embedded Glauber dynamics before moving on to our main result. Although we restrict the discussion to the one dimensional Ising model ${\cL}^{\operatorname{Ising}}$ with nearest neighbour interactions, our argument would easily extend to higher dimensions: Define the sets $A_i=\{i\}$ for $i\in\mathbb{Z}$, so that $A_i\partial=\{i-1,i,i+1\}$. Clearly $ \cup_{i\in \mathbb{Z}} A_{2i}\partial$ tile the whole integers. Moreover, for this classical Glauber dynamics with nearest neighbour interactions, we also have that the conditional expectations $E_{A_{2i}}$ commute for all $i\in\ZZ$ (cf. \Cref{EAcondexp}). Then, for $\Lambda=[-(2n+1),2n+1]$, denoting $ A:= \cup_{i\in \mathbb{Z}} A_{2i}$ and defining $E_{A\cap \Lambda}:=\prod_{k=-n}^{n}E_{A_{2k}}$, we have for all $\rho\in\cD(\cH_\Lambda)$:
\begin{align}\label{equ:chainruleagain}
D(\rho_t\|\sigma^\Lambda)=D(\rho_t\|E_{A\cap \Lambda*}(\rho_t))+D(E_{A\cap \Lambda*}(\rho_t)\|\sigma^\Lambda)\,,
\end{align}
where $\rho_t:=\e^{t\overline{\cL}^{\operatorname{Ising}}}(\rho)$. Through the application of the CMLSI we are able to control the first term on the r.h.s. of~\eqref{equ:chainruleagain}, as the size of the region on which each conditional expectation acts is bounded. Moreover, since the image of the conditional expectation $E_{A\cap\Lambda *}$ is diagonal in the computational basis over $\cH_\Lambda$, the second relative entropy in \eqref{equ:chainruleagain} is classical. Therefore, we can control it in terms of the classical modified logarithmic Sobolev inequality constant:
\begin{align*}
D(\rho_t\|\sigma^\Lambda)&\le \e^{-\alpha_{\ccc}(\overline{\cL}^{\operatorname{Ising}}_{A\cap\Lambda})t}D(\rho\|E_{A\cap \Lambda*}(\rho))+\e^{-\alpha(L^{\operatorname{Ising}})t}D(\mathcal{C}_\Lambda(\rho)\|\sigma^\Lambda)\\
&\le 2 \max\big\{ \e^{-\alpha_{\ccc}(\overline{\cL}^{\operatorname{Ising}}_{A\cap\Lambda})t},  \e^{-\alpha(L^{\operatorname{Ising}})t}\big\}\,D(\rho\|\sigma^\Lambda)\,,
\end{align*}
 where $L^{\operatorname{Ising}}$ is the restriction of $\overline{\cL}^{\operatorname{Ising}}$ to the classical algebra. By a direct extension of the above method, we arrive at the following result:
\begin{remark}
	Let $\cL^{G}$ be a uniform family of embedded Glauber Lindbladians, and denote by $L^{{G}}$ their restriction to the classical algebra. Then the following conditions are equivalent:
\begin{itemize}
	\item[(i)]  $L^{G}$ satisfies \eqref{DSM} for some constant $\gamma\in (0,\infty)$.
		\item[(ii)] $L^{G}$ is gapped.
	\item[(iii)] $L^{G}$ has a positive logarithmic Sobolev constant independent of the system size.
	\item[(iv)] There exists $\tilde{\alpha}>0$ such that, for any system $\Lambda$, for all $t\ge 0$, and all $\rho\in\cD(\cH_\Lambda)$,
	\begin{align}\label{MLSI2}
	D(\e^{t\overline{\cL}_\Lambda^{G}}(\rho)\|\sigma^\Lambda)\le 2\,e^{-\tilde{\alpha} t}D(\rho\|\sigma^\Lambda)\,.
	\end{align}
	\item[(v)] $\cL^{G}$ satisfies the rapid mixing condition.
\end{itemize}
\end{remark}
	 
Although the bound \eqref{MLSI2} is enough to derive rapid mixing and its consequences, it is unsatisfactory from a mathematical point of view, due to the presence of the factor $2$ on its right-hand side which prevents us from claiming the existence of a modified logarithmic Sobolev constant for embedded Glauber dynamics and renders the bound trivial for small times. Moreover, we would like to extend the result to non-classical Gibbs states. The analysis carried out in the next sections will allow us to solve both these issues.

\subsection{Geometric conditions for MLSI}\label{geomconditions}
In this subsection, we introduce a condition inspired from the use of the map $E_{A\cap \Lambda*}$ above that we require for our proof of the MLSI. We recall that, given a family of conditional expectations $\{E_\Lambda\}_{\Lambda\ssubset \ZZ^d}$ associated to a $\kappa$-local Gibbs sampler,
 $$\partial\Lambda:=\supp(E_\Lambda)\backslash \Lambda,~~~\text{ and }~~~\Lambda\partial=\Lambda\cup\partial\Lambda\,.$$
 Next, define a coarse-graining of $\ZZ^d$ as follows: Given the hypercube $A_0=[0,D-1]\times \dots \times [0,D-1]\ssubset \ZZ^d$ of size $D^d$, for some integer $D>2\kappa$, we cover the whole lattice $\ZZ^d$ with translations $A_j$ of $A_0$ and their boundaries. In what follows, the sets $A_j$ will be called \textit{pixels}. More explicitly, singling out the first coordinate basis $e_1\in\ZZ^d$, we first construct a non-planar sheet of pixels orthogonal to $e_1$: for $j=1,...,d-1$, define the translations $\cT^{j}$ by the vector 
 \begin{align*}
     t^{j}:=\Big( \Big\lfloor \frac{D-1+\kappa}{2} \Big\rfloor, 0,...,0, \underbrace{D-1+\kappa}_{j+1\text{-th coordinate}},0,...,0 \Big)\,,
 \end{align*}
and define the sets $A_{j}:=\mathcal{T}^{j}(A_0)$. The rest of the non-planar sheet is constructed by translations of the $d$ pixels $A_j$, $j=0,...,d-1$, by the vectors $\pm 2\ell(D+\kappa-1)e_k$, $k\in \{2,...,d\}$ and $\ell\in\NN$. In a final step, we translate all the pixels generated by the previous procedure by the vectors $\pm \ell(D+\kappa-1)e_1$, $\ell\in \NN$. We refer to the set generated by the translation of a pixel along the direction $e_1$ as a \textit{column} $C$ of the tiling. The centre of the column refers to the sites in $C$ that are at distance at least $\kappa$ from $C^c$. The tiling $A:=\bigcup_{j\in\cJ}A_j $ generated this way enjoys the following two properties:
\begin{itemize}
    \item[(i)] For any $i\in \cJ$ 
   \begin{equation*}
       \underset{\mathcal{J}\ni j \ne i}{\text{inf}} \, \{\text{dist}(A_i, A_j) \} = \kappa \, .
   \end{equation*}
   \item[(ii)] The pixels and their boundaries cover the whole lattice, i.e.
     \begin{equation*}
       \underset{j \in \mathcal{J}}{\bigcup} \,  A_j \partial =  \mathbb{Z}^d \, .
   \end{equation*}
\end{itemize}

	\begin{figure}[!ht]
	\centering
	\includegraphics[width=0.7\linewidth]{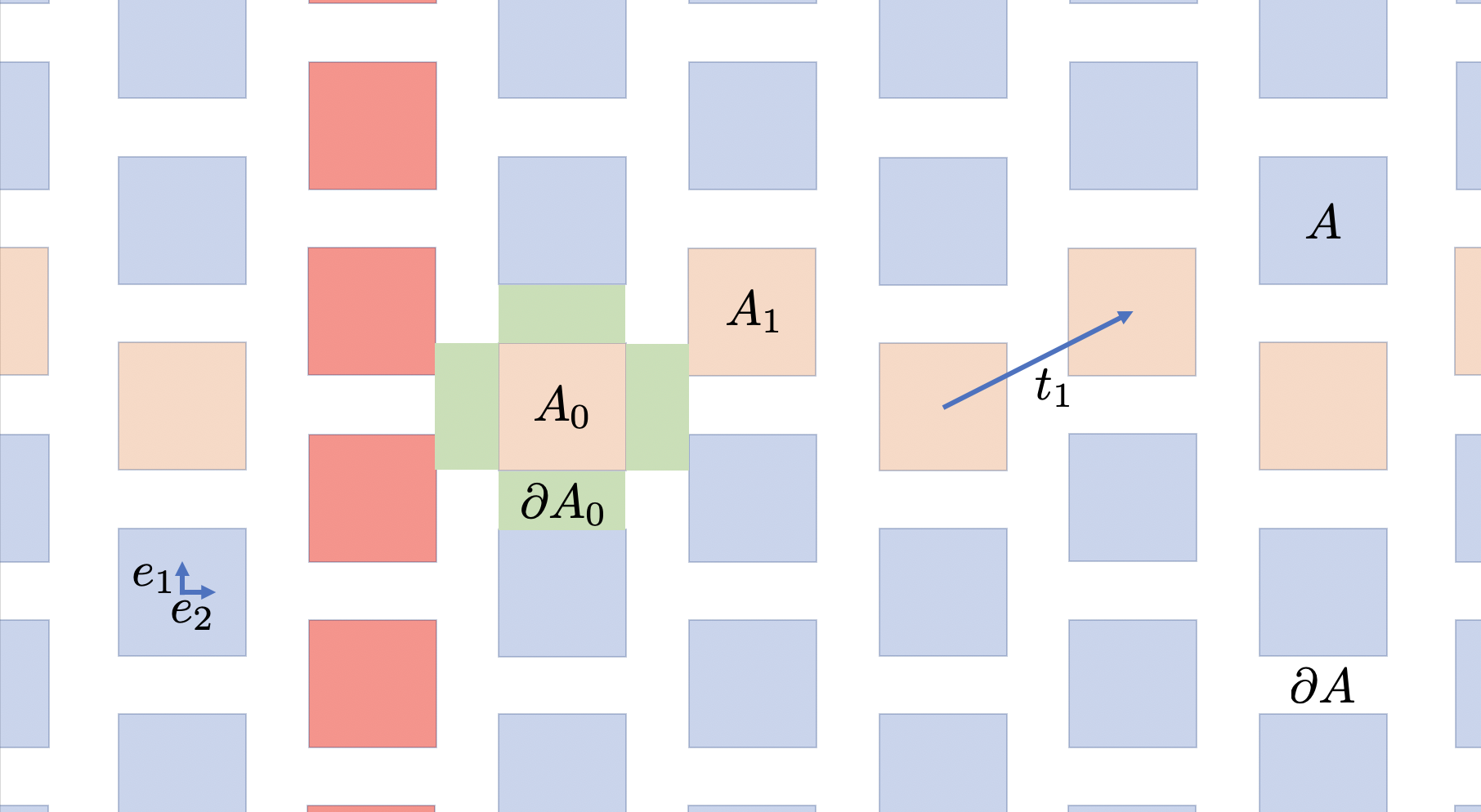}
	\caption{Tiling of $\ZZ^2$. Here, we assumed two-local interactions. The pixels in pink represent the non-planar sheet constructed out of translations of the pixels $A_0$ and $A_1$. The ones in red correspond to a column.}
\label{figuretiling}
	\end{figure}

These conditions lead to a pattern such as the one showed in \Cref{figuretiling}. Moreover, note that from condition (i) above, clearly $A_i \cap A_j \partial = \emptyset$ for any two different pixels $A_i$ and $A_j$. Any two pixels with intersecting boundaries will be called adjacent. Next, a \textit{cluster of pixels} is any finite union $A_{\mathcal{J}'}:=\bigcup_{i\in\mathcal{J}'} A_i$ of pixels such that, for any two $i,j\in\mathcal{J}'$, there exists a path of adjacent pixels in $A_{\mathcal{J}'}$  connecting $A_i$ and $A_j$. 
\begin{lemma}\label{grainedsets}
    Given any cluster of pixels $A_{\mathcal{J}'}:=\bigcup_{j\in \mathcal{J}'}A_j$, there exists a finite connected set $\widetilde{S}_{\mathcal{J}'}$ such that:
    \begin{itemize}
        \item[$\operatorname{(i)}$] $A_{\mathcal{J}'}\subset \widetilde{S}_{\mathcal{J}'}$\,;
        \item[$\operatorname{(ii)}$] $\partial \widetilde{S}_{\cJ'}\subset \partial A_{\mathcal{J}'}$\,;
        \item[$\operatorname{(iii)}$]  $A_{\cJ'}\partial^c =\widetilde{S}_{\cJ'}\partial^c$\, .
    \end{itemize}
\end{lemma}

\begin{proof}
The proof proceeds by an enumeration of the sites at the boundary of $A_{\cJ'}$. First, we consider the situation when a site $j\in \partial A_{\cJ'}$ belongs to a column $C$: we distinguish two cases:
\begin{itemize}
    \item[(1)]$j$ is in the boundary of two pixels in $A_{\cJ'}$: in that case, we keep it if it is in the centre of $C$.
    \item[(2)] $j$ is in the boundary between a pixel in $A_{\cJ'}$ and a pixel in $A_{\cJ'}^c$: in this case, we reject $j$.   
\end{itemize}
This procedure permits to join adjacent pixels of $A_{\cJ'}$ belonging to a same column. Next, we consider sites in the boundary of $A_{\cJ'}$ which sit in between two columns $C$ and $C'$: then the problem reduces to a 2 dimensional problem (see \Cref{figuretiling2}). Here again, we need to distinguish between different situations. 
\begin{itemize}
    \item[(1')] $j$ is in between three pixels in $A_{\cJ'}$: then it is kept. 
    \item[(2')] either $C$ or $C'$ does not contain any pixel of $A_{\cJ'}$ whose boundary contains $j$: in that case, we reject $j$.
    \item[(3')] both $C$ and $C'$ exactly contain one pixel of $A_{\cJ'}$ whose boundary contains $j$: in that case, we keep $j$ if it lies in the intersection of the boundaries of the aforementioned pixels. Otherwise, we reject it.
\end{itemize}
This second separation of cases permits us to join adjacent pixels of $A_{\cJ'}$ which belong to different columns. Then, we define the set $\widetilde{S}_{\cJ'}$ as the smallest simply connected set which includes the union of the sites $j$ kept and $A_{\cJ'}$. Those sites $j$ which were rejected constitute the boundary of $\widetilde{S}_{\cJ'}$.

\end{proof}

For any cluster of pixels $A_{\mathcal{J}'}$, we call the largest set $\widetilde{S}_{\cJ'}$ satisfying conditions $(i)$ and $(ii)$ of \Cref{grainedsets} a \textit{grained set}, and denote the set $\partial A_{\cJ'}\backslash \partial \widetilde{S}_{\cJ'} \equiv \partial(\widetilde{S}\cap A)_{\operatorname{in}}$ (see  \Cref{figuretiling2}). The collection of all grained sets in $\ZZ^d$ is denoted by $\widetilde{\mathcal{S}}$.

 \begin{figure}[!ht]
	\centering
	\includegraphics[width=0.7\linewidth]{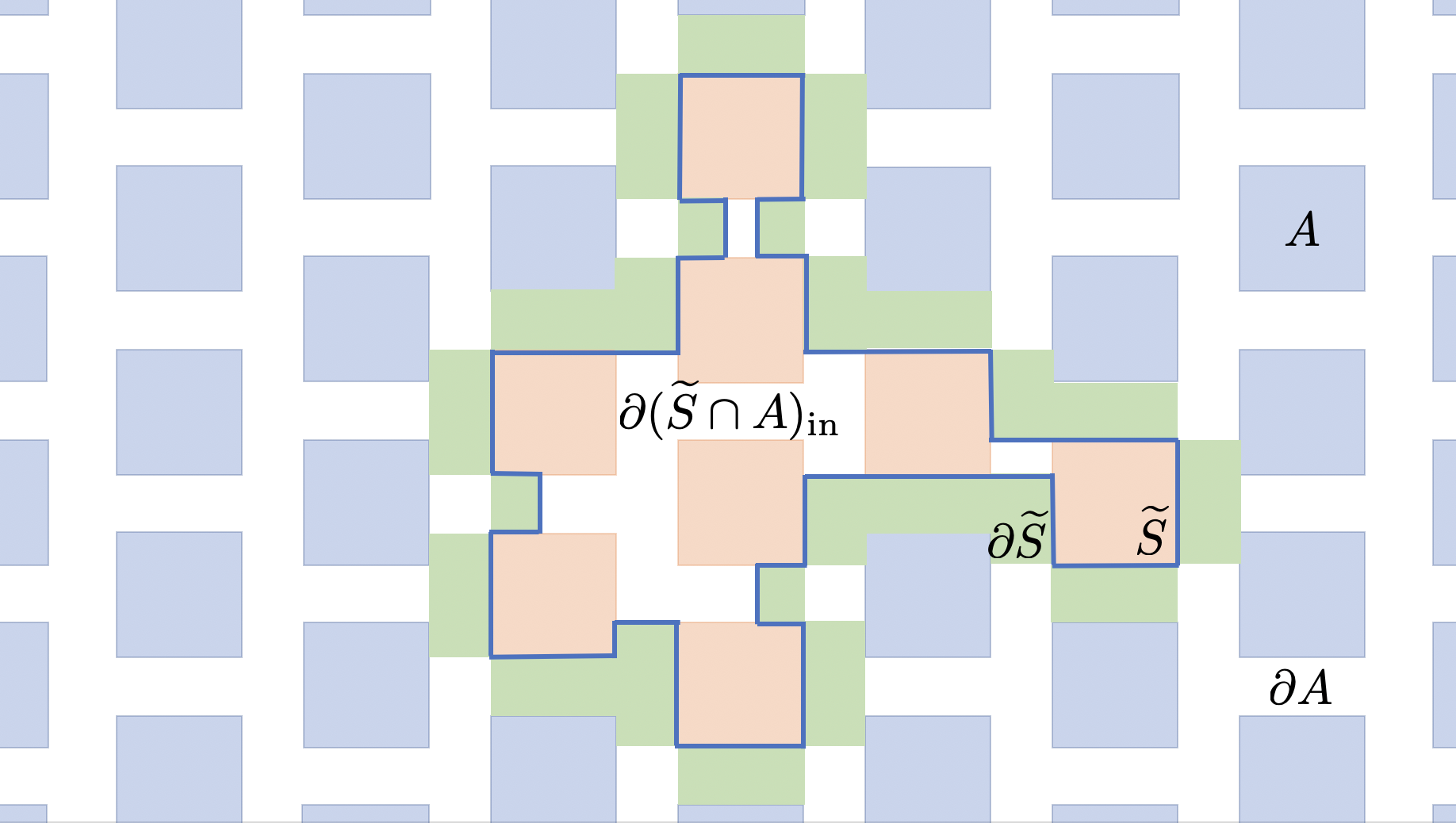}
	\caption{A grained set in $\ZZ^2$. Here again, we assume two-local interactions.}
	\label{figuretiling2}
\end{figure}

With these definitions and properties at hand, we are now ready to define our third condition for the existence of a MLSI:
		\begin{condition}\label{Pi}
The covering $A=\bigcup_{i\in \mathcal{J}} A_i$ defined above satisfies:
\begin{itemize} 
	\item[(i)] For all $i,j\in\mathcal{J}$, $E_{A_i}\circ E_{A_j}=E_{A_j}\circ E_{A_i}=E_{A_i\cup A_j}$; and 
	\item[(ii)] For any grained set  $\widetilde{S}\in \widetilde{\mathcal{S}}$, 

  there exists a decomposition $\cK_j^{\widetilde{S}}:=\bigoplus_{k}\cH^{(j,k)}$ such that
	$$\cF(\cL_{A\cap\widetilde{S}}):=\Id_{A\cap\tilde{S}}\otimes\bigoplus_{j\in I_{\partial\widetilde{S}}}\bigoplus_k\,\Id_{\cH^{(j,k)}}\otimes \cB(\cH_{j}^{\widetilde{S}})\,.$$

\end{itemize}
\end{condition}

\begin{remark}
Note that from \Cref{Pi}(i) and frustration-freeness it clearly follows that
$$ E_{\widetilde{S}}\circ E_{A}=E_{A}\circ E_{\widetilde{S}}$$
for any grained set  $\widetilde{S}\in \widetilde{\mathcal{S}}$. This observation will be used at the very last step of the proof of the main result (see \Cref{betatoalpha}).
\end{remark}

\Cref{Pi}(i) is crucially needed together with \Cref{CMLSIholds} in order to control the relative entropy by the CMLSI constant over a fixed sized region $A_j\in A$ in the decomposition \eqref{eqchainrule} below. Moreover, \Cref{Pi}(ii) plays a crucial role in the reduction of the analysis into smaller blocks (see a first discussion at the end of this subsection). For the time being, in the next propositions, we prove that \eqref{Pi}(i) is satisfied for various Gibbs samplers and that Condition \eqref{Pi}(ii) is satisfied for Schmidt semigroups.

\begin{lemma}[Examples for Conditions \eqref{Pi}(i)]\label{lemma:ConditionI}
Let $A,B\ssubset \ZZ^d$ be two regions such that $A\partial\cap B=\emptyset$ and $A\cap B\partial=\emptyset$. Then $E_A\circ E_B=E_B\circ E_A=E_{A\cup B}$ for the conditional expectations of the Schmidt semigroups corresponding to commuting potentials. 
\end{lemma}
\begin{proof}
The proof follows a similar path as Lemma~\ref{lemmaSchmidt}. First, let us see that the conditional Schmidt expectations commute. From the assumptions $A\partial\cap B=\emptyset, A\cap B\partial=\emptyset$, we see from the decomposition in Equation~\eqref{equ:decomp} that $E_A$ and $E_B$ only act nontrivially on disjoint Hilbert spaces. This is because there is no edge that connects both $A$ and $B$ to the same vertex in the intersection of their boundaries. From this it follows that they commute. The fact that the corresponding product is a conditional expectation $E_{A\cup B}$ follows along the same lines.
\end{proof}

\begin{remark}
For classical evolutions over quantum systems, we can heavily simplify the above construction, as we can take regions $A_j$ whose union tiles the lattice. In particular, $E_A$ and $E_B$ commute even for $A\cap\partial B\ne \emptyset$. However, this is not necessarily the case in the commuting setting, even in the case of the $2$-local Schmidt conditional expectations, since there are edges on which both $E_A$ and $E_B$ act non-trivially.
 \end{remark}

\begin{proposition}[Examples for Condition \eqref{Pi}(ii)]
Condition \eqref{Pi}$\operatorname{(ii)}$ holds for the Schmidt semigroups corresponding to $2$-local interactions in any dimension, as well as for embedded Glauber dynamics.  
\end{proposition}

\begin{proof}
The case of embedded Glauber dynamics is obvious, since the boundary can be decomposed into tensor products of local projections onto the classical basis. We focus our attention to the case of Schmidt generators. Following the proof of Lemma~\ref{lemmaSchmidt}, we see that $\operatorname{Ker}(\cL_{A\cap\widetilde{S}})=\cF(E_{A\cap\widetilde{S}})$. Moreover, it is not difficult to see that $A\cap\widetilde{S}$ corresponds to a cluster of pixels. The desired decomposition then immediately follows from $\mathcal{A}_{A\cap\widetilde{S},\operatorname{out}}$ given in Equation~\eqref{equ:decomp}. 
\end{proof}

The proposition below provides a justification to the introduction of \Cref{Pi} and showcases how working with a restricted set of input states significantly simplifies the analysis of the relative entropy on different regions:

\begin{theorem}[Approximate tensorization of the relative entropy\label{ATAC}]
Let $\cL$ be a Gibbs sampler corresponding to a commuting potential. Assume further that the family $\cL$ satisfies \eqref{LiLinftyaa} with parameters $c\ge 0$ and $\xi>0$, as well as \Cref{Pi}. Then, for any $C,D\in\widetilde{\mathcal{S} }$ such that $C,D\subset \Lambda\ssubset \ZZ^d$ with $2c\,|C\cup D|\,\exp\big(  -\frac{\dist(C\backslash D,D\backslash C)}{\xi} \big)<1$, and all  $\rho\in\cD(\cH_{\Lambda})$,
\begin{align*}
D(\omega\|E_{C\cup D*}(\omega))\le \frac{1}{1-2c\,|C\cup D|\,\e^{-\frac{\dist(C\backslash D,D\backslash C)}{\xi}}}\,\Big( D(\omega\|E_{C*}(\omega))+D(\omega\|E_{D*}(\omega))\Big)\,,
\end{align*}
with $\omega:=E_{A\cap \Lambda*}(\rho)$.
\end{theorem}

We include the proof of this result in Appendix \ref{appendix}. Some other results in the same spirit have appeared in the last years in the literature of quantum systems, frequently termed as \textit{approximate factorization} \cite{capel2018quantum, BardetCapelLuciaPerezGarciaRouze-HeatBath1DMLSI-2019} or \textit{approximate tensorization} \cite{bardet2020approximate} of the relative entropy. Such approximate tensorization statements constitute the most important step in recent classical proofs of functional inequalities~\cite{cesi2001quasi}.

\subsection{Main result for $d$-dimensional systems}\label{sec:mainresult}

In this section, we state the main result of the current manuscript, namely the positivity of  the  MLSI constant of a family of Lindbladians with a specific geometry satisfying certain conditions of clustering of correlations. Before stating this theorem, we need to conceive a new geometrical argument, inspired by that of  \cite{cesi2001quasi,[D02]} by restricting the analysis to some grained sets such as the ones presented in Lemma \ref{grainedsets}. For that, we need to introduce the notion of ``subordinated grained fat rectangle'' from that of ``fat rectangle'' presented in \cite{[D02]}. 
 
\begin{definition}[Fat rectangle]
	Let $x\in  \ZZ^d$ be a site and $l_1, \ldots, l_d \in \NN$. We define the following \textit{rectangle}:
	\begin{equation}
	T(x;l_1,\ldots,l_d):=x+([1,l_1]\times \ldots \times [1,l_d])\cap \ZZ^d.
	\end{equation}
	Given a rectangle of this form, we define its \textit{size} by $\text{max}\qty{l_k \, : \, k=1,\ldots,d}$, and we say that the rectangle is \textit{fat} if
	\begin{equation}
	\text{min} \qty{l_k \, : \, k=1,\ldots,d} \geq \frac{1}{10} \text{max}\qty{l_k \, : \, k=1,\ldots,d}.
	\end{equation}
	A rectangle is denoted by $T$, and the class of rectangles of size at most $L$ is written by $\mathcal{T}_L$.  We further write 
	\begin{equation*}
	    \mathcal{T} := \underset{L \geq 1}{\bigcup} \, \mathcal{T}_L \, . 
	\end{equation*}
 \end{definition}

Now, given a rectangle $T$, we define the  \textit{grained rectangle subordinated to $T$} as the largest grained set contained in $T$, and denote it by $\widetilde{T}$. Note that for $T$ large enough,  $\widetilde{T}$ always exists and can be constructed by considering the pixels contained in $T$ and following Lemma \ref{grainedsets}. $\widetilde{T}$ is then said to be a   \textit{grained fat rectangle} if there exists a fat rectangle $T$ such that $\widetilde{T}$ is the grained set subordinated to $T$. 

We are ready to state and prove our main result:

\begin{theorem}\label{maintheorem}
		Let $\{ \Lambda\}_{\Lambda  \ssubset \ZZ^d}$ be an increasing family of fat rectangles such that $\Lambda \nearrow \mathbb{Z}^d$ and let $\widetilde{\Lambda}$ be the subordinated grained rectangle associated to each $\Lambda$. Let	$\cL:=\{\cL_{\widetilde{\Lambda}}, \cL_{\partial \widetilde{\Lambda}}\}_{\widetilde{\Lambda}}$ be a uniform family of local, primitive, reversible and frustration-free Lindbladians satisfying \eqref{LiLinftyaa}. Moreover, assume \eqref{Pi} holds. Then,
	\begin{equation*}
\underset{\widetilde{\Lambda} \nearrow \mathbb{Z}^d}{\operatorname{lim \, inf}}\,\alpha(\cL_{\widetilde{\Lambda}})>0 \,,
	\end{equation*}
	where the infimum above is taken over all families of subordinated fat grained rectangles $\widetilde{\Lambda}$.

\end{theorem}

\begin{remark}
Note that the same result would be satisfied for $\{ \overline{\mathcal{L}}_{\widetilde{\Lambda}} \}$  after fixing the boundary conditions. 
\end{remark}

In the next section, we present a simplified version of this result for 1D and 2D systems, based on a splitting of the plane into some \textit{rhomboids}, which constitute a particular and elegant case of the aforementioned subordinated grained sets. The proof of \Cref{maintheorem}, i.e. for $n$-dimensional systems, essentially follows the same steps, but needs to involve \textit{subordinated grained rectangles}, and thus presents some subtleties and more elaborate notations. Since the former is more instructive for the reader, we decide to prove it in the main text and  leave the proof of the latter to Appendix \ref{appendix:nD}.

To conclude this section, in the case of an embedded Glauber dynamics, we recover the full equivalence as a consequence of \Cref{classclust} and \Cref{maintheorem}:

\begin{corollary}\label{cor:glauberok}
	Let $\cL^{G}$ be a uniform family of embedded Glauber Lindbladians, and denote by $L^{{G}}$ their restriction to the classical algebra. Then the following conditions are equivalent:
\begin{itemize}
	\item[(i)]  $L^{G}$ satisfies \eqref{DSM} for some constant $\gamma\in (0,\infty)$.
	\item[(ii)] $L^{G}$ has a positive logarithmic Sobolev constant independent of the system size.
	\item[(iii)] $\cL^G$ has a positive $\operatorname{MLSI}$ constant independent of system size.
	\item[(iv)] $\cL^{G}$ satisfies the rapid mixing condition.
		\item[(v)] $L^{G}$ is gapped.
		\item[(vi)] $\cL^G$ is gapped.
\end{itemize}
For $\cL^S$ a uniform family of Schmidt evolutions with $2$-local interactions, the chain of implications $\eqref{LiLinftyaa} \Rightarrow(iii) \Rightarrow (iv) \Rightarrow (vi) \Rightarrow \eqref{qL2}$ holds.
\end{corollary}

\subsection{Main result for 1D and 2D systems}\label{proof2D}

In this section, we present a simplified version of \Cref{maintheorem} for 2D systems (note that 1D systems can be seen as a particular case of the 2D setting) by introducing a simpler and more visual geometry than the one appearing in the statement of the aforementioned result. For that, we need to introduce the notion of ``rhombi'' and ``rhomboids''. 

\begin{remark}
Note that, for 1D systems, we can rewrite Hamiltonians with $k$-local interactions over quantum spin chains for any $k \geq 3$ in terms of $2$-local interactions, following an argument of coarse-graining. Indeed, we could regroup the sites composing the chain in a proper way, combine their associated Hilbert spaces and rewrite the interactions so that they are $2$-local in the new framework. Therefore, the 1D case can be interpreted as a particular case of Theorem \ref{maintheorem2} below which presents the advantage of allowing for a more general condition of locality. This argument does not hold in larger dimensions, where our proof only works for the $2$-local case.  
\end{remark}

 Given a grained set $ \widetilde{S} \in \widetilde{\mathcal{S}}$ in 2D, we call it a \textit{rhombus} of size $L$ if it satisfies the following conditions (see \Cref{figuretiling1}):
 \begin{figure}[h!]
	\centering
	\includegraphics[width=0.7\linewidth]{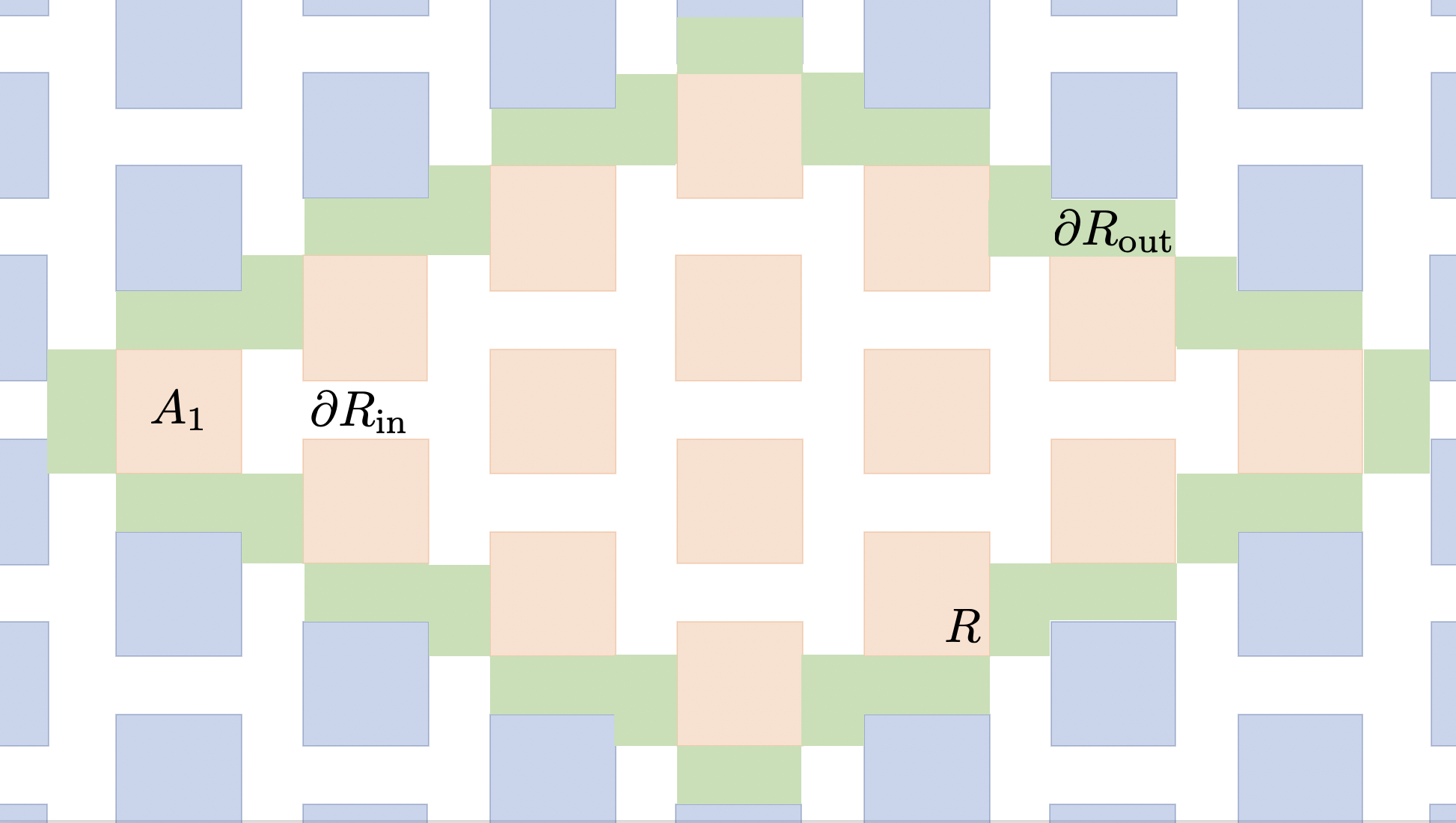}
	\caption{A Rhombus $R$ in $\ZZ^2$, its inner boundary $\partial R_{\operatorname{in}}$ and its outer boundary $\partial R_{\operatorname{out}}$.}
	\label{figuretiling1}
\end{figure}

 \begin{itemize}
     \item[(i)]  There is a unique pixel $A_1 \subset \widetilde{S}$  containing a site $x^1=(x^1_1 , x^1_2) \in  \widetilde{S} $ such that $x^1_1 \leq y_1 $ for any other site $y=(y_1, y_2) \in  \widetilde{S}$. Note that this site will not be unique, since there is a whole face of $A_1$ satisfying this condition. 
     \item[(ii)] In a second layer, there are exactly $2$ pixels $A_{2, i} \subset \widetilde{S}$  such that
     \begin{itemize}
         \item[(a)] dist$(A_1, A_{2, i}) = \kappa$ for  $i=1, 2$.
         \item[(b)] Given any site  $x^2=(x^2_1 ,  x^2_2)$ in $ A_{2,i}$ for $i=1, \ldots 2^d$, its first coordinate verifies
         \begin{equation*}
           x_1^1 + D + \kappa -1 \leq x_1^2 \leq  x_1^1 + 2D + \kappa - 2  \, ,
         \end{equation*}
     \end{itemize}
     This means that each $A_{2, i}$ is translated from $A_1$ by a vector whose first coordinate is equal to $D+\kappa-1$.
   \item[(iii)] The same construction follows recursively until layer $L$, in which there are exactly $L$ pixels $A_{L, j} \subset \widetilde{S}$ such that 
   \begin{itemize}
         \item[(a)] For any $j=1, \ldots {L-1}$, there is exactly one pixel $A_{L-1, i}$ such that    dist$(A_{L-1, i}, A_{L, j})=$  dist$(A_{L-1, i}, A_{L, j+1}) = \kappa$.
         \item[(b)] Given any site  $x^L=(x^L_1, x^L_2)$ in $ A_{L, j}$ for any $j=1, \ldots L$, its first coordinate verifies
        \begin{equation*}
          x_1^1 + (L-1)(D + \kappa-1) \leq x_1^L \leq  x_1^1 + L D - 1 + (L-1 )(\kappa-1) \, .
        \end{equation*}
    \end{itemize}
    \item[(iv)] From layer $L+1$ until layer $2L-1$, the number of pixels belonging to $\widetilde{S}$ decreases recursively in the following way: In layer $k=L+1, \ldots 2L-1$, there are exactly $2L-k$ pixels $A_{k, i} \subset \widetilde{S}$ such that 
  \begin{itemize}
       \item[(a)] For any $j=1, \ldots {L-1}$, there are exactly two pixels $A_{k-1, i-1}, \, A_{k-1, i-1}$ such that \phantom{sadasdad} dist$(A_{k-1, i-1}, A_{k, j})= $ dist$(A_{k-1, i}, A_{k, j}) = \kappa$.
       \item[(b)] Given any site  $x^k=(x^k_1 , x^k_2)$ in $ A_{k, i}$ for any $i=1, \ldots 2L-k$, its first coordinate verifies
       \begin{equation*}
          x_1^1 + (k-1)(D + \kappa-1) \leq x_k^L \leq  x_1^1 + k D - 1 + (k-1 )(\kappa-1) \, .
         \end{equation*}
     \end{itemize}
     Note with this construction that layer $2L-1$ consists of a unique pixel $A_{2L-1}$.
     \item[(v)] The set $\widetilde{S} $ also contains the intersections of the boundaries of adjacent pixels, i.e.
     \begin{equation*}
       \emptyset \neq   \partial A_{l, i} \cap \partial A_{l+1, j}  \subset \widetilde{S} \, ,
     \end{equation*}
     for every $1 \leq l \leq 2L-2$.
 \end{itemize}
 Another notion that is necessary for the geometrical construction in the main result in 2D is that of \textit{rhomboid}, namely a deformation of a rhombus as introduced above in which all the sides do not have the same length in number of pixels. Given a rhomboid with sides of length $l_1$ and $l_2$ respectively, we  call $L:= \text{max}\{l_1, l_2\}$  the \textit{size} of a rhomboid and define a \textit{fat rhomboid} as a rhomboid for which
 \begin{equation*}
     \frac{1}{10} \text{max}\{l_1, l_2\} \leq \text{min}\{l_1, l_2\} \, .
 \end{equation*}
 We  denote hereafter a rhombus or a rhomboid by $R$ and we further denote by $\mathcal{R}_L$ the set of all fat rhomboids with side at most $L$. Moreover, we take
 \begin{equation*}
     \mathcal{R} = \underset{L \geq 1}{\bigcup} \, \mathcal{R}_L \, .
 \end{equation*}
  Now, we are ready to state and prove our main result in 2D:
 
 \begin{theorem}\label{maintheorem2}
		Let $\{ R\}_{R  \ssubset \ZZ^2}$ be an increasing family of rhomboids and let	$\cL:=\{\cL_{R}, \cL_{\partial R}\}_{R}$ be a uniform family of local, primitive, reversible and frustration-free Lindbladians satisfying \eqref{LiLinftyaa}. Moreover, assume that \eqref{Pi} holds. Then, 
	\begin{equation*}
\underset{R \nearrow \mathbb{Z}^2}{\operatorname{lim \, inf}}\,\alpha(\cL_{R})>0 \,,
	\end{equation*}
	where the infimum above is taken over all increasing families of rhomboids. 
	
\end{theorem}

The main trick to the proof of \Cref{maintheorem2} can be easily summarized. First, we consider the tiling introduced above and  use Condition \ref{Pi} together with the chain rule \eqref{chainrulerelatent} in order to reduce the problem to that of proving the MLSI for the restricted class of approximately clustering states: for any state $\rho\in\cD(\cH_\Lambda)$,
	\begin{equation}\label{eqchainrule}
	D(\rho \| \sigma^\Lambda ) = D(\rho\|E_{A\cap\Lambda*}(\rho))+D(E_{A\cap \Lambda*}(\rho)\|\sigma^\Lambda)
	\end{equation}
Then, \Cref{maintheorem2} is a direct consequence of the two following results:
\begin{lemma}\label{firstterm}
Under the conditions of \Cref{maintheorem2},	there exists a constant $\alpha>0$, independent of $|\Lambda|$, such that any $\rho\in\cD(\cH_\Lambda)$, 
\begin{align*}
	4\alpha \,D(\rho\|E_{A\cap\Lambda*}(\rho))\le \operatorname{EP}_{\cL_\Lambda}(\rho)\,.
	\end{align*}
\end{lemma}

\begin{proof}
		The right hand side can be controlled assuming complete MLSI:
	\begin{align}
	D(\rho\|E_{A\cap\Lambda*}(\rho))\le \frac{1}{4\alpha_\ccc(\cL_{A\cap \Lambda})}\operatorname{EP}_{\cL_{A\cap \Lambda}}(\rho)\overset{(1)}{\le} \frac{1}{4\alpha_\ccc(\cL_{A\cap \Lambda})}\operatorname{EP}_{\cL_\Lambda}(\rho)\,,\label{mlsifirstterm}
	\end{align} 
	where $(1)$ follows by \Cref{EPexpress}. We conclude by noticing that, by the tensorization property of CMLSI  together with the fact that the size of each of the regions $A_j$ constituting $A$ is uniformly bounded, $\alpha_\ccc(\cL_A\cap\Lambda)$ is lower bounded by a positive constant $\alpha$ independent of $\Lambda$ by Theorem~\ref{CMLSIholds}. 
\end{proof}

Note that the proof of this lemma does not depend on the dimension or the geometry employed after the tiling. We further need the following theorem, to which we devote the rest of the section:

\begin{theorem}\label{secondterm}

Under the conditions of \Cref{maintheorem2},	there exists a constant $\beta>0$, independent of $|R |$, such that for all $\rho\in\cD(\cH_R)$,
\begin{align*}
4\beta D(E_{A \cap R *}(\rho)\|\sigma^R)\le \operatorname{EP}_{\cL_R}(\rho)\,.
\end{align*}
	\end{theorem}
	
Before proving \Cref{secondterm}, we briefly prove \Cref{maintheorem2} assuming \Cref{firstterm} and \Cref{secondterm}:
\begin{proof}[Proof of \Cref{maintheorem2}]
This follows directly from the use of  \Cref{firstterm} and \Cref{secondterm}  into \Cref{eqchainrule}.
	\end{proof}

Now, we turn our attention to the proof of \Cref{secondterm}. The geometric construction that we devise  is an extension of the strategy used in order to prove the result in the case of classical Gibbs samplers \cite{cesi2001quasi,[D02]}, as well as in the proof of the positivity of the spectral gap of Davies generators in \cite{[BK16]}. Roughly speaking, the idea of the proof  is to split some regions of the lattice into smaller subregions, and reduce the analysis of the MLSI constant on $\Lambda$ to that on those. As an original contribution, apart from the aforementioned rhombi and rhomboids, we introduce and use the notion of  a ``Pinched MLSI constant''.

\begin{definition}[Pinched MLSI]\label{def:pinchedmlsi}
	For any $C\subset\Lambda\ssubset \ZZ^d$, the generator $\cL_C$ satisfies a \textit{Pinched modified logarithmic Sobolev inequality} if there exists a constant $\beta>0$ such that, for all $\rho\in\cD(\cH_\Lambda)$: 
	\begin{align}\label{PinchedMLSI}
	4\beta D(E_{A\cap\Lambda*}(\rho)\|E_{C*}\circ E_{A\cap \Lambda*}(\rho))\le \operatorname{EP}_{\cL_C}(\rho)\,.
	\end{align}
	The largest constant satisfying \Cref{PinchedMLSI} is denoted by $\beta_\Lambda(\cL_C)$.
\end{definition}

\begin{remark}
Note that the notion of Pinched MLSI is introduced for any dimension, not only for dimension 2.
Moreover, it plays the analogous role in this proof to that of conditonal MLSI in \cite{BardetCapelLuciaPerezGarciaRouze-HeatBath1DMLSI-2019}, with which it would coincide if $E_{A \cap \Lambda *}(\rho) = \rho$.
\end{remark}

Now, for the first step of the proof consider a rhombus $\Lambda \subset \mathbb{Z}^2$ and  split it into $C$ and $D$ as shown in \Cref{fig:1}. Then, by virtue of the approximate tensorization for the relative entropy stated above, we can prove the following:

\begin{figure}[h!]
	\centering
	\includegraphics[width=0.7\linewidth]{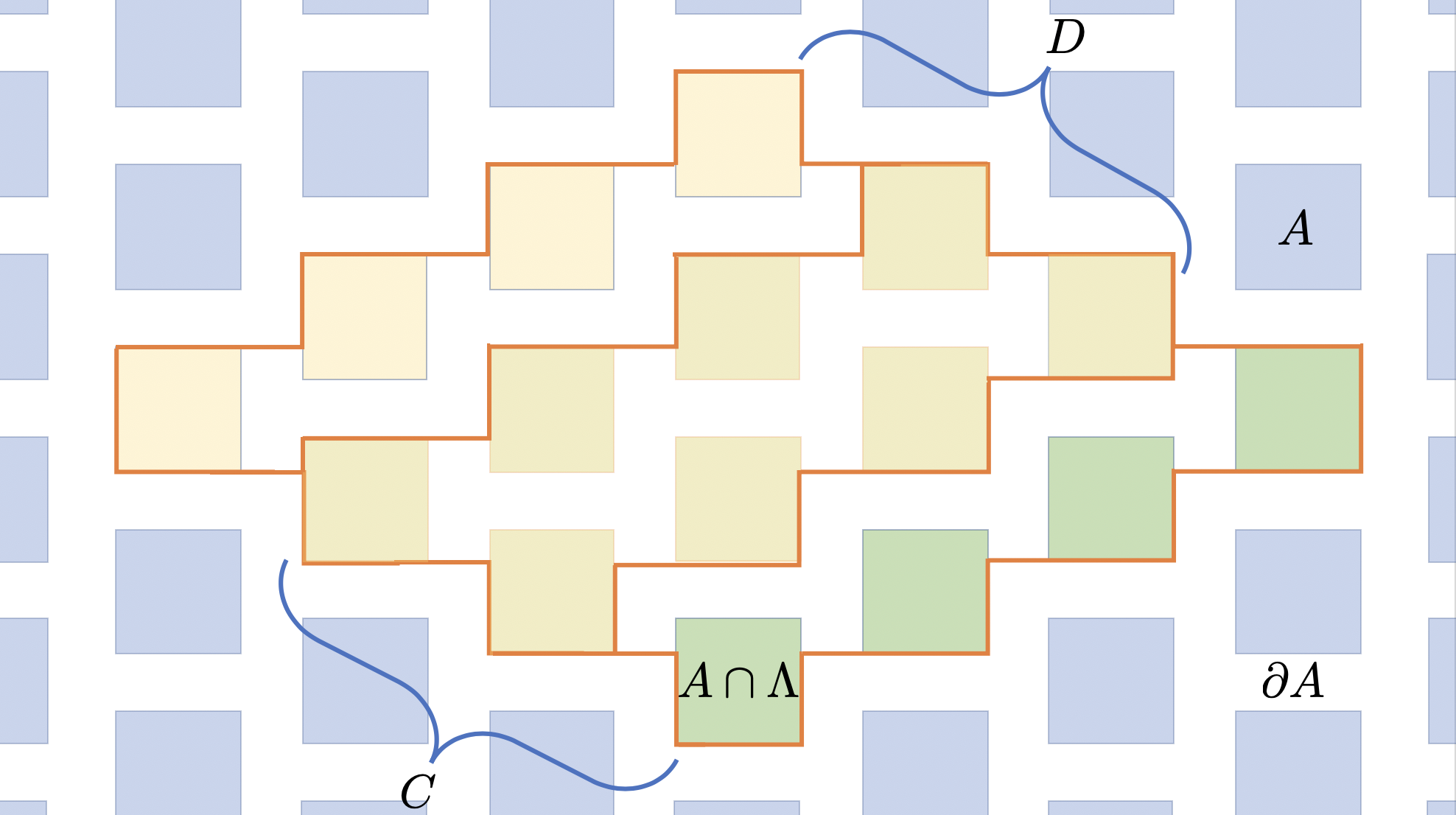}
	\caption{Splitting of a rhombus $\Lambda=C\cup D$ into rhombi $C$ and $D$.}
	\label{fig:1}
\end{figure}

\begin{step}\label{step:1}
Assuming \Cref{LiLinftyaa}, the following holds for every $\rho \in \mathcal{D}(\mathcal{H}_\Lambda)$ and $C, D \subset \Lambda$  such that $c\, |C\cup D| \, \e^{-\dist(C\backslash D,D\backslash C)/\xi}< 1/2 $ (see Figure \ref{fig:1}):
\begin{equation*}
D(E_{A\cap\Lambda*}(\rho)\|E_{C\cup D*}\circ E_{A\cap\Lambda*}(\rho))  \leq \frac{ \theta(C,D)}{4 \min \qty{\beta_\Lambda(\LL_C), \beta_\Lambda(\LL_{D})}} \left( \operatorname{EP}_{\LL_{C\cap D} }(\rho) + \operatorname{EP}_{\LL_{C \cup D}}(\rho)   \right),
\end{equation*}
where $\displaystyle \theta(C,D):= \frac{1}{1-2\, c \, |C\cup D| \, \e^{-\dist(C\backslash D,D\backslash C)/\xi} }  $.
\end{step}

\begin{proof}
Denote $\omega:=E_{A\cap\Lambda*}(\rho)$.  By  \Cref{ATAC}, we have:
\begin{equation*}
	D(\omega\|E_{C\cup D*}(\omega))\le \theta(C,D)\,\left(D(\omega\|E_{C*}(\omega))+D(\omega\|E_{D*}(\omega))  \right)\,
\end{equation*}
Now, recalling the definitions of the Pinched MLSI constants in $C$ and $D$, as given in \eqref{PinchedMLSI}, one has
\begin{align*}
D(\omega\|E_{C\cup D*}(\omega)) &
\leq \theta(C,D) \left( \frac{-\tr[\LL_{C*}(\rho)(\log\rho-\log \sigma)]}{4 \beta_\Lambda(\LL_C)}+\frac{-\tr[\LL_{D*}(\rho)(\log\rho-\log \sigma)]}{4 \beta_\Lambda(\LL_D)} \right) \\
& \leq  \frac{\theta(C,D)}{4 \min \qty{\beta_\Lambda(\LL_C), \beta_\Lambda(\LL_D)}} \left( \operatorname{EP}_{\LL_C}(\rho) + \operatorname{EP}_{\LL_D}(\rho)   \right)\\
& = \frac{ \theta(C,D)}{4 \min \qty{\beta_\Lambda(\LL_C), \beta_\Lambda(\LL_{D})}} \left( \operatorname{EP}_{\LL_{C\cap D} }(\rho) + \operatorname{EP}_{\LL_{C \cup D}}(\rho)   \right),
\end{align*}
where in the last equality we are using the fact that
\begin{equation}
\LL_{C*}(\rho)+\LL_{D*}(\rho)=\LL_{C\cup D*}(\rho)+ \LL_{C \cap D*}(\rho)
\end{equation} 
for every $\rho \in \mathcal{D}(\mathcal{H}_\Lambda)$.
\end{proof}

\begin{remark}
Note that in the next step of the proof we need to choose $C$ and $D$ carefully so that $\theta(C,D)$ satisfies some desired decaying behaviour. Indeed, we will consider $C$ and $D$ such that $|C\cup D|\sim L^d$ and dist$(C,D)= \sqrt{L}$ for a certain $L \in \mathbb{N}$, obtaining the necessary decay for $\theta(C,D)$ as a consequence of the fact that $\e^{-\sqrt{L}}$ decays faster than any polynomial.
\end{remark}

In the second step of the proof, we split a certain region of the lattice into two subregions and get a lower bound for the Pinched MLSI constant of the former in terms of the Pinched MLSI constants of the latter. For that, we construct a suitable family of fat rhomboids in the rhombus $R$ where we apply the previous step. 

Let $R$ be a rhombus of size $2L$, for $L$ large enough. We define $a_L:=\lfloor\sqrt{L}\rfloor$ and $n_L:= \lfloor \frac{L}{10a_L}  \rfloor$, where $\lfloor\cdot\rfloor$ denotes the integer part. Consider one of the faces of $R$ and enumerate the pixels of the face corresponding to this side by $A_{1,1}, \ldots , A_{1, 2L}$. For the next layer we obtain after removing that face, we also enumerate the pixels as  $A_{2,1}, \ldots , A_{2, 2L}$. Analogously we enumerate all the pixels of $R$, until the ones of the opposite face are denoted by $A_{2L,1}, \ldots , A_{2L, 2L}$. Then, for every integer $1 \leq n \leq n_L$, we define the following pair of sets of pixels:
\begin{equation}\label{AnBn}
C_{A,n}=\qty{A_{i,j} \subset R \, : \,  1 \leq j \leq L + n a_L}\,,~~~
D_{A,n}=\qty{A_{i,j} \subset R \, : \, L + (n-1)a_L < j \leq l_d}\,.
\end{equation}
and we cover $R$ with the rhomboids generated by them by including the intersection of the boundaries of adjacent pixels, which we name $C_n$ and $D_n$ respectively (see Figure \ref{fig:2}). Furthermore, it is clear by construction that $C_n$ and $D_n$ are both  fat rhomboids. Hence, for $n$ fixed, it is clear that $C_n \cap D_n \neq \emptyset$ and the shortest side of the overlap has length of order $\sqrt{L}$ pixels. 

\begin{step}\label{step:2}
There exists a positive constant $K$, independent of the size $2L$ of $R$ such that
\begin{equation}\label{ineq:L-large}
\underset{n=1,\ldots,n_L}{\min} \qty{\beta_\Lambda(\LL_{C_n}), \beta_\Lambda(\LL_{D_n})}\left( 1+\frac{\kappa}{\sqrt{L}} \right)^{-1} \leq \beta_\Lambda(\LL_{R}),
\end{equation}
for every $1 \leq n  \leq n_L$ and $L$ large enough.
\end{step}
\begin{figure}
	\centering
	\includegraphics[width=0.7\linewidth]{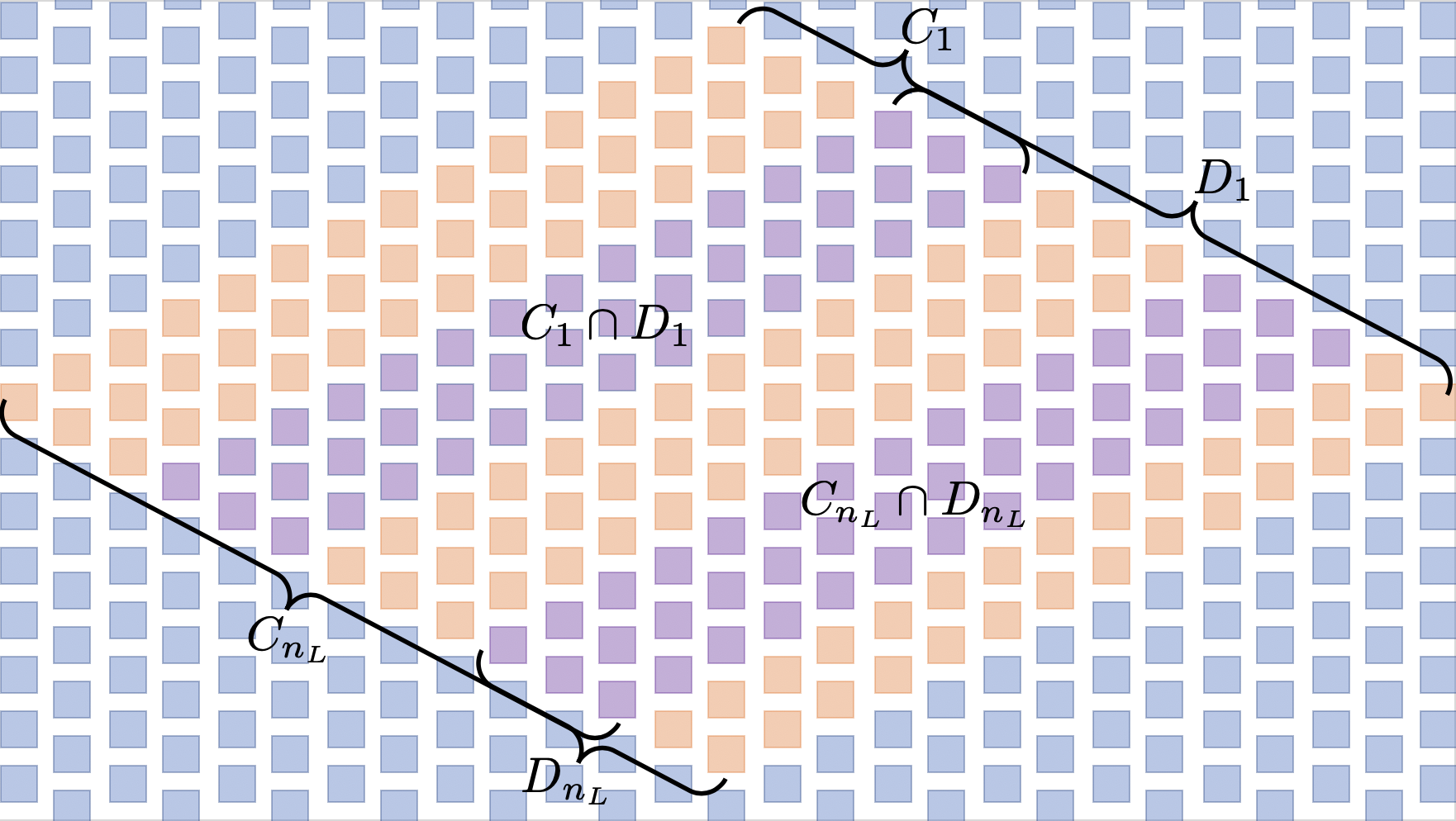}
	\caption{Splitting in $C_n$ and $D_n$.}
	\label{fig:2}
\end{figure}
\begin{proof}
	Once again, we denote $\omega:=E_{A\cap\Lambda*}(\rho)$. Then, using the sets $C_n$ and $D_n$ in the expression obtained in Step \ref{step:1}, we get, for every $1 \leq n  \leq n_L$,
\begin{equation}
D(\omega \| E_{R*}(\omega)) \leq \frac{ \theta(C_n,D_n)}{4 \min \qty{\beta_\Lambda(\LL_{C_n}), \beta_\Lambda(\LL_{D_n})}} \left( \operatorname{EP}_{\LL_{C_n \cap D_n} }(\rho) + \operatorname{EP}_{\LL_{C_n \cup D_n}}(\rho)   \right),
\end{equation}
where
\begin{equation*}
\theta(C_n,D_n)=  \frac{1}{1-2\, c \, |R| \, \e^{-\sqrt{L}/\xi}    }  \leq  \frac{1}{1-2\, \widetilde{c} \, L^d \, , \e^{-\sqrt{L}/\xi}    } 
\end{equation*}
for every $1 \leq n  \leq n_L$. Let us denote the latter by $ \theta(\sqrt{L})$. Now, by the definition of $C_n$ and $D_n$, the two following properties clearly hold:
\begin{enumerate}
\item $C_i \cap D_i \cap C_j \cap D_j = \emptyset$ for every $i \neq j$;
\item $ \ds \underset{1\leq n \leq n_L}{\bigcup} \left( C_n \cap D_n \right) \subseteq R $.
\end{enumerate}
Therefore, we can average over $n$ the previous expression to obtain:
\begin{align*}
D(\omega \| E_{R*}(\omega)) &  \leq \frac{1}{n_L} \,\underset{n=1}{\overset{n_L}{\sum}}\,\frac{ \theta(C_n,D_n)}{4 \, \min \qty{\beta_\Lambda(\LL_{C_n}), \beta_\Lambda(\LL_{D_n})}} \left( \operatorname{EP}_{\LL_{C_n \cap D_n} }(\rho) + \operatorname{EP}_{\LL_{R}}(\rho)   \right) \\
&\leq \frac{\theta(\sqrt{L})}{4 \, \underset{n=1,\ldots,n_L}{\min} \qty{\beta_\Lambda(\LL_{C_n}), \beta_\Lambda(\LL_{D_n})}} \left(\operatorname{EP}_{\LL_{R}}(\rho) + \frac{1}{n_L} \,\underset{n=1}{\overset{n_L}{\sum}} \,\operatorname{EP}_{\LL_{C_n \cap D_n} }(\rho)  \right) \\
&\leq \frac{\theta(\sqrt{L})}{4 \, \underset{n=1,\ldots,n_L}{\min} \qty{\beta_\Lambda(\LL_{C_n}), \beta_\Lambda(\LL_{D_n})}}\left( 1+\frac{1}{n_L} \right)\operatorname{EP}_{\LL_{R}}(\rho) .
\end{align*}
Hence, by the definition of $\beta_\Lambda(\LL_{R})$, we have
\begin{equation}
\frac{ \underset{n=1,\ldots,n_L}{\min} \qty{\beta_\Lambda(\LL_{C_n}), \beta_\Lambda(\LL_{D_n})}}{\theta(\sqrt{L})}\left( 1+\frac{1}{n_L} \right)^{-1} \leq \beta_\Lambda(\LL_{R}),
\end{equation}
Note that
\begin{equation*}
\theta(\sqrt{L}) \geq 1 \text{ for every }L>1 \text{ and }\underset{L \rightarrow \infty}{\lim} \theta(\sqrt{L})  = 1.
\end{equation*}
Then, for $L$ large enough, the following inequality holds:
\begin{equation}\label{ineq:L-large1'}
\underset{n=1,\ldots,n_L}{\min} \qty{\beta_\Lambda(\LL_{C_n}), \beta_\Lambda(\LL_{D_n})}\left( 1+\frac{K}{\sqrt{L}} \right)^{-1} \leq \beta_\Lambda(\LL_{R}),
\end{equation}
for $K>1$ independent of $L$.
\end{proof}

Now, let us first define the following quantities for $L>1$:
\begin{equation}
S(L):= \underset{R \in \RRR_L}{\inf} \beta_\Lambda(\LL_{R})\,.
\end{equation}
In the next step, we obtain a recursion between the quantities $S(L)$ which will later allow us to get a lower bound for the global MLSI constant in terms of size-fixed Pinched MLSI constants. 
\begin{step}\label{step:3}
There exists a positive constant $K$ independent of the size of $R$ such that
\begin{equation}
S(2L)\geq  \left( 1+\frac{K}{\sqrt{L}} \right)^{-6} S(L)\phantom{asdd} \text{for } L \text{ large enough}.
\end{equation}

\end{step}

\begin{proof}
Consider the expression obtained in the previous step. Let us analyse the value of the pinched MLSI constant in the rhomboids $C_n$ and $D_n$. Let us consider the grained rectangle $C_n$ (the analysis is analogous for $D_n$). One side of $C_n$ (the one corresponding to the direction of the cut) has length less than or equal to $1.2L$, by definition of $C_n$, whereas the other side has length $2L$. Then, we apply again the previous step, now along this side, and obtain then:
\begin{equation}\label{ineq:L-large1}
\underset{n=1,\ldots,n_L}{\min} \qty{\beta_\Lambda(\LL_{C_n}), \beta_\Lambda(\LL_{D_n})} \geq \left( 1+\frac{K}{\sqrt{L}} \right)^{-1} S \left( 1.2 L \right),
\end{equation}
since now both sides of each rhomboid have size less than or equal to $1.2L$. Therefore,
\begin{equation}
\beta_\Lambda(\LL_{R})\geq \left( 1 + \frac{K}{\sqrt{L}} \right)^{-2} S \left(1.2 L \right)\,,
\end{equation}
and since the rhombus that we were considering in Step \ref{step:2} verified  $R \in \RRR_{2L}$, we obtain
\begin{equation}
S(2L) \geq \left( 1 + \frac{K}{\sqrt{L}} \right)^{-2} S \left( 1.2 L \right).
\end{equation}
To conclude, we iterate this expression one more time to obtain
\begin{equation}
S(2L) \geq \left( 1 + \frac{K}{\sqrt{L}} \right)^{-2}  \left( 1 + \frac{K}{\sqrt{0.6 L}} \right)^{-2} S \left( 0.72 L \right),
\end{equation}
and since $ \displaystyle S\left( 0.72 L \right)\geq S (L)$, we obtain
\begin{equation}
S(2L) \geq \left( 1+ \frac{K}{\sqrt{L}} \right)^{-4} S(L),
\end{equation}
where $\displaystyle K$ is a constant independent of the size of the system.
\end{proof}

Finally, in the last step of the proof, using recursively the relation obtained in the previous one, we get a lower bound for the global MLSI constant in terms of complete MLSI constants.  Similar to above, we define the quantities for $L>1$:
\begin{equation}
T(L):= \underset{R \in \RRR_L}{\inf} \alpha_\ccc(\LL_{\widetilde{R}})\,.
\end{equation}
\begin{step}
There exists a constant $L_0\in\NN$, independent of $\Lambda$ such that the following holds:
\begin{equation*}
\alpha( \LL_\Lambda)  \geq  \Psi (L_0) \,T(L_0)\,,
\end{equation*}
where $\Psi (L_0)$ does not depend on the size of $\Lambda$.
\end{step}

\begin{proof}
Let us denote by  $L_0$ the first integer for which inequality (\ref{ineq:L-large1}) holds. By virtue of the previous step, it is clear that the following holds for $L_0$:
\begin{equation}
S(2L_0) \geq \left( 1+ \frac{K}{\sqrt{L_0}} \right)^{-4} S(L_0),
\end{equation}
Note now that the limit of $\Lambda$ tending to $\ZZ^d$ is the same as the one of $S(n L_0)$ with $n$ tending to infinity. Therefore,
 \begin{align*}
  \underset{\Lambda \rightarrow \ZZ^2}{\lim} \, \alpha(\LL_\Lambda) & =  \underset{n \rightarrow \infty}{\lim}  \, S(2^n L_0) \\
  & \geq   \left( \underset{n=1}{\overset{\infty}{\prod}} \left( 1+ \frac{K}{\sqrt{2^{n-1}L_0}} \right) \right)^{-4} S(L_0) \\
 &  \geq \left( \exp \left[ \underset{n=0}{\overset{\infty}{\sum}} \, \frac{K}{2^n L_0} \right] \right)^{-4} S(L_0) \\
 & = \exp \left[\frac{-4 K}{L_0}(2+\sqrt{2}) \right] S(L_0),
 \end{align*}
where the constants $L_0$ and $K$ do not depend on the size of $\Lambda$. We conclude from the following simple observation that 
\begin{align*}
S(L_0)=\underset{R\in\RRR_{L_0}}{\inf}\beta_\ccc(\LL_{R})\overset{(1)}{\ge}
\underset{R\in \RRR_{L_0}}{\inf} \alpha_\ccc(\LL_{R})=T(L_0)\,.
\end{align*}
where $(1)$ follows from \Cref{betatoalpha} below.

\end{proof}

\begin{lemma}\label{betatoalpha}
	For any rhomboid $R \subseteq \Lambda\ssubset \ZZ^2$, 
	\begin{align*}
	\beta_\Lambda(\cL_{R})\ge \alpha_{\ccc}(\cL_{R})\,.
	\end{align*}
\end{lemma}

\begin{proof}
	The proof follows by a simple use of the data processing inequality: for all $\rho\in\cD(\cH_\Lambda)$
	\begin{align*}
	D(E_{A\cap \Lambda*}(\rho)\|E_{R*}\circ E_{A\cap \Lambda*}(\rho))&\overset{(1)}{=}	D(E_{A\cap \Lambda*}(\rho)\|E_{A\cap \Lambda*}\circ E_{R *}(\rho))\\
	&\overset{(2)}{\le} D(\rho\|E_{R *}(\rho))\\
	&\le \frac{1}{4\alpha_{\ccc}\,(\cL_{R})}\,\operatorname{EP}_{\cL_{R}}(\rho)\,.
	\end{align*}
	Above, $(1)$ follows from \Cref{Pi}, and $(2)$ from the use of the data processing inequality for the channel $E_{A\cap \Lambda*}$. 
\end{proof}

\section{Applications}\label{sec:applications}

\subsection{Optimization on noisy quantum annealers}

Our results show that a quantum system coupled to a classical environment at high enough temperature can only hold information for a time scaling logarithmically in the number of qubits and inversely proportional to the MLSI constant without error correction, as one can show that various capacities of the underlying channels decay with the MLSI~\cite{Muller_Hermes_2015,muller2016relative,bardet2019group}.
A related question is for how long a noisy quantum device can sustain computations that cannot be done in polynomial time classically.

It has been recently shown in~\cite{2009.05532} that MLSIs can be used to estimate this for the annealing model of quantum computation. Although their results can be immediately applied given the results on MLSI obtained here, in this section we give another proof strategy to reach similar conclusions based on quantum optimal transport inequalities~\cite{Rouz__2019,gao2018fisher,Carlen20171810}.

In this section we will consider a quantum annealer coupled to a classical environment and model the thermal noise through embedded Glauber dynamics.
We will consider general interaction graphs $G=(V,E)$ with $|V|=n$ vertices and local dimension $d=2$. It is easy to check that the conditions of Theorem~\ref{maintheorem} are always satisfied at high enough temperatures. We start with the following proposition of \cite{2009.05532}:
\begin{proposition}[Theorem 1 in \cite{2009.05532}]\label{prop:relentdecay}
Let $\cS_{s*}(\rho)=-i[H_s,\rho]+r\cL_{*}(\rho)$ be a time dependent Lindbladian  such that the primitive Lindbladian $\cL_*$ satisfies $\operatorname{MLSI}$ with constant $\alpha(\cL)>0$ and corresponding invariant state $\sigma$ and $r>0$ is the noise rate. Then the evolution from time $0$ to $t$ under $\cS_{s*}$, $\mathcal{T}_{t*}$, satisfies:
\begin{align}\label{equ:entroadiabatic}
D(\mathcal{T}_t(\rho)\|\sigma)\leq e^{-4\alpha r t}D(\rho\|\sigma)+\int\limits_{0}^t e^{-4\alpha r(t-\tau)}\|\sigma^{-\frac{1}{2}}[H_\tau,\sigma]\sigma^{-\frac{1}{2}}\|_{\infty}d\tau \, .
\end{align}
\end{proposition}
This estimate is particularly useful in settings in which $[\sigma,H_t]\simeq0$ for large times. This is the case for current implementations of quantum annealers under noise driving the system to a classical state.
This is because such annealers aim at preparing a ground state $\ket{\psi_G}$ of the classical Hamiltonian
\begin{align*}
H_1=\sum_{i\sim j} a_{i,j}Z_i Z_j+\sum_{i}b_iZ_i
\end{align*}
adiabatically for some graph $G=(V,E)$ and coefficients $a_{i,j},b_i$. For simplicity we will assume that the interaction graph of the thermal state and of $H_1$ are the same.
Such annealers start by preparing the ground state of
\begin{align}\label{equ:definitionH0}
    H_0=- \sum_i\gamma_i X_i \, 
\end{align}
for $\gamma_i>0$. This is easily seen to be the state $\ket{\psi_0}=\ket{+}^{\otimes n}$, where $n$ is the number of qubits. By letting 
\begin{align}\label{equ:definitionsHs}
    H_s=g_1(s)H_1+g_0(s)H_0 \, ,
\end{align}
where $g_0,g_1$ are smooth functions of $s$ such that $g_1(1)=1,g_1(0)=1$ and $g_0(1)=0,g_0(0)=1$, we converge to $\ket{\psi_1}$ by evolving slowly enough with $H_s$.

For example, if we now assume that the noisy annealer is such that it is affected by some time-independent Lindbladian with a classical fixed point  $\sigma$ and we take the adiabatic path $H_s=\big(1-\tfrac{s}{T}\big)H_0+\frac{s}{T}H_1$ for some large $T$. Then we obtain from Proposition~\ref{prop:relentdecay} that:
\begin{align}\label{equ:boundentropyising}
D(\mathcal{T}_T(\ketbra{+}{+}^{\otimes n})\|\sigma)\leq e^{-4\alpha rT}D(\ketbra{+}{+}^{\otimes n}\|\sigma)+
 \|\sigma^{-\frac{1}{2}}[H_0,\sigma]\sigma^{-\frac{1}{2}}\|_\infty\,
  \frac{(1-4e^{-4\alpha r T}r\alpha T-e^{-4\alpha r T})}{16r^2\alpha^2T}\,,
\end{align}
where we used that $[\sigma,H_1]=0$, since $\sigma$ is assumed to be a classical state. That is, in the adiabatic limit $T\to\infty$, the output of the device is actually given by $\sigma$. The main idea of~\cite{2009.05532} is now to use Equation~\eqref{equ:boundentropyising} to obtain finite time bounds when the energy of the output is approximated well by that of a Gibbs state that can be sampled from in polynomial time. 

We will now show another approach to classical simulability based on transportation-entropy inequalities. Let us recall some notions of quantum optimal transport.

Recall that embedded Glauber dynamics are GNS symmetric and, thus, satisfy the assumptions of Theorem~\ref{thm:normalformCM}.
Thus, given $\cL:=\overline{\cL}^G_\Lambda$, one possible definition of the \emph{Lipschitz constant} of an observable $X$ is:
\begin{align}\label{lipnorm}
\|X\|_{\operatorname{Lip}}:= \left(  \sum_{j\in\mathcal{J}} c_j (\e^{-\omega_j/2}+\e^{\omega_j/2})\|\partial_jX\|_{\infty}^2\right)^{1/2},
\end{align}
where $\partial_jX=[\tilde{L}_j,X]$ and the constants $\omega_j,c_j$ come from the normal form of Theorem~\ref{thm:normalformCM} as well. Based on that, it is also possible to define the \emph{non-commutative 1-Wasserstein distance} of two states $\rho,\sigma$ as:
\begin{align*}
    W_{1,\cL}(\rho,\sigma)=\sup\limits_{X=X^\dagger\in\cB(\cH),\, \|X\|_{\operatorname{Lip}}\leq1}\left|\operatorname{Tr}\left(X(\rho-\sigma)\right)\right|.
\end{align*}
We refer to~\cite{Rouz__2019} for more details on these quantities and note that we adopt a different normalization of the Lipschitz constant which is more convenient for ourpurposes, meaning that our constant is $2^{\frac{n}{2}}$ times larger. It immediately follows from the definitions that for any observable $X$:
\begin{align}\label{equ:Lipschitz}
    \left|\operatorname{Tr}\left(X(\rho-\sigma)\right)\right|\leq \|X\|_{\operatorname{Lip}} \, W_{1,\cL}(\rho,\sigma) \, .
\end{align}
Moreover, for our purposes it is particularly important to recall the connection between MLSI inequalities and the Wasserstein distances. This is given through a transport-entropy inequality, as proved in the quantum case in~\cite{Rouz__2019}. If the semigroup generated by $\cL$ satisfies a MLSI inequality with constant $\alpha$, then one can show that~\cite[Theorem 3, Theorem 4]{Rouz__2019}:
\begin{align}\label{equ:transportation-cost}
    W_{1,\cL}(\rho,\sigma)\leq \sqrt{\frac{1}{\alpha}\,D(\rho\|\sigma)} \, .
\end{align}
This inequality can be seen as a strengthening of Pinsker's inequality for observables with small Lipschitz norm.

Putting all of these elements together we conclude that:

\begin{proposition}\label{prop:issue_ising}
Let $H_s$ be defined as in Equation~\eqref{equ:definitionsHs} for $0\leq s\leq1$ and $\cS_{s*}(\rho)=-i[H_s,\rho]+r\,\overline{\cL}^{{G}}_{V *}(\rho)$ be a time dependent Linbladian  such that $\overline{\cL}^{{G}}_{V *}$ is an embedded classical Glauber dynamics satisfying the conditions of Corollary~\ref{cor:glauberok} and converging to a classical Gibbs state $\sigma=e^{-\beta H_V}/\tr\left[e^{-\beta H_V}\right]$ and $r>0$. Moreover, let $\mathcal{T}_{t*}$ be the evolution from time $0$ to $t$ under $\cS_{s*}$.
Then there is a constant $\alpha>0$ such that:
\begin{align*}
    \left|\operatorname{Tr}\left(H_1(\mathcal{T}_{t*}(\ketbra{+}{+}^{\otimes n})-\sigma)\right)\right|\leq \alpha^{-\frac{1}{2}}\,\|H_1\|_{\operatorname{Lip}}\,R(t)^{\frac{1}{2}} \, ,
\end{align*}
with
\begin{align*}
   R(t)=e^{-4r\alpha T}D(\ketbra{+}{+}^{\otimes n}\|\sigma)+\|\sigma^{-\frac{1}{2}}[H_0,\sigma]\sigma^{-\frac{1}{2}}\|_\infty\,\frac{(1-4e^{-4r\alpha T}\alpha r T-e^{-4\alpha rT})}{16r^2\alpha^2T}\,. 
\end{align*}

\end{proposition}
\begin{proof}
As noted in Equation~\eqref{equ:Lipschitz}, we have:
\begin{align*}
    \left|\operatorname{Tr}\left(H_1(\mathcal{T}_{t*}(\ketbra{+}{+}^{\otimes n})-\sigma)\right)\right|\leq \|H_1\|_{\operatorname{Lip}}W_{1,\overline{\cL}_V^G}(\mathcal{T}_{t*}(\ketbra{+}{+}^{\otimes n}),\sigma) \, .
\end{align*}
As Corollary~\ref{cor:glauberok} implies that $\overline{\cL}_V^G$ satisfies a MLSI with system-size dependent constant $\alpha>0$, and $\overline{\cL}_V^G$ is GNS symmetric, we have that:
\begin{align*}
    W_{1,\overline{\cL}_V^G}(\rho,\sigma)\leq \alpha^{-\frac{1}{2}} \sqrt{D(\mathcal{T}_{t*}(\ketbra{+}{+}^{\otimes n})\|\sigma)} \, .
\end{align*}
Moreover, the MLSI also allows us to conclude that Equation~\eqref{equ:boundentropyising} holds, from which the claim follows.
\end{proof}

Let us discuss the bound above in a bit more detail. Assume that the graph has maximal degree $\kappa$.
Then, it is not difficult to see by a Taylor expansion that $\|\sigma^{-\frac{1}{2}}[H_0,\sigma]\sigma^{-\frac{1}{2}}\|_\infty=\cO(\beta\kappa n)$. Moreover, 
\begin{align*}
   \|H_1\|_{\operatorname{Lip}}=\cO(\kappa^2 \sqrt{n}) \, .
\end{align*}
To see this, note that each $\tilde{L}_j$ for the embedded Glauber dynamics will act on $\kappa+1$ qubits. Thus, only $\cO(\kappa)$ terms of $H_1$ will not commute with $\tilde{L}_j$, giving the bound. 
To see the scaling of the relative entropy bound, note that $e^{-4\alpha rT}D(\ketbra{+}{+}^{\otimes n}\|\sigma)=\cO(e^{-4\alpha rT}(\beta+1)n)$.

Thus, putting all these estimates together we conclude that for such models we have:
\begin{align*}
    \left|\operatorname{Tr}\left(H_1(\mathcal{T}_{T*}(\ketbra{+}{+}^{\otimes n})-\sigma)\right)\right|&\leq \|H_1\|_{\operatorname{Lip}}W_{1,\overline{\cL}_V^G}(\mathcal{T}_{T*}(\ketbra{+}{+}^{\otimes n}),\sigma)\\
    &=\mathcal{O}\left( \kappa^2(\beta+1)e^{-2\alpha r T}n+\beta\kappa^2 nr^{-1}T^{-\frac{1}{2}}\right) \, .
\end{align*}
In particular, whenever the term $\beta\kappa^2 nr^{-1}T^{-\frac{1}{2}}$ is dominant, then we conclude that the energy density of the output of the noisy annealer is essentially the one of the fixed point up to an error of $\epsilon \kappa n$ for times $\cO( \kappa \left(r\epsilon\right)^{-2})$. Whenever the term $\kappa^2(\beta+1)e^{-2\alpha r T}n$ is dominant, that is, for very high temperatures, we reach the same conclusion for times $\cO(r^{-1}\log(\kappa \epsilon^{-1}))$. 

Using the estimates of~\cite{2009.05532} instead of transportation methods in the $\beta=\Theta(1)$ regime, we obtain from our modified logarithmic Sobolev inequality that the output of the annealer is comparable with that of a polynomial time Gibbs sampler (not necessarily the fixed point of the noise) at times $\cO( \kappa r^{-2}\epsilon^{-1})$. Thus, our stronger statement that the energy of the output of the device is comparable with that of the fixed point of the evolution comes at the expense of a quadratically worse dependency on the error.

Although current implementations are coupled to environments at low temperatures, the fact that they also suffer from control errors translates to additional local depolarizing noise. This will drive the system to a high entropy/temperature state.
Therefore, we expect that the fixed point of the evolution is a high temperature Gibbs state if we combine the depolarizing noise with the thermal noise and the former is the dominant source of noise. Thus, as long as the system also suffers from local depolarizing noise on top of the thermal noise, we expect the conditions of Proposition~\ref{prop:issue_ising} to be fulfilled.

\subsection{Concentration inequalities and eigenstate thermalization}
Another application of our results is to derive Gaussian concentration inequalities for the outcome distribution of observables of high-temperature Gibbs states. Such inequalities have been the subject of many recent works~\cite{anshu_concentration_2016,kuwahara2020gaussian,kuwahara_connecting_2016,depalma2020quantum} and can be applied to obtain the equivalence of thermodynamical ensembles~\cite{brandao_equivalence_2015,kuwahara2020gaussian}. Here we will show how to obtain versions of such results that apply to a significantly larger class of observables from a MLSI. Let us first introduce some notation to discuss these. Given some observable $O\in\cB(\cH_{\Lambda})_{\operatorname{sa}}$ with eigendecomposition
\begin{align*}
    O=\sum\limits_io_i\ketbra{w_i}{w_i}
\end{align*}
and a quantum state $\rho$, we will let $\langle O\rangle_{\rho}:=\tr[O\rho]$. Moreover, for $r\in\R$ we denote by $\Pi_{\langle O\rangle+r}$ the projection onto the eigenspace of eigenvalues greater or equal to $\langle O\rangle+r$, i.e.
\begin{align*}
    \Pi_{\langle O\rangle_\rho+r}=\sum_{i:o_i\geq \langle O\rangle_\rho+r}\ketbra{w_i}{w_i}.
\end{align*}
We then have~\cite[Theorem 8]{Rouz__2019}:
\begin{lemma}
Let $\{\Phi(X)\}_{X\ssubset \ZZ^d}$ be a potential, and assume that the corresponding Gibbs states $\{\sigma^{\Lambda}\}_{\Lambda\ssubset\ZZ^d}$ are stationary states of a primitive uniform family of Lindbladians $\cL$ with $\alpha\equiv \alpha(\cL)>0$ and associated Lipschitz norm $\|\cdot\|_{\operatorname{Lip},\cL}$. Then for all $O\in\cB(\cH)_{sa}$:
\begin{align}\label{equ:concentration_ineq}
    \tr[\sigma^{\Lambda} \Pi_{\langle O\rangle_{\sigma^{\Lambda}}+r}\left(O-\langle O\rangle_{\sigma^{\Lambda}}I\right)]\leq \operatorname{exp}\left( -\frac{\alpha r^2}{8\|\Delta^{-\frac{1}{2}}(O)\|_{\operatorname{Lip},\cL}}\right).
\end{align}
\end{lemma}
The inequality in~\eqref{equ:concentration_ineq} improves upon the state-of-the-art~\cite[Corollary 1]{kuwahara2020gaussian}, as it holds for arbitrary orbservables, while previous results had to impose a locality structure for $O$. That being said,~\cite{kuwahara2020gaussian} establishes concentration bounds for a significantly larger class of Gibbs states, including potentials with long range interactions, while we only know the existence of a MLSI for the examples discussed in our main result Theorem~\ref{maintheorem}.

Armed with the transportation cost inequality in Equation~\eqref{equ:transportation-cost}, we can also obtain a simple proof of the eigenstate thermalization hypothesis~\cite{srednicki_chaos_1994} under a suitable hypothesis on the density of states at a given energy level. Recall that the eigenstate thermalization hypothesis states that eigenstates of the Hamiltonian are locally indistinguishable from a thermal state. Denote by $\{\ket{E_m^{\Lambda}}\}$ the eigenstates of $H_{\Lambda}=\sum_{X\subset \Lambda}\Phi(X)$ for a potential $\{\Phi(X)\}_{X\ssubset \ZZ^d}$ with corresponding energy $E_m$. Morover, denote by $f_\beta(m)=\frac{e^{-\beta E_m}}{\mathcal{Z}_\beta}$ the density of states at energy $E_m$ and inverse temperature $\beta$.
We then have:
\begin{proposition}[Eigenstate thermalization (ETH) from tranportation cost]\label{prop:eigenstate}
Let $\{\Phi(X)\}_{X\ssubset \ZZ^d}$ be a potential, and assume that the corresponding Gibbs states at inverse temperature $\beta$ $\{\sigma^{\beta,\Lambda}\}_{\Lambda\ssubset\ZZ^d}$ are stationary states of a primitive uniform family of Lindbladians $\cL$ with $\alpha\equiv \alpha(\cL)>0$ and associated Lipschitz norm $\|\cdot\|_{\operatorname{Lip},\cL}$. Then for all $O\in\cB(\cH)_{sa}$:
\begin{align}
    \tr[(\sigma^{\beta,\Lambda}-\ketbra{E_m}{E_m})O]\leq\|O\|_{\operatorname{Lip},\cL} \sqrt{\frac{\log(f_\beta(m)^{-1})}{\alpha}}. 
\end{align}
In particular, for $n=|\Lambda|$ and $\|O\|_{\operatorname{Lip},\cL} =\cO(n^{-\frac{1}{2}})$ and $m$ such that $\log(f_\beta(m)^{-1})=o(n)$ we have that:
\begin{align*}
    \lim\limits_{\Lambda\nearrow \mathbb{Z}^d}\tr[(\sigma^{\beta,\Lambda}-\ketbra{E_m}{E_m})O]=0.
\end{align*}
\end{proposition}
\begin{proof}
The proof follows immediately from Equation~\eqref{equ:transportation-cost} by noting that $D(\ketbra{E_m}{E_m}\|\sigma^{\beta,\Lambda})=\log(f_\beta(m)^{-1})$.
\end{proof}
Note that for averages of local observables of the form $O=|\Lambda|^{-1}\sum_{i=1}^{|\Lambda|}X_i$ with $X_i$ acting on a constant number of  sites we have $\|O\|_{\operatorname{Lip}} =\cO(n^{-\frac{1}{2}})$. Thus, Proposition~\ref{prop:eigenstate} establishes the ETH  for all states covered in Theorem~\ref{maintheorem} and for all Lipschitz observables, which includes averages over subsytems. Also note that Proposition~\ref{prop:eigenstate} is much more general: it asserts that, for any family of states $\tau^{\Lambda}$ such that $D(\tau^{\Lambda}\|\sigma^{\beta,\Lambda})=o(n)$, local averages coincide with those of the thermal state. Once again, when compared with other results in a similar direction~\cite{kuwahara2020eigenstate} the advantage of our approach is that it applies to a larger family of observables (Lipschitz versus local). Moreover, we do not only obtain the statement for eigenstates, but any state with small enough relative entropy, as remarked before. In contrast, in~\cite{brandao_equivalence_2015} the authors derived results of similar flavour, but needed the relative entropy scaling  as $o(n^{\frac{1}{d+1}})$ and not single eigenstates. However, once again these results hold for a much more general class of states than those covered by Theorem~\ref{maintheorem}.
\subsection{Strong converses in quantum hypothesis testing}

Quantum hypothesis testing concerns the problem of discriminating between two different quantum states\footnote{It is often referred to as {\em{binary}} quantum hypothesis testing, to distinguish it from the case in which more than two states are being tested.}. This task is of paramount importance in quantum information theory, since many other tasks can be reduced to it. In the language of hypothesis testing, one considers two hypotheses -- the {\em{null hypothesis}} $\mathbb{H}_0 : \rho$ and the {\em{alternative hypothesis}} $\mathbb{H}_1 :\sigma$, where $\rho$ and $\sigma$ are two quantum states. In an operational setting, say Bob receives a state $\omega$ with the knowledge that either $\omega=\rho$ or $\omega=\sigma$. His goal is then to infer which hypothesis is true, i.e., which state he has been given, by means of 
a measurement on the state he receives. The measurement is given most generally by a POVM $\{T, \Id - T\}$ where $0\le T\le \Id$. Adopting the nomenclature from classical hypothesis testing, we refer to $T$ as a {\em{test}}. 
The probability that Bob correctly guesses the state to be~$\rho$ is then 
equal to $\tr[T\rho]$, whereas his probability of correctly guessing
the state to be $\sigma$ is $\tr[(\Id-T)\sigma]$. Bob can erroneously infer the state to be $\sigma$ when it is actually $\rho$ or vice versa. The corresponding error probabilities are referred to as the {\em{type I error}} and {\em{type II error}} respectively. They are denoted as follows:
\begin{align} \alpha(T) &:= \tr\left[( \Id -  T)\rho\right], \quad \beta(T) := \tr\left[T \sigma\right], 
\end{align}
where $\alpha(T)$ is the probability of accepting $\mathbb{H}_1$ when 
$\mathbb{H}_0$ is true, while $\beta(T)$ is the probability
of accepting $\mathbb{H}_0$ when $\mathbb{H}_1$ is true. Obviously, there is a trade-off between the two error probabilities, and
there are various ways to jointly optimize them, depending on whether or not 
the two types of errors are treated on an equal footing. Here, we are concerned with the setting of {\em{asymmetric hypothesis testing}}, in which one minimizes the type II error under a suitable constraint on the type I error.

Quantum hypothesis testing was originally studied in the {\em{asymptotic i.i.d.~setting}} in which Bob is provided not with just a single copy of the state but with multiple (say $n$) 
identical copies of the state, say $\rho^{\otimes n}$ or $\sigma^{\otimes n}$, where $\rho$ and $\sigma$ are states on a finite dimensional Hilbert space $\cH$, and he is allowed to do a joint measurement on all these copies. The optimal exponential decay rate of the type II error
under the assumption that the type I error remains bounded. This is given by \textit{Stein's lemma} and its refinements \cite{HP91,ON00}: for any $\eps\in(0,1)$, 
\begin{align*}
-\lim_{n \to \infty} \frac{1}{n} \ln \min_{0\le T_n\le \Id_n}\left\{  \beta(T_n):\,\alpha(T_n)\le \eps    \right\}=D(\rho\|\sigma)\,.
\end{align*}
Extensions of this result to non i.i.d. settings, such as Gibbs states on lattice spin systems of different potentials, were also considered \cite{jakvsic2012quantum,mosonyi2008asymptotic,hiai2007large,hiai2008error,datta2016second}.

Stein's lemma only holds when the size of the region being tested goes to infinity. However, in a more practical situation, one might be interested in getting estimates on the errors made when a finite number $n$ of copies are available. This is the so-called \textit{finite blocklength regime}. 

The following result is adapted from \cite{rouze2019functional} (see also \cite{beigi2018quantum}):

\begin{theorem}\label{theostrongconvnoniid}[\cite{rouze2019functional}, Theorem 13.1.6]
	Let $\{\Phi(X)\}_{X\ssubset \ZZ^d}$ be a potential, and assume that the corresponding Gibbs states $\{\sigma^{\Lambda}\}_{\Lambda\ssubset\ZZ^d}$ are stationary states of a primitive uniform family of Lindbladians $\cL$ with $\alpha\equiv \alpha(\cL)>0$. Next, let $\{\rho^\Lambda\}_{\Lambda\in\ZZ^d}$ be another family of full-rank states such that
		\begin{align}\tag{$\star$}\label{stare}
	\sup_{\Lambda\ssubset  \ZZ^d}\sup_{t\ge 0}\,\frac{1}{|\Lambda|t}\,D_{\max}(\e^{t\overline{\cL}_{\Lambda*}}(\rho^{\Lambda})\|\rho^{\Lambda})<\gamma <\infty\,,
	\end{align}
for some $\gamma>0$, where $D_{\max}(\rho\|\sigma):=\ln\|\sigma^{-\frac{1}{2}}\rho\sigma^{-\frac{1}{2}}\|_\infty$ is the max-relative entropy between a state $\rho$ and a full-rank state $\sigma$. Then, for any subregion $\Lambda\ssubset \ZZ^d$ and any test $0\le T_\Lambda\le \Id_\Lambda$:
\begin{align}\label{eq-statenoniid}
	-\frac{1}{|\Lambda|}\,\ln \tr[\rho^{\Lambda} T_\Lambda ] &\leq \frac{1}{|\Lambda|}\,D(\sigma^{\Lambda}\|\rho^{\Lambda}) + \frac{2}{\sqrt{|\Lambda|}} 
	\sqrt{{ \frac{\gamma}{4\alpha} \ln\frac{1}{\tr[\sigma^\Lambda T_\Lambda]}}}    -\frac{1}{4\alpha\,|\Lambda|}\ln\tr[\sigma^\Lambda T_\Lambda]\,.
\end{align}
\end{theorem}	

Building on the main result of this manuscript, we can prove the following result in the line of the last theorem. 

\begin{corollary}
Let $\Phi^{(1)}:=\{\Phi^{(1)}(X)\}_{X\ssubset \ZZ^d}$ and $\Phi^{(2)}:=\{\Phi^{(2)}(X)\}_{X\ssubset \ZZ^d}$ be two local commuting potentials, and assume that the Gibbs states $\{\sigma_1^\Lambda\}_{\Lambda\ssubset\ZZ^d}$ corresponding to $\Phi^{(1)}$ are fixed points of a uniform family of Lindbladians $\cL$ such that $\alpha(\cL)>0$. Then the two corresponding families $\{\sigma^\Lambda\equiv\sigma^\Lambda_1\}_{\Lambda\ssubset \ZZ^d}$ and $\{\rho^\Lambda\equiv \sigma^\Lambda_2\}_{\Lambda\ssubset \ZZ^d}$ of Gibbs states satisfy \eqref{eq-statenoniid} for some $\gamma,\alpha>0$. 
\end{corollary}
\begin{proof}
The result follows after proving \Cref{stare}, i.e.:
\begin{align*}
 \ln\|(\sigma^\Lambda_2)^{-\frac{1}{2}}\e^{t\overline{\cL}^S_{\Lambda*}}(\sigma^\Lambda_2)(\sigma^{\Lambda}_2)^{-\frac{1}{2}}\|_\infty\le \gamma\,|\Lambda|\, t\,.
\end{align*}
This is done simply by defining the tilted generator $\cL'_\Lambda:=\Gamma^{-1}_{\sigma_2^{\Lambda}}\circ \overline{\cL}^S_{\Lambda*}\circ \Gamma_{\sigma_2^{\Lambda}}$, so that the left-hand side above is equal to
\begin{align*}
    \ln\|\e^{t\cL_\Lambda'}(\Id)\|_\infty=\ln\|\e^{t\cL'_\Lambda}:\cB(\cH_\Lambda)\to\cB(\cH_\Lambda)\|\le t\|\cL_\Lambda':\cB(\cH_\Lambda)\to\cB(\cH_\Lambda)\|\,.
\end{align*}
Now, since $\overline{\cL}_\Lambda^S$ is local, and since $\sigma_2^{\Lambda}$ arises from a commuting potential, the generator $\cL_\Lambda'$ itself is local, and its norm is upper bounded by:
\begin{align*}
    \|\cL_\Lambda':\cB(\cH_\Lambda)\to\cB(\cH_\Lambda)\|\le \gamma \,|\Lambda| \, ,
\end{align*}
for some $\gamma>0$. The result follows.

\end{proof}

Thus, it follows from our main result Theorem~\ref{maintheorem} that we are able to establish strong converses for hypothesis testing between arbitrary classical and commuting Gibbs states at high enough temperature.

\subsection{Efficient local quantum Gibbs samplers for nearest neighbour potentials}
The Schmidt generators define a family of local Lindbladians converging to the Gibbs state of a commuting potential. In~\cite{Kliesch_2011}, the authors show how to simulate the evolution under a family of local Lindbladians using a unitary circuit whose size is polynomial in the system size and evolution time. Combining this with our main result, we conclude that in the rapidly mixing regime, the Schmidt generators give rise to efficient \emph{local} circuits for the preparation of Gibbs states corresponding to commuting potentials on a quantum computer. This is in contrast to~\cite{brandao2019finite}, where the authors show how to efficiently prepare such states with gates whose locality scales logarithmically with system size.

In this subsection we briefly explain how to efficiently obtain the Schmidt generators with nearest neighbour interactions.  First, recall that the Schmidt generators are defined in terms of the conditional expectations onto minimal fixed point algebra generated by the neighbourhood Schmidt span of each site $k$, 
$${\cF}^S_k:=\Id_k\,\otimes \,\cA_{k,\operatorname{out}}\,\otimes \, \underset{j \in \Lambda \setminus \{k\} \partial }{\bigotimes}\cB(\cH_{j}) \,.$$
Thus, it suffices to obtain an explicit decomposition of such algebras to also obtain the corresponding Schmidt generators. 

However, this decomposition was already presented in Equation \eqref{equ:decomp} and, in~\cite{Murota_2010}, the authors showed how to efficiently find a unitary $U\in\cB(\cH_{k\partial})$ and subspaces $\cH_{jj},\cH_{ji},\cH_{jk}$ for any $j \in \partial k$ and $i$ adjacent to $j$ with $i \in \Lambda \setminus k$, such that
\begin{align*}
    \cA_{k, \operatorname{out}}=U \left( \bigotimes_{j\in\partial k}\,\bigotimes_{\substack{i\in V\backslash \{k\}\\(i,j)\in E}}\,\bigoplus_{\alpha_j}\Id_{\cH_{jj}^{\alpha_j}}\,\otimes\,\cB(\cH_{ji}^{\alpha_j})\,
  \otimes
  \bigotimes_{\substack{i'\in \{k\}\\(j,i')\in E}} \Id_{ \cH_{ji'}^{\alpha_{j}}} \right) U^\dagger \, ,
\end{align*}
given the list of generators. With this at hand, we conclude that it is possible to efficiently find the description of the local Schmidt conditional expectations:
\begin{proposition}[Efficient preparation of Gibbs states of commuting potentials]
Consider $\Phi:=\{\Phi(X)\}_{X\ssubset \ZZ^d}$ a $2$-local potential satisfying \eqref{qCA}.
Then we can efficiently find a circuit of local quantum channels of depth $\mathcal{O}(\ln(|\Lambda|)\epsilon^{-1})$ that prepares a quantum state that is $\epsilon$ in trace distance to $\sigma^{\Lambda}$. 
\end{proposition}
\begin{proof}
By \Cref{maintheorem} and \Cref{propDScorrdecay}, the family $\cL^S$ of Schmidt generators corresponding to the $2$-local potential $\{\Phi(X)\}_{X\ssubset \ZZ^d}$ satisfies MLSI with constant $\alpha(\cL^S)>0$. This implies that, in particular, that the family $\cL^S$ is rapidly mixing. As discussed above, we can efficiently find the local terms of $\cL^S$ given the 2-local potential.
The result then follows from~\cite{Kliesch_2011}, where the authors show how to simulate the evolution $e^{t\cL^S_{*}}$ for times $\mathcal{O}(\ln(|\Lambda|)\epsilon^{-1})$ with a unitary circuit of depth $\mathcal{O}(\ln(|\Lambda|)\epsilon^{-1})$.
\end{proof}

\section{Conclusion}
In this work, we have substantially advanced the understanding of modified logarithmic Sobolev inequalities for quantum spin systems. For the first time, we have an unconditional proof of a system-size independent MLSI for classically interacting dynamics of quantum systems at high enough temperatures. Moreover, our results would give the first system-size independent MLSI for quantum models. 

In addition to that, our results further advance the program of relating dynamical and statistical properties of quantum Gibbs states at high temperatures and adds further connections to the zoo of correlation measures in quantum Gibbs states, as displayed in Figure~\ref{figurezoo}.
Finally, the different flavors of our novel applications to fundamental problems in quantum statistical physics, quantum information and computation demonstrate the strength of our results, in particular when combined with optimal transport methods.

Several research directions remain open. It would be interesting to investigate how to generalize our results to models beyond $2$-body interactions and to obtain similar results for Davies and Heat-bath generators. Some equivalences of notions of clustering of correlations also remain open.

\paragraph{Acknowledgements} The authors are grateful to Li Gao, Kohtaro Kato and Peter Johnson for useful discussions. AC is partially supported by an MCQST Distinguished PostDoc fellowship. CR is partially supported by a Junior Researcher START Fellowship from the MCQST.  CR and AC acknowledge funding by the Deutsche Forschungsgemeinschaft (DFG, German Research Foundation) under Germany's Excellence Strategy EXC-2111 390814868. DSF was supported by VILLUM FONDEN via the QMATH Centre of Excellence under Grant No. 10059.

\bibliographystyle{abbrv}
\bibliography{library}

\appendix
\section{Properties of conditional expectations}\label{appendixcondexp}
In this section we gather some properties of conditional expectations we require for our results. Let us first recall that a conditional expectation satisfies the following useful properties, whose proofs can be found in \cite{Aspects2003}:

\begin{proposition}\label{propositioncondexp1}
	Let $\cM\subset \cN$ be a von Neumann subalgebra of $\cN$ and $E:\cN\to\cM$  a conditional expectation with respect to $\sigma$ of $\cN$ onto $\cM$. Then, $E$ satisfies the following properties:
	\begin{itemize}
		\item[(i)] The map $E$ is completely positive and unital.
		\item[(ii)] For any $X\in\cN$ and any $Y,Z\in\cM$, $E[YXZ]=Y E[X]Z$.
		\item[(iii)] $E$ is self-adjoint with respect to the scalar product $\langle .,\,.\rangle_\sigma$. In other words:
		\begin{align*}
			\Gamma_\sigma\circ E=E_*\circ\Gamma_\sigma   \,,
		\end{align*}
		where $E_*$ denotes the adjoint of $E$ with respect to the Hilbert-Schmidt inner product and the map $\Gamma_\sigma $ is given by $\Gamma_\sigma (X) = \sigma^{1/2} X \sigma^{1/2} $ for every observable $X$.
	\end{itemize}
\end{proposition}

We recall the definition of the modular automorphism group $\left\lbrace \Delta_\sigma^{i s} \right\rbrace_{s \in \mathbb{R}}$, where each element of the group is given by the following operator:
\begin{equation*}
 \Delta_\sigma^{is} (X) := \sigma^{is} X \sigma^{- is} \, ,
\end{equation*}
for every observable $X$. Then, a conditional expectation with respect to $\sigma$ satisfies the following properties concerning these operators. 

\begin{proposition}\label{propositioncondexp}[See \cite{Takesaki_1972}]
Let $E$ be a conditional expectation  with respect to $\sigma$ of $\cN$ onto $\cM$. Then, $E$ commutes with the modular automorphism group of $\sigma$, i.e. for any $s\in\RR$,
		\begin{align}
			\Delta_\sigma^{is}\circ E=E\circ \Delta^{is}_\sigma\,.
		\end{align}
		
Moreover, given a von Neumann subalgebra $\cM\subset \cN$ and a faithful state $\sigma$, the existence of a conditional expectation $E$ is equivalent to the invariance of $\cM$ under the modular automorphism group $(\Delta_\sigma^{is})_{s\in\RR}$.  In this case, $E$ is uniquely determined by $\sigma$.
\end{proposition}

In order to derive an expression for the conditional expectations in the case of Gibbs states of commuting Hamiltonians, we need to delve into the theory of commuting $C^*$-algebras. The tools that will be used in this section are by no means new to quantum information theory \cite{bravyi2003commutative,6108194,schuch2011complexity}. 

Let $\sigma^\Lambda:=\frac{\e^{-\beta\,H_\Lambda}}{Z}$ be the Gibbs state corresponding to the commuting Hamiltonian $H_\Lambda$ at inverse temperature $\beta$. By the quantum Hammersley-Clifford theorem \cite{brown2012quantum},  given any three regions $A-B-C$ of $\Lambda$, where $B$ shields $A$ away from $C$, $\sigma_{ABC}$ is a quantum Markov chain, which means that there exists a decomposition of the Hilbert space $\cH_B:=\bigoplus_{i\in I_B}\,\cH^{i}_{B_{\operatorname{in}}}\otimes\cH^{i}_{B_{\operatorname{out}}}$ such that 
\begin{align}\label{decomqmc}
\sigma^\Lambda:=\bigoplus_{i\in I_B}\mu(i)\,\sigma_{AB^{i}_{\operatorname{in}}}\otimes\sigma_{B^{i}_{\operatorname{out}}C}\,,
\end{align}
for some probability mass function $\mu$. In fact, to each geometrical decomposition $\Lambda=A-B-C$, there is in general more than one possible decomposition of the state $\sigma_\Lambda$ as in \Cref{decomqmc} (see \cite{jenvcova2006sufficiency,luczak2014quantum,johnson2015general}). To each of these decompositions one can associate a $*$-algebra 
\begin{align*}
\cF_A:=\bigoplus_{i\in I_B}\Id_{{AB_{\operatorname{in}}^{i}}}\otimes\cB(\cH_{B_{\operatorname{out}}^{i}C})\,.
\end{align*}
One can readily verify that any such constructed algebra $\cF_A$ is invariant under the action of the modular group corresponding to $\sigma^\Lambda$: 
$$\Delta_{\sigma^{\Lambda}}(\cF_A)=\cF_A \, .$$ 
It follows from \Cref{propositioncondexp} that there exists a conditional expectation $E_A:\cB(\cH_\Lambda)\to\cF_A$, so that
\begin{align}\label{eq:condexpA}
E_A[X]:=\bigoplus_{i\in I_B}\,\Id_{{AB_{\operatorname{in}}^{i}}}\otimes\tr_{\cH_{AB_{\operatorname{in}}^{i}}}\left[ P_{i}XP_{i}\left(\Id_{{B_{\operatorname{out}}^{i}C}}\otimes\sigma^{i}_{AB_{\operatorname{in}}^{i}}\right)\,\right]\,,
\end{align}
where $\{P_{i}\}_{i\in I_B}$ is the set of minimal projections of $\cF_A$. Among all the algebras that are invariant under $\Delta_{\sigma^\Lambda}$, there exists a maximal one, call it $\cF_A^{\max}$ \cite{carlen2020recovery}. It was shown by some of the present authors that this algebra coincides with the set of fixed points of both the Heat-bath and Davies generators on the subregion $A$ being considered \cite{bardet2020approximate}.

\begin{lemma}[Compatibility of conditional expectations]\label{lemm1}
Let $E_1,E_2$ be two conditional expectations on $\cB(\cH)$ with respect to the same state $\sigma$, and assume that $\cF(E_2)\subset \cF(E_1)$, i.e. $E_2\circ E_1=E_1\circ E_2=E_2$. Then, assuming the following block decomposition of the smallest algebra: for $\cH:=\bigoplus_j P_j\cH\equiv \bigoplus_j \cH_j\otimes \cK_j $,
\begin{align*}
\cF(E_2):=\bigoplus_j\,\cB(\cH_j)\otimes \Id_{\cK_j}\,,
\end{align*}
$E_1|_{P_j\cB(\cH)P_j}=\id_{\cB(\cH_j)}\otimes E_1^{(j)}$, for some conditional expectations $E^{(j)}$ on $\cB(\cK_j)$.
\end{lemma}
\begin{proof}
Write $\sigma:=\bigoplus_j\,\tr_{\cK_j}(P_j\sigma P_j)\otimes \tau_j$. First, we observe that $E_1[P_jXP_j]=E_1[P_j^2XP_j^2]=P_jE_1[P_jXP_j]P_j$, since $P_j\in \cF(E_2)\subset \cF(E_1)$. Therefore, each subalgebra $\cB(\cH_j\otimes \cK_j)$ is preserved by the map $E_1$. Moreover, since $E_1\circ E_2=E_2$, we have that $\cB(\cH_j)\otimes \Id_{\cK_j}\subset \cF(E_1)$, so that $\cF(E_1)=\cB(\cH_j)\otimes  \cN_j$ for some subalgebra $\cN_j$ of $\cB(\cK_j)$. The result follows.

\end{proof}

\section{Proof of Theorem \ref{thm:analyticity-high-temperature}}\label{ap:Analyticity}

In this appendix, we prove that the condition of analyticity after measurement (cf. Definition \ref{def:analyticity_after_measurement}) holds above a critical temperature. The proof of this result follows the steps of that of Theorem 20 in \cite{harrow2020classical}, although we include it here for sake of completeness. 

More explicitly, given a geometrically-local Hamiltonian $H$,  $\delta >0$ and $N \geq 0$ with $\norm{N}_\infty = 1$, here we prove that there exists a constant $c$ such that: 
\begin{equation*}
    \left| \ln \left( \tr[\operatorname{e}^{- \sum_{X \subset \Lambda} z_X \Phi(X)} N] \right)\right|  \leq c \abs{\Lambda}   \; \phantom{asdda} \forall z_X \in \mathbb{C}, \abs{z_X - \beta} \leq \delta \, .
\end{equation*}
First, we fix $0< \delta <\frac{1}{5 \operatorname{e} g h \kappa }$  and $N \geq 0$ with $\norm{N}_\infty = 1$, and define 
\begin{equation*}
    g_\mathbf{z} (\Lambda) :=  \tr \left[\operatorname{e}^{- \sum_{X \subset \Lambda} z_X \Phi(X)} N \right] \, .
\end{equation*}
Then, we aim to show that: 
\begin{equation*}
    \left| \ln(g_\mathbf{z} (\Lambda)) \right| \leq O(\abs{\Lambda}) \, .
\end{equation*}
More specifically, as stated in Theorem \ref{thm:analyticity-high-temperature}, we show below that given $0< \delta <\frac{1}{5 \operatorname{e} g h \kappa }$, for $\beta_c=\frac{1}{5 \operatorname{e} g h \kappa } - \delta$  we have that for all $\beta \in [0,\beta_c)$, the function $\mathbf{z}\mapsto \ln(g_\mathbf{z}(\Lambda))$ is analytic and bounded in modulus by $( \operatorname{e}^2 g h (\beta + \delta) + \ln(d) ) \abs{\Lambda}$.

For the proof of this theorem, we need to make the following reduction. Given a lattice $\Lambda \in \mathbb{Z}^d$, consider a sequence of sublattices $\Lambda_0 \subset \Lambda_1 \subset ... \subset \Lambda_n = \Lambda$ such that each sublattice $\Lambda_{j}$ has one fewer vertex than $\Lambda_{j+i}$ and $\Lambda_0 = \emptyset$. Then, we can write
 \begin{equation*}
    g_\mathbf{z} (\Lambda) =  \underset{j=0}{\overset{\abs{\Lambda}-1}{\prod}} \left( \frac{g_\mathbf{z} (\Lambda_{j+1})}{g_\mathbf{z} (\Lambda_{j})}  \right) \, ,
\end{equation*}
and we can conclude if we  show
 \begin{equation}\label{eq:analyticity}
\left|\ln \left(  \frac{g_\mathbf{z} (\Lambda_{j+1})}{g_\mathbf{z} (\Lambda_{j})}  \right)\right| \leq O(1) \, .
\end{equation}
This fact is proven in two steps, the construction of a cluster expansion for the analyticity after measurement function and its use to show the required bound by induction in the number of sites of the lattice.

\paragraph{Step 1: Cluster expansion for the analyticity after measurement function.}

Before presenting the cluster expansion, we need to introduce some preliminary notions that will be required for the next results. 

\begin{definition}[Connected sets,  Definition 23 in \cite{harrow2020classical}]
Given $x_0 \in \Lambda$,  a collection of sublattices $\mathcal{X}=(X_1, \ldots , X_k)$ is called a \emph{connected set containing }$x_0$ with size $\mathcal{X}=k$ if the following holds:
\begin{itemize}
    \item[(i)] There exist $\kappa, R >0$ such that $1 \leq |X_i| \leq \kappa$ and diam$(X_i) \leq R$ for every $1 \leq i \leq k$.
    \item[(ii)] For any $X_i$ in $\mathcal{X}$, there exists at least one $X_j$ in  $\mathcal{X}$ such that $X_i \neq X_j$, $X_i \cap X_j \neq \emptyset$ and $x_o \in \cup_{i=1}^k X_i$.
\end{itemize}
We denote by $\abs{\mathcal{X}}$ the size of $\mathcal{X}$ and by supp$(\mathcal{X})$ all the sites it contains. 
\end{definition}

We will also make use of the following two lemmata, extracted from \cite{kliesch2014locality}.

\begin{lemma}\label{lemma:number-connected-sets}
    The number of connected sets $\mathcal{X}$ of size $\abs{\mathcal{X}}$ containing a site $x_0 \in \Lambda$ is upper bounded by $g^{\abs{\mathcal{X}}}$.   
\end{lemma}
We adopt the notation of~\cite{kliesch2014locality} and for a graph $G=(V,E)$ and $\mathcal{X}\subset E$ let $\mathcal{X}^*$ be the set of subsets of connected edges of $E$ containing $\mathcal{X}$.
\begin{lemma}\label{lemma:bound-W}
Given $(V,E)$ a finite graph and $H\geq 0$, for any $\mathcal{X} \subset E$ we have
    \begin{equation*}
        \underset{w \in \mathcal{X}^* \, : \, \mathcal{X} \subset w}{\sum} \, \frac{\abs{\lambda H}^{\abs{w}}}{\abs{w} !} = (e^{\abs{\lambda H}} - 1)^{\abs{\mathcal{X}}} \, .
    \end{equation*}
\end{lemma}

With this at hand, we can present the following cluster expansion for the function $g_\mathbf{z}(\Lambda)$, which constitutes the analogue in our setting to Lemma 26 of \cite{harrow2020classical}.

\begin{lemma}[High temperature expansion]\label{lemma:cluster-expansion}
    For any $x_0 \in \Lambda$, the function $g_\mathbf{z}(\Lambda)$ admits the following decomposition for $\displaystyle \beta \leq \frac{1}{g h (e-1)} - \delta$:
    \begin{equation*}
        g_\mathbf{z}(\Lambda) =  g_\mathbf{z}(\Lambda \backslash\{ x_0\}) +  \sum_{\substack{\mathcal{X} \, : \, x_0 \in \mathcal{X}\\\mathcal{X} \text{ is connected}}} W_{\mathbf{z}}(\mathcal{X}) \, g_\mathbf{z}(\Lambda \setminus \operatorname{supp}(\mathcal{X})) \, , 
    \end{equation*}
    where $W(\mathcal{X}) $ is defined as
    \begin{equation*}
        W_{\mathbf{z}}(\mathcal{X}) = \underset{p= \abs{\mathcal{X}}}{\overset{\infty}{\sum}} \; \frac{1}{p!} \left( \sum_{\substack{(X_1, \ldots , X_p)\\\forall1\leq  i \leq p, \; X_i \in \mathcal{X}\\ \mathcal{X} = \cup_{i=1}^p \{ X_i \}}}   \tr_{\operatorname{supp}(\mathcal{X})} \left[ \underset{j=1}{\overset{p}{\prod}} \left( - z_{X_j} \Phi(X_j) \right) \right] \right)\, .
    \end{equation*}
\end{lemma}

\begin{proof}
We only present a sketch of the proof, since it resembles that of   Lemma 26 in \cite{harrow2020classical}. First, we write a Taylor expansion of the exponential in $g_{\mathbf{z}}(\Lambda)$:
\begin{align*}
    g_{\mathbf{z}}(\Lambda)& = \tr \left[ \underset{k=0}{\overset{\infty}{\sum}} \frac{1}{k!} \; \left( - \underset{X \subset \Lambda}{ \sum} z_X \Phi(X) \right)^k N \right] \\
    & = \tr \left[ \underset{k=0}{\overset{\infty}{\sum}} \; \frac{1}{k!} \left( - \underset{X \subset \Lambda \setminus x_0}{ \sum} z_X \Phi(X) \right)^k N \right]+  \tr \left[ \underset{l=1}{\overset{\infty}{\sum}} \; \frac{1}{l!} \sum_{\substack{(X_1, \ldots , X_l)\\\forall 1\leq  i \leq l, \; X_i \subset \Lambda\\ \exists X_i \, : \, x_0 \in X_i   }}  \underset{j=1}{\overset{l}{\prod}} \left( - z_{X_j} \Phi(X_j) \right)  N \right] \\
    & = g_{\mathbf{z}(\Lambda \setminus x_0)}+  \underset{p=\abs{\mathcal{X}} , \, q=0}{\overset{\infty}{\sum}} \; \binom{p+q}{p}
  \frac{1}{(p+q)!}  \sum_{\substack{(X_1, \ldots , X_p)\\\forall 1\leq  i \leq p, \; X_i \in \mathcal{X}\\  \mathcal{X} = \cup_{i=1}^p \{ X_i \} }}  \tr_{\operatorname{supp}(\mathcal{X})} \left[ \underset{j=1}{\overset{l}{\prod}} \left( - z_{X_j} \Phi(X_j) \right)   \right] \\
  & \phantom{asdasdasdasdasdasdasdasasdasdadad} \times \sum_{\substack{(X_{p+1}, \ldots , X_{p+q})\\   X_{p+1}\cap\operatorname{supp}(\mathcal{X}) = \emptyset }} \tr\left[ \underset{j=1}{\overset{l}{\prod}} \left( - z_{X_j} \Phi(X_j) \right)  N  \right] \, ,
\end{align*}
where we have used the fact that the first term in the second line is the Taylor expansion of $g_{\mathbf{z}}(\Lambda \setminus x_0)$ and a simplification of the second term of the same line by partitioning each sequence into a connected set that contains $x_0$ and the rest. Note that the coefficient in the third line can be rewritten as $\frac{1}{p! \, q!}$. Thus, the term in the last line along with $\frac{1}{q!}$ represent the Taylor expansion of $g_{\mathbf{z}}(\Lambda \setminus \operatorname{supp}(\mathcal{X}) )$. Therefore, by defining
$ W_{\mathbf{z}}(\mathcal{X})$ as in the statement of the Lemma, we obtain:
\begin{equation*}
   g_\mathbf{z}(\Lambda) =  g_\mathbf{z}(\Lambda \setminus x_0) +  \sum_{\substack{\mathcal{X} \, : \, x_0 \in \mathcal{X}\\\mathcal{X} \text{ is connected}}} W_{\mathbf{z}}(\mathcal{X}) \, g_\mathbf{z}(\Lambda \setminus \operatorname{supp}(\mathcal{X})) \, .
\end{equation*}
Now, we need to find the regime for $\beta$ for which there is absolute convergence of the series above. First, we can obtain from Lemma \ref{lemma:bound-W} the following upper bound for $W_{\mathbf{z}}(\mathcal{X})$:
\begin{align}\label{eq:bound_modulus_Wz}
 \abs{W_{\mathbf{z}}(\mathcal{X})} \leq d^{\abs{\operatorname{supp}(\mathcal{X})}}  (\operatorname{e}^{(\delta + \beta)h}-1)^{\abs{\mathcal{X}}} \, ,
\end{align}
where we have used the fact that 
\begin{equation*}
    \abs{z_{{X}}} \leq \abs{z_{{X}}- \beta + \beta} \leq \delta + \beta \, ,
\end{equation*}
for every $X \in \mathcal{X}$. Then, by virtue of Lemma \ref{lemma:number-connected-sets}, we see that
\begin{equation*}
  \sum_{\substack{\mathcal{X} \, : \, x_0 \in \mathcal{X}\\\mathcal{X} \text{ is connected}}} \abs{W_{\mathbf{z}}(\mathcal{X})} \, \abs{g_\mathbf{z}(\Lambda \setminus \operatorname{supp}(\mathcal{X}))}    \leq d^{\abs{\Lambda}} \operatorname{e}^{g h (\delta + \beta) \abs{\Lambda}} \underset{\abs{\mathcal{X}}=1}{\overset{\infty}{\sum}}g^{\abs{\mathcal{X}}} (\operatorname{e}^{(\delta  + \beta)h}-1)^{\abs{\mathcal{X}}} \, ,
\end{equation*}
  where we have used Hölder's inequality to  bound $\abs{g_\mathbf{z}(\Lambda \setminus \operatorname{supp}(\mathcal{X}))} \leq \operatorname{e}^{g h (\delta + \beta) \abs{\Lambda}}$ and the fact that $d^{\abs{\operatorname{supp}(\mathcal{X})}} \leq d^{\abs{\Lambda}}$. To conclude, it is clear that the right-hand side of this inequality is finitely bounded whenever
  \begin{equation*}
     g  (e^{(\delta + \beta)h}-1) < 1 \, ,
  \end{equation*}
  which along with $e^x - 1 \leq (e-1) x$ implies
  \begin{equation*}
     \beta < \frac{1}{g h (e-1)} - \delta \, .
  \end{equation*}
\end{proof}

\paragraph{Step 2: Application of the cluster expansion by induction in the number of sites.} 
In the second step, we use the cluster expansion obtained in the previous one to provide a bound on $g_\mathbf{z}(\Lambda)$. We do this by induction in the number of sites of the lattice. First, we need to introduce the following technical result, which is completely analogous to Lemma 27 of \cite{harrow2020classical}.

\begin{lemma}\label{lemma:nice-bound-analyticity}
    In the same conditions that Theorem \ref{thm:analyticity-high-temperature}, the following bound holds  for $\displaystyle \beta \leq \frac{1}{5 \operatorname{e} g h \kappa } - \delta$:
    \begin{equation*}
         \sum_{\substack{\mathcal{X} \, : \, x_0 \in \mathcal{X}\\\mathcal{X} \text{ is connected}}} (e^{(\delta + \beta)h}-1)^{\abs{\mathcal{X}}}  \operatorname{e}^{g h \operatorname{e}^2 (\delta + \beta) \abs{\operatorname{supp}(\mathcal{X})} } \leq \operatorname{e} (\operatorname{e}-1) g h (\delta + \beta) \, .
    \end{equation*}
\end{lemma}

Now, we can proceed to the proof of Equation \eqref{eq:analyticity}.

\begin{proof}[Proof of Theorem \ref{thm:analyticity-high-temperature}.]

This proof follows that of Theorem 20 in \cite{harrow2020classical}, up to some slight modifications. As mentioned above, it is enough to show the bound in Equation \eqref{eq:analyticity}. More specifically, given $x_0 \in \Lambda$, here we prove by induction that the following bound holds:
\begin{equation}\label{eq:bound-induction}
   \left| \ln \left( \frac{g_\mathbf{z} (\Lambda)}{g_\mathbf{z} (\Lambda \setminus x_0)}  \right) \right|\leq \operatorname{e}^2 g h (\beta + \delta) \, , \phantom{adasdad}\forall \beta \leq \frac{1}{5 \operatorname{e} g h \kappa } - \delta \, . 
 \end{equation}
 We set by convention that $g_\mathbf{z}(\emptyset) =1$ and assume that the following  holds for every sublattice $X$ strictly contained in $\Lambda$ and containing $x_0$:
 \begin{equation}
     \left| \ln \left( \frac{1}{d^{\abs{\text{supp}(X)}}} \frac{g_\mathbf{z} (\Lambda \setminus X )}{g_\mathbf{z} (\Lambda \setminus x_0)}  \right) \right|\leq \operatorname{e}^2 g h (\beta + \delta) \abs{\text{supp}(X)} \, ,
 \end{equation}
 for  $\beta \leq \frac{1}{5 \operatorname{e} g h \kappa } - \delta$. Then, by Lemma \ref{lemma:cluster-expansion}, the following holds:
 \begin{equation*}
    \frac{g_\mathbf{z}(\Lambda)}{g_\mathbf{z}(\Lambda \setminus x_0)} \leq 1 +  \sum_{\substack{\mathcal{X} \, : \, x_0 \in \mathcal{X}\\\mathcal{X} \text{ is connected}}} W_\mathbf{z}(\mathcal{X}) \left(  \frac{g_\mathbf{z}(\Lambda \setminus \operatorname{supp}(\mathcal{X}))}{g_\mathbf{z}(\Lambda \setminus x_0)} \right) \, ,
 \end{equation*}
 and using the induction hypothesis, Equation \eqref{eq:bound_modulus_Wz}, the inequality $\abs{\ln \abs{1 + \xi}} \leq - \ln ( 1 - \abs{\xi})$ and following the lines of Theorem 20 of \cite{harrow2020classical}, we get
 \begin{align*}
     \left| \ln \left(  \frac{1}{d}\frac{g_\mathbf{z}(\Lambda)}{g_\mathbf{z}(\Lambda \setminus x_0)}  \right) \right| & \leq - \ln \left( 1-  \sum_{\substack{\mathcal{X} \, : \, x_0 \in \mathcal{X}\\\mathcal{X} \text{ is connected}}}  (\operatorname{e}^{(\beta + \delta) h} - 1)^{\abs{\mathcal{X}}} \operatorname{e}^{\operatorname{e}^2 g h (\beta + \delta) \abs{\operatorname{supp}(\mathcal{X})}} \right) \\
     & \leq - \ln \left( 1- \operatorname{e} (\operatorname{e}-1) g h (\delta + \beta)  \right)  \\
     & \leq \operatorname{e}^2 g h (\delta + \beta)  \, ,
 \end{align*}
 where we have used Lemma \ref{lemma:nice-bound-analyticity} in the second line and the inequality $-\ln ( 1- (\operatorname{e}-1) y ) \leq \operatorname{e} y$, which holds $\forall y \in [0,1]$. 

\end{proof}
\section{Proof of Theorem \ref{ATAC}}\label{appendix}

In this appendix, we prove the result of approximate tensorization stated in Theorem  \ref{ATAC}. More specifically, we prove under the assumption of \eqref{LiLinftyaa} with parameters $c\ge 0$ and $\xi>0$ as well as \Cref{Pi}, that for any $C,D\in\widetilde{\mathcal{S} }$ such that $C,D\subset \Lambda\ssubset \ZZ^d$ with $2c\,|C\cup D|\,\exp\big(  -\frac{\dist(C\backslash D,D\backslash C)}{\xi} \big)<1$ and all  $\rho\in\cD(\cH_{\Lambda})$, the following holds
\begin{align*}
D(\omega\|E_{C\cup D*}(\omega))\le \frac{1}{1-2c\,|C\cup D|\,\e^{-\frac{\dist(C\backslash D,D\backslash C)}{\xi}}}\,\Big( D(\omega\|E_{C*}(\omega))+D(\omega\|E_{D*}(\omega))\Big)\,,
\end{align*}
with $\omega:=E_{A\cap \Lambda*}(\rho)$.
The proof is similar to that of \cite[Theorem 1]{bardet2020approximate}. Its first step consists in showing the following bound:
	\begin{equation}\label{eq:quasitensor-lnM}
	D(\omega\|E_{C \cup D *}( \omega ) ) \leq D(\omega \|E_{C  *}( \omega )) + D(\omega\|E_{ D *}( \omega )) + \ln \tr[M] \, ,
	\end{equation}
where $M:= \text{exp}[- \ln E_{C \cup D *}( \omega ) + \ln E_{C *}( \omega ) + \ln E_{ D *}( \omega )  ]$. Indeed, it is clear that the difference of relative entropies can be expressed as:
\begin{align*}
 	D(\omega \|E_{C \cup D *}( \omega ) )- D(\omega \|E_{C  *}( \omega )) -& D(\omega\|E_{ D *}( \omega ))  \\
 	&= \tr\left[ \omega \left( - \ln \rho \underbrace{- \ln E_{C \cup D *}( \omega ) + \ln E_{C *}( \omega ) + \ln E_{ D *}( \omega )}_{\ln M}  \right)\right] \\
 	& = - D(\omega \| M) \, .
 	\end{align*} 	
 	Now, since $\tr[M]\neq 1$ in general, from the non-negativity of the relative entropy of two states it follows that:
	\begin{equation*}
	D(\omega \| M) \geq  -\log \tr[M].
	\end{equation*}
Next, we bound the error term making use of \cite[Theorem 7]{Lieb1973}  and \cite[Lemma 3.4]{Sutter2017}, which respectively concern Lieb's extension of Golden-Thompson inequality and the rotated expression for Lieb's pseudo-inversion operator via multivariate trace inequalities: Applying Lieb's theorem to inequality (\ref{eq:quasitensor-lnM}), we have:
	\begin{align*}
	\tr [M] =& \tr \left[ \operatorname{exp}\left( - \ln E_{C \cup D *}( \omega ) + \ln E_{C *}( \omega ) + \ln E_{ D *}( \omega ) \right) \right] \leq  \tr \left[ E_{ D *}( \omega )\mathcal T_{ E_{C \cup D *}( \omega )} (E_{C *}( \omega ))\right],
	\end{align*}
	where $\mathcal T_{ E_{C \cup D *}( \omega )} $ is given by:
	\begin{equation*}
	\mathcal T_{ E_{C \cup D *}( \omega )} (X) := \int_0^\infty  ( E_{C \cup D *}( \omega ) + t)^{-1} X ( E_{C \cup D *}( \omega ) + t)^{-1} dt \, ,
	\end{equation*}
	and because of multivariate trace inequalities \cite{Sutter2017}, 
	\begin{equation*}
	\tr [M] \leq \int_{-\infty}^{+\infty}  \,\tr \left[ E_{C *}( \omega )  \,  E_{C \cup D *}( \omega )^{\frac{-1-it}{2}}  E_{D *}( \omega ) \,  E_{C \cup D *}( \omega )^{\frac{-1+it}{2}} \right]\,\beta_0 (t) \, dt\,, 
	\end{equation*}
	with $\beta_0$ given by:
	\begin{equation*}
	\beta_0 (t) := \frac{\pi}{2} (\cosh (\pi t) + 1)^{-1} \, .
	\end{equation*}	
Now, note that if we subtract $ E_{C \cup D *}( \omega )$ from $ E_{C  *}( \omega )$ and $ E_{ D *}( \omega )$ we have:
		\begin{align*}
	&\tr \left[   \left(E_{C  *}( \omega )  - E_{C \cup D *}( \omega ) \right)  E_{C \cup D *}( \omega )^{\frac{-1-it}{2}}   \left( E_{D *}( \omega )   - E_{C \cup D *}( \omega ) \right)  E_{C \cup D *}( \omega )^{\frac{-1+it}{2}}  \right] 
	\\
	&\phantom{adasdadadsasdasdasd}= \tr \left[ E_{C  *}( \omega )  \, E_{C \cup D *}( \omega ) ^{\frac{-1-it}{2}}	\,  E_{D *}( \omega ) \, E_{C \cup D *}( \omega ) ^{\frac{-1+it}{2}} \right] -1 -1 + 1,
	\end{align*}
	since $ E_{C \cup D *}$, $E_{C  *}$ and $E_{ D *}$ are conditional expectations in the Schr\"{o}dinger picture and, thus, trace preserving. Therefore, defining $X:=\Gamma^{-1}_{E_{C\cup D*}(\omega)}(\omega)$, we have
	\begin{align*}
	\ln \tr[M] &	\leq \ln  \int_{-\infty}^{+\infty}  \tr \left[ E_{C  *}( \omega )  \, E_{C \cup D *}( \omega ) ^{\frac{-1-it}{2}}	\,  E_{D *}( \omega ) \, E_{C \cup D *}( \omega ) ^{\frac{-1+it}{2}} \right] \,\beta_0 (t) \, dt \, \\
	&= \ln \int_{-\infty}^{+\infty}  \left( \tr \left[   \left(E_{C  }[X   - E_{C \cup D }[ X ] \right) \Delta_{E_{C\cup D*}(\omega)}^{-\frac{it}{2}}   \left( E_{D *}( \omega )   - E_{C \cup D *}( \omega ) \right)    \right]   + 1  \right)\,\beta_0 (t) \,dt \\
	& \leq  \int_{-\infty}^{+\infty} \tr \left[   \left(E_{C  }[ X ]  - E_{C \cup D }[ X  \right) \Delta_{E_{C\cup D*}(\omega)}^{-\frac{it}{2}}  \left( E_{D *}( \omega )   - E_{C \cup D *}( \omega ) \right)   \right]  \, \beta_0 (t)\,dt\, \\
	& = \int_{-\infty}^{+\infty} \tr \left[   \left( X  - E_{C \cup D }[ X ] \right) \Delta_{E_{C\cup D*}(\omega)}^{-\frac{it}{2}}  \left( E_{C *} \circ E_{D *}( \omega )   - E_{C \cup D *}( \omega ) \right)   \right]  \, \beta_0 (t)\,dt\,,
	\end{align*}
	where we have used that $	\ln(x +1)\le x$ for positive real numbers and the monotonicity of the logarithm and $E_{C*} \circ E_{C \cup D *} = E_{C \cup D *}$. The result would then follow if we can show the following bound 
\begin{align*}
\underbrace{\tr\left[ (X-E_{C\cup D}[X])\,(E_{C*}\circ E_{D*}-E_{C\cup D*})(\omega_t)
 )\right]}_{=:(\#)} \le 2\widetilde{c}\,D(\omega\| E_{C\cup D*}(\omega))\,,
	\end{align*}
for $\widetilde{c}:=c\,|C\cup D|\, \e^{-\frac{\dist(C\backslash D,D\backslash C)}{\xi}}$ and $\omega_t:= \Delta_{E_{C\cup D*}(\omega)}^{-\frac{it}{2}}(\omega)$. 
However, since $\omega$ is by definition a state of $\cF(\cL_{A\cap (C\cup D)})$, and since $E_C$, $E_D$ and $E_{C\cup D}$ only act non-trivially on the blocks
  $\cK_{j}^{C\cup D}$, we can simplify the left-hand side above using \Cref{Pi}(ii): first of all, denote 
   \begin{align*}
  \omega:=\sum_{j,k}\tau^{(j,k)}\otimes \tr_{ \cH^{(j,k)}}\big[P^{(j,k)}\omega P^{(j,k)}  \big] \,,
 \end{align*}
for some full-rank normalized states $\tau^{(i,k)}$ supported in $ \cH^{(j,k)}$, where we denote by $P^{(j,k)}$ the orthogonal projection onto $\cH^{(j,k)}\otimes \cH_{j}^{C\cup D}$. Then,
 \begin{align*}
     (\#)&=\sum_{j,k} \, \tr\big[P^{(j,k)}\omega  \big] \\
     & \; \; \; \times \tr\big[(\tau^{(j,k)}-E_{C\cup D*}^{(j)}(\tau^{(j,k)}))\,(\tau_j^{C\cup D})^{-\frac{1}{2}}\,(E_{C*}^{(j)}\circ E_{D*}^{(j)}(\tau^{(j,k)}_t)-E_{C\cup D*}^{(j)}(\tau^{(j,k)}_t))(\tau_j^{C\cup D})^{-\frac{1}{2}} \big]\\
    &\le \sum_{j,k} \tr\big[P^{(j,k)}\omega  \big]\,\|\tau^{(j,k)}-E^{(j)}_{C\cup D*}(\tau^{(j,k)})\|_1\,\|(E_{C}^{(j)}\circ E_{D}^{(j)}-E_{C\cup D}^{(j)})(\Gamma_{\tau_j^{C\cup D}}^{-1}(\tau_t^{(j,k)}))\|_\infty\\
    &\le \widetilde{c}\sum_{j,k}\,\tr\big[P^{(j,k)}\omega  \big]\,\|\tau^{(j,k)}-\tau_j^{(C\cup D)}\|_1\,\|\tau_t^{(j,k)}-\tau_j^{C\cup D}\|_1\\
    &\le 2\widetilde{c}\,D(\omega\|E_{C\cup D*}(\omega))\,,
 \end{align*}
where $\tau^{(j,k)}_t:=  \Delta_{\tau^{{C\cup D}}_j}^{\frac{-it}{2}}(\tau^{(j,k)}) $, and where the last inequality follows from Pinsker's inequality.

\section{Proof of \Cref{maintheorem}}\label{appendix:nD}

In this appendix, we prove \Cref{maintheorem} for $n$-dimensional systems. The proof is similar to that of Section \ref{proof2D} in spirit, although the rhombi construction devised there in 2D cannot be extended to an $n$-dimensional setting. Therefore, we need to employ the notion of subordinated grained rectangles introduced in Section \ref{sec:mainresult}.

The strategy followed in the proof of \Cref{maintheorem} mainly consists of two technical results, which can be summarized as follows. First, we need again to reduce the problem of proving MLSI to that of proving it for a restricted class of approximately clustering states according to the tiling presented in Section \ref{geomconditions}. For that, we consider that tiling and use \Cref{Pi} together with the chain rule \eqref{chainrulerelatent}, obtaining for any state $\rho\in\cD(\cH_\Lambda)$:
	\begin{equation}\label{eqchainrule2}
	D(\rho \| \sigma^\Lambda ) = D(\rho\|E_{A\cap\Lambda*}(\rho))+D(E_{A\cap \Lambda*}(\rho)\|\sigma^\Lambda)
	\end{equation}
Then, \Cref{maintheorem} is a direct consequence of Lemma \ref{firstterm}, which was proven independently of the dimension and geometry, and the \Cref{secondterm2} below, to whom we devote the rest of the appendix. Note that Theorem \ref{maintheorem} would directly follow from Lemma \ref{firstterm} and Theorem \ref{secondterm2} analogously to what we showed in Section  \ref{proof2D}.

In what follows, a grained fat rectangle is denoted by $\widetilde{T}\equiv \widetilde{T}(A_0;k_1,\ldots,k_d)$, where the integers $k_1,\ldots,k_d$ correspond to the side lengths of $\widetilde{T}$ in pixels and call $K:= \text{max} \{ k_1,\ldots,k_d \}$ the \textit{size} of $\widetilde{T}$.  We also denote by $\widetilde{\mathcal{T}}_K$ the class of all such grained fat rectangles in $\ZZ^d$ of size at most $K \in \NN$ (in pixels)  and $$\ds \widetilde{\mathcal{T}}= \underset{K \geq 1}{\bigcup} \widetilde{\mathcal{T}}_K\,.$$ 
Moreover, we use the notation $L$ for the size of a grained rectangle in sites, namely the size of a rectangle $T:=T(x;l_1, \ldots , l_d)$ \textit{generated} by $\widetilde{T}$, which we define as the smallest rectangle for which $\widetilde{T}$ is its subordinated grained rectangle. In this case, we denote by $\widetilde{\mathcal{T}}_L$ the class of all grained fat rectangles in $\ZZ^d$ of size at most $L \in \NN$ (in sites). 

Now we are in position to state and prove the following result.

\begin{theorem}\label{secondterm2}

Given $\widetilde{T} \ssubset \ZZ^d$ a fat grained rectangle and under the conditions of \Cref{maintheorem},	there exists a constant $\beta>0$, independent of $|\widetilde{T} |$, such that for all $\rho\in\cD(\cH_{\widetilde{T}} )$,
\begin{align*}
4\beta D(E_{A \cap \widetilde{T}  *}(\rho)\|\sigma^{\widetilde{T}} )\le \operatorname{EP}_{\cL_{\widetilde{T}} }(\rho)\,.
\end{align*}
	\end{theorem}

 Let us recall the notion of pinched MLSI constant, which was introduced in Definition \ref{def:pinchedmlsi} independently of the geometry or the dimension. Then, for the first step of the proof consider a fat grained rectangle $\widetilde{T}  \subset \mathbb{Z}^d$ and split it into $C$ and $D$ as shown in \Cref{fig:1'}. Then, by virtue of the approximate tensorization for the relative entropy stated in \Cref{ATAC}, we can prove the following:

\begin{figure}[h!]
	\centering
	\includegraphics[width=0.7\linewidth]{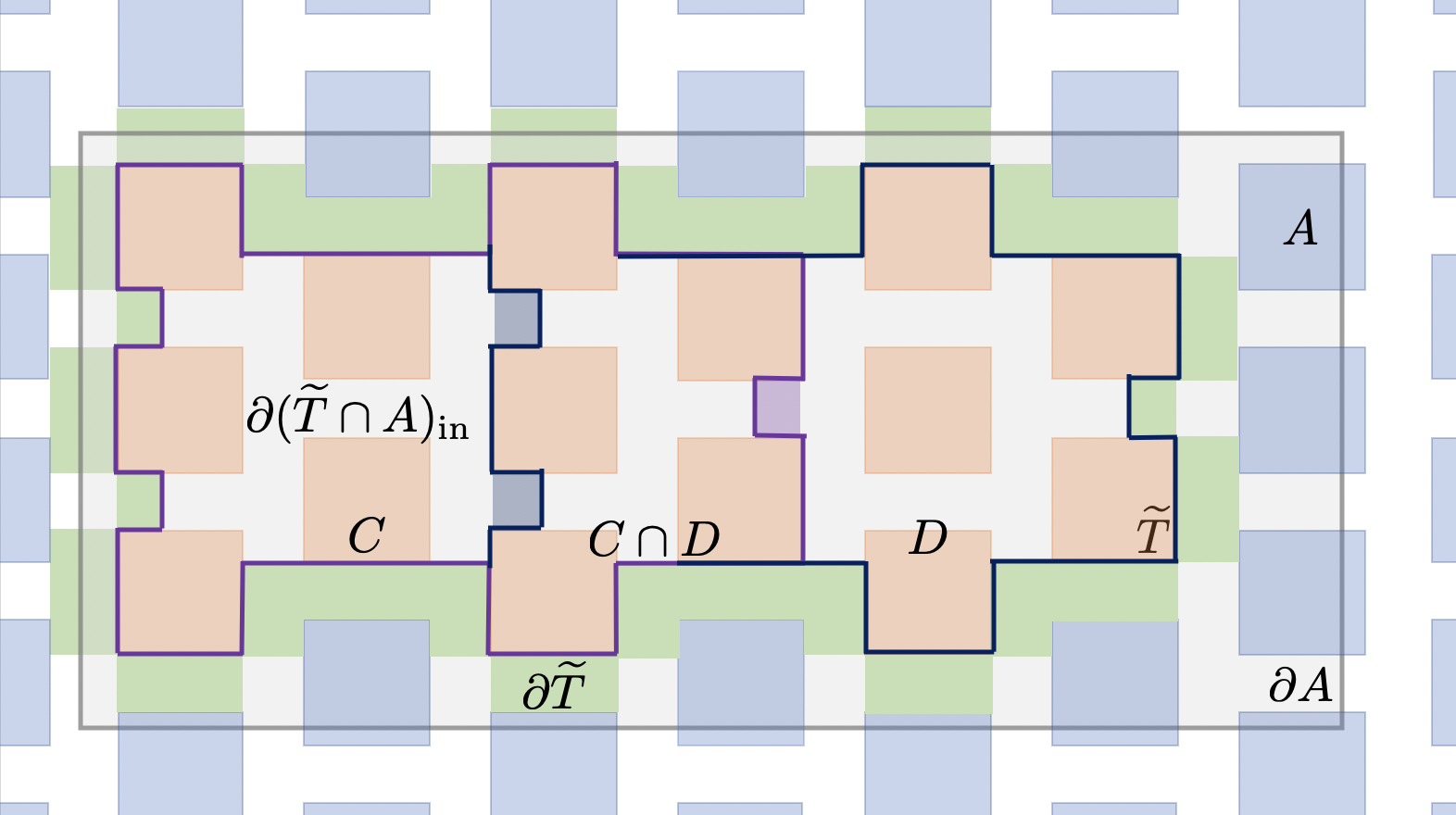}
	\caption{Splitting of a grained rectangle $\widetilde{T}=C\cup D$ into grained rectangles $C$ and $D$.}
	\label{fig:1'}
\end{figure}

\begin{step2}\label{step:12}
Assuming \Cref{LiLinftyaa}, the following holds for every $\rho \in \mathcal{D}(\mathcal{H}_{\widetilde{T}})$ and $C, D \subset {\widetilde{T}}$  such that $c\, |C \cup D | \,\e^{-\dist(C\backslash D,D\backslash C)/\xi}< 1/2 $ (see Figure \ref{fig:1}):
\begin{equation*}
D(E_{A\cap\widetilde{T} *}(\rho)\|E_{C\cup D*}\circ E_{A\cap{\widetilde{T}}*}(\rho))  \leq \frac{ \theta(C,D)}{4 \min \qty{\beta_{\widetilde{T}} (\LL_C), \beta_{\widetilde{T}}(\LL_{D})}} \left( \operatorname{EP}_{\LL_{C\cap D} }(\rho) + \operatorname{EP}_{\LL_{C \cup D}}(\rho)   \right),
\end{equation*}
where $\displaystyle \theta(C,D):= \frac{1}{1-2\, c \,|C \cup D | \, \e^{-\dist(C\backslash D,D\backslash C)/\xi} }  $.
\end{step2}

The proof of this result is completely analogous to that of \Cref{step:1}. Moreover, analogously to the proof of Theorem \ref{maintheorem2}, in the next step of the proof we need to choose $C$ and $D$ carefully so that $\theta(C,D)$ satisfies some desired decaying behaviour. For that, we will again consider $C$ and $D$ such that $|C\cup D|\sim L^d$ and dist$(C,D)= \sqrt{L}$ for a certain $L \in \mathbb{N}$, obtaining the necessary decay for $\theta(C,D)$ as a consequence of the fact that $\e^{-\sqrt{L}}$ decays faster than any polynomial.

In the second step of the proof, we split a fat grained rectangle into two smaller grained rectangles and get a lower bound for the Pinched MLSI constant of the former in terms of the Pinched MLSI constants of the latter. For that, we construct a suitable family of fat grained subrectangles in the grained rectangle $\widetilde{T}$, where we apply the previous step. Let $\widetilde{T}:=\widetilde{T}(A_0;k_1,\ldots ,k_d)$. Without loss of generality, assume that $k_1\leq \ldots \leq k_d$. Consider now the rectangle $T:=T(x;l_1, \ldots , l_d)$ \textit{generated} by $\widetilde{T}$, which we define as the smallest rectangle for which $\widetilde{T}$ is its subordinated grained rectangle. We assume for simplicity that $x=0$. Then, it is clear that for every $1\leq j \leq d$, the coordinates of both rectangles are related in the following way:
\begin{equation*}
    (2D-1)k_j + 2(\kappa-1)(k_j - 1) \leq l_j < (2D+1) k_j + 2 (\kappa - 1) k_j \, .
\end{equation*}
Let us also suppose that $l_1\leq \ldots \leq l_d$ and $(2D+1)^2 < l_d = 2L$, with $L$ large enough. We define $a_L:=\lfloor\sqrt{L}\rfloor$ and $n_L:= \lfloor \frac{L}{10a_L}  \rfloor$, where $\lfloor\cdot\rfloor$ denotes the integer part. For every integer $1 \leq n \leq n_L$, we cover $T$ with the following pair of sets:
\begin{equation}\label{AnBn'}
C_n=\qty{x \in T \, : \,  0 \leq x_d \leq \frac{l_d}{2}+ n a_L}\,,~~~
D_n=\qty{x \in T \, : \, \frac{l_d}{2} + (n-1)a_L < x_d \leq l_d}\,.
\end{equation}
Moreover, we denote by $\widetilde{C}_n$, resp. by $\widetilde{D}_n$, the grained sets subordinated to $C_n$, resp. to $D_n$. Then, it is clear by construction that $\widetilde{C}_n$ and $\widetilde{D}_n$ are both grained fat rectangles and $\widetilde{C}_n \cup \widetilde{D}_n = \widetilde{T}$ for every $1 \leq n \leq n_L$. Moreover, for $n$ fixed, it is clear that $\widetilde{C}_n \cap \widetilde{D}_n \neq \emptyset$ and the shortest side of the overlap has length of order $\sqrt{L}$ (due to the fact that we are considering $T$ a fat rectangle, so $l_1 \geq \frac{1}{10} \, l_d > \frac{L}{10}$ and if we had $\sqrt{L}>l_1$, we would have $\sqrt{L} > \frac{L}{10}$, or, equivalently, $\frac{L}{100} < 1$, which only holds for $L$ small). See Figure \ref{fig:2'}.

\begin{figure}[h!]
	\centering
	\includegraphics[width=0.7\linewidth]{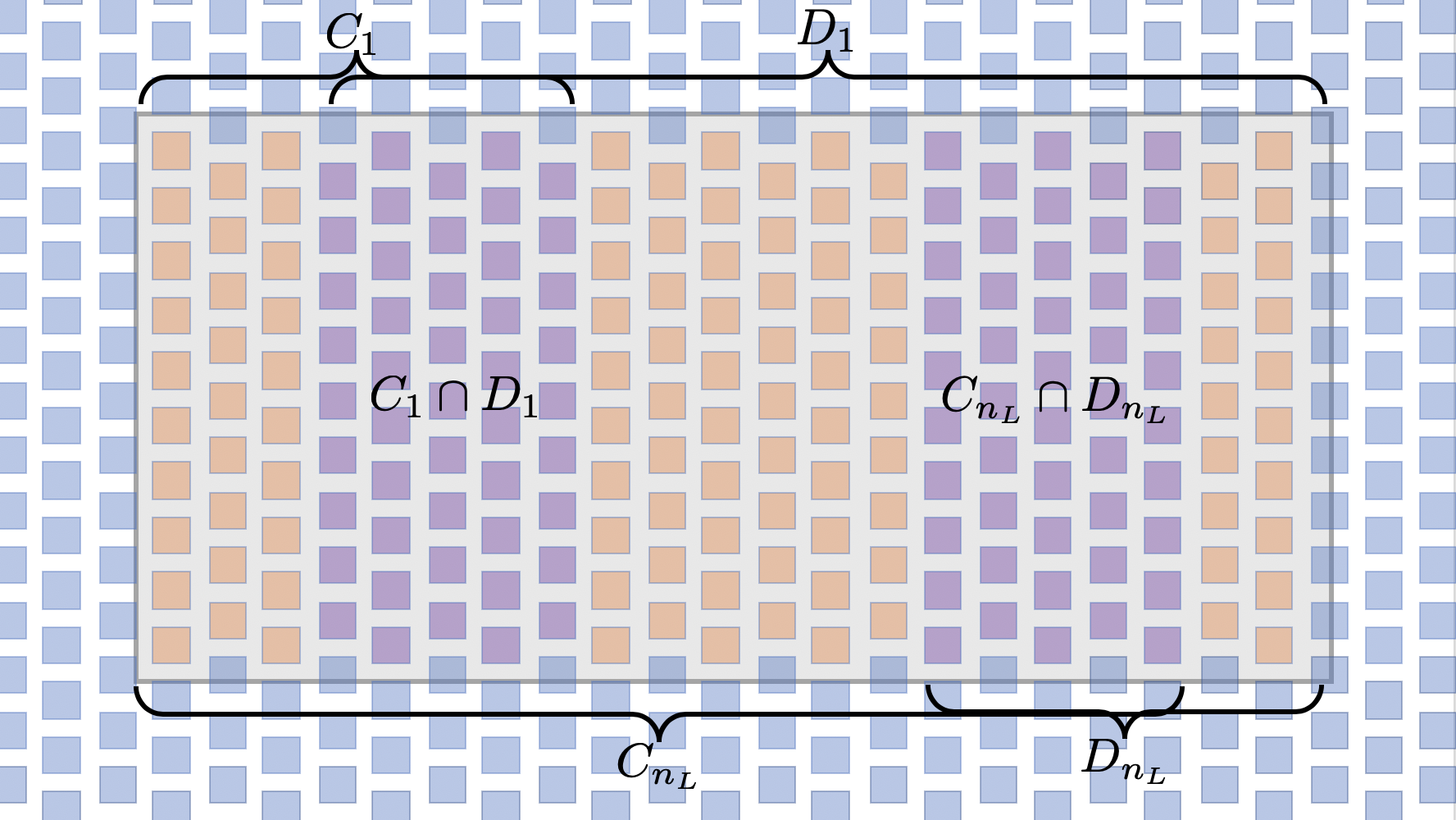}
	\caption{Splitting in $C_n$ and $D_n$.}
	\label{fig:2'}
\end{figure}

\newpage

\begin{step2}\label{step:22}
There exists a positive constant $C$, independent of the size of $\widetilde{T}$ such that:
\begin{equation}\label{ineq:L-large2}
\underset{n=1,\ldots,n_L}{\min} \qty{\beta_\Lambda(\LL_{\widetilde{C}_n}), \beta_\Lambda(\LL_{\widetilde{D}_n})}\left( 1+\frac{C}{\sqrt{L}} \right)^{-1} \leq \beta_\Lambda(\LL_{\widetilde{T}}),
\end{equation}
for every $1 \leq n  \leq n_L$ and $L$ large enough.
\end{step2}

\begin{proof}
	Once again, we denote $\omega:=E_{A\cap\Lambda*}(\rho)$. Then, using the sets $\widetilde{C}_n$ and $\widetilde{D}_n$ in the expression obtained in Step \ref{step:12}, we get, for every $1 \leq n  \leq n_L$,
\begin{equation}
D(\omega \| E_{\widetilde{T}*}(\omega)) \leq \frac{ \theta(\widetilde{C}_n,\widetilde{D}_n)}{4 \min \qty{\beta_\Lambda(\LL_{\widetilde{C}_n}), \beta_\Lambda(\LL_{\widetilde{D}_n})}} \left( \operatorname{EP}_{\LL_{\widetilde{C}_n \cap \widetilde{D}_n} }(\rho) + \operatorname{EP}_{\LL_{\widetilde{C}_n \cup \widetilde{D}_n}}(\rho)   \right),
\end{equation}
where
\begin{equation*}
\theta(\widetilde{C}_n,\widetilde{D}_n)=  \frac{1}{1-2\, c \, |\widetilde{C}_n \cup \widetilde{D}_n | \, \e^{-\sqrt{L}/\xi}    } \leq \frac{1}{1-2\, \widetilde{c} \, L^d \, \e^{-\sqrt{L}/\xi}    }  \, ,
\end{equation*}
for every $1 \leq n  \leq n_L$. Let us denote the latter by $ \theta(\sqrt{L})$. Now, by the definition of $\widetilde{C}_n$ and $\widetilde{D}_n$, the two following properties clearly hold:
\begin{enumerate}
\item $\widetilde{C}_i \cap \widetilde{D}_i \cap \widetilde{C}_j \cap \widetilde{D}_j = \emptyset$ for every $i \neq j$;
\item $ \ds \underset{1\leq n \leq n_L}{\bigcup} \left( \widetilde{C}_n \cap \widetilde{D}_n \right) \subseteq \widetilde{T} $.
\end{enumerate}
Therefore, as we did in Step \ref{step:2}, we can average over $n$ the previous expression to obtain:
\begin{align*}
D(\omega \| E_{\widetilde{T}*}(\omega)) &  \leq \frac{\theta(\sqrt{L})}{4 \, \underset{n=1,\ldots,n_L}{\min} \qty{\beta_\Lambda(\LL_{\widetilde{C}_n}), \beta_\Lambda(\LL_{\widetilde{D}_n})}}\left( 1+\frac{1}{n_L} \right)\operatorname{EP}_{\LL_{\widetilde{T}}}(\rho) \, .
\end{align*}
Hence, by the definition of $\beta_\Lambda(\LL_{\widetilde{R}})$,  for $L$ large enough we have
\begin{equation}\label{ineq:L-large12}
\underset{n=1,\ldots,n_L}{\min} \qty{\beta_\Lambda(\LL_{\widetilde{C}_n}), \beta_\Lambda(\LL_{\widetilde{D}_n})}\left( 1+\frac{C}{\sqrt{L}} \right)^{-1} \leq \beta_\Lambda(\LL_{\widetilde{R}}) \, ,
\end{equation}
for $C>1$ independent of $L$.
\end{proof}
Now, let us first define the following quantities for $L>1$:
\begin{equation}
T(L):= \underset{\widetilde{T} \in \widetilde{\mathcal{T}}_L}{\inf} \beta_\Lambda(\LL_{\widetilde{T}})\,.
\end{equation}
In the next step, we obtain a recursion between the quantities $T(L)$ which will later allow us to get a lower bound for the global MLSI constant in terms of size-fixed Pinched MLSI constants. 
\begin{step2}\label{step:32}
There exists a positive constant $C$ independent of the size of $T$ such that
\begin{equation}
T(2L)\geq  \left( 1+\frac{C}{\sqrt{L}} \right)^{-3d} T(L)\phantom{asdd} \text{for } L \text{ large enough}.
\end{equation}

\end{step2}

\begin{proof}
Consider the expression obtained in the previous step. Let us analyze the value of the MLSI constant in the grained rectangles $\widetilde{C}_n$ and $\widetilde{D}_n$. Let us consider the grained rectangle $\widetilde{C}_n$ (the analysis is analogous for $\widetilde{D}_n$). The side of $\widetilde{C}_n$ corresponding to the coordinate $x_d$ has length (in sites) less than or equal to $1.2L$, by definition of $C_n$. For the other sides, we have to distinguish between two different cases.
\begin{enumerate}
\item If $ \displaystyle \max\qty{l_k \, : \, k=1, \ldots d-1} \leq \frac{3}{2} L$, then the longest side of $\widetilde{C}_n $ is less than or equal to $ \displaystyle \frac{3}{2}L$, so $\displaystyle \widetilde{C}_n \in \widetilde{\mathcal{T}}_{\frac{3}{2}L}$ and $ \displaystyle \beta_\Lambda(\LL_{\widetilde{C}_n}) \geq T\Big(\,\frac{3}{2}L\,\Big)$.
\item If the largest side of $\widetilde{C}_n$, which we call $l_i$, satisfies $\displaystyle l_i > \frac{3}{2}L$, it is clear that $\widetilde{C}_n$ verifies max$\qty{l_k}>1.5 L$ and  min$\qty{l_k}\leq 1.2 L$. Hence,
\begin{equation}
\beta_\Lambda(\LL_{\widetilde{C}_n}) \geq \underset{\widetilde{T}:\,\max\qty{l_k}>1.5 L, \, \min\qty{l_k}\leq 1.2 L}{\min} \beta_\Lambda(\LL_{\widetilde{T}})\,.
\end{equation}
\end{enumerate}
Therefore, from the right-hand side of equation (\ref{ineq:L-large2}), we have
\begin{multline*}
\left( 1 + \frac{C}{\sqrt{L}}  \right)^{-1} \underset{n=1,\ldots,n_L}{\min} \qty{\beta_\Lambda(\LL_{\widetilde{C}_n}), \beta_\Lambda(\LL_{\widetilde{D}_n})}\\
 \geq \left( 1 + \frac{C}{\sqrt{L}}  \right)^{-1} \min\qty{T\left(\frac{3}{2}L\right), \underset{R:\,\max\qty{l_k}>1.5 L, \, \min\qty{l_k}\leq 1.2 L}{\min} \beta_\Lambda(\LL_{\widetilde{T}})}\,.
\end{multline*}
Now, we consider a grained fat rectangle in $\widetilde{\mathcal{T}}_{2L}$ such that its longest side is greater than or equal to $1.5 L$ and its shortest side has length less than or equal to $1.2 L$. Iterating Step \ref{step:22} at most $d-1$ times on that grained rectangle, we end up with a grained rectangle whose longest side is shorter than or equal to $1.5 L$. Hence,
\begin{equation}
\underset{\widetilde{T}:\,\max\qty{l_k}>1.5 L, \, \min\qty{l_k}\leq 1.2 L}{\min} \beta_\Lambda(\LL_{\widetilde{T}}) \geq  \left( 1 + \frac{C}{\sqrt{L}} \right)^{-(d-1)} T \left( \frac{3}{2}L \right)\,.
\end{equation}
Therefore,
\begin{equation}
 \beta_\Lambda(\LL_{\widetilde{R}}) \geq \left( 1 + \frac{C}{\sqrt{L}} \right)^{-d} T \left( \frac{3}{2}L \right)\,,
\end{equation}
and since the rectangle that we were considering in Step \ref{step:22} verified  $\widetilde{T} \in \widetilde{\mathcal{T}}_{2L}$, we obtain
\begin{equation}
T(2L) \geq \left( 1 + \frac{C}{\sqrt{L}} \right)^{-d} T \left( \frac{3}{2}L \right).
\end{equation}
To conclude, we iterate this expression two more times to obtain
\begin{equation}
T(2L) \geq \left( 1 + \frac{C}{\sqrt{L}} \right)^{-d}  \left( 1 + \frac{C}{\sqrt{\frac{3L}{4}}} \right)^{-d} \left( 1 + \frac{C}{\sqrt{\frac{9L}{16}}} \right)^{-d} T \left( \frac{27}{32}L \right),
\end{equation}
and since $ \displaystyle S\left( \frac{27}{32} L \right)\geq S (L)$, we obtain
\begin{equation}
T(2L) \geq \left( 1+ \frac{C}{\sqrt{L}} \right)^{-3d} T(L) \, ,
\end{equation}
where $\displaystyle C$ is a constant independent of the size of the system.
\end{proof}
Finally, in the last step of the proof, using recursively the relation obtained in the previous one, we get a lower bound for the global MLSI constant in terms of complete MLSI constants.  Similar to above, we define the quantities for $L>1$:
\begin{equation}
U(L):= \underset{\widetilde{T} \in \widetilde{\mathcal{T}}_L}{\inf} \alpha_\ccc(\LL_{\widetilde{T}})\,.
\end{equation}
\begin{step2}\label{step:42}
There exists a constant $L_0\in\NN$, independent of $\Lambda$ such that the following holds:
\begin{equation*}
\alpha( \LL_\Lambda)  \geq  \Psi (L_0) \,U(L_0)\,,
\end{equation*}
where $\Psi (L_0)$ does not depend on the size of $\Lambda$.
\end{step2}

\begin{proof}

Let us denote by  $L_0$ the first integer for which inequality (\ref{ineq:L-large12}) holds. By virtue of the previous step, it is clear that the following holds for $L_0$:
\begin{equation}
T(2L_0) \geq \left( 1+ \frac{C}{\sqrt{L_0}} \right)^{-3d} T(L_0),
\end{equation}
Note that the limit of $\Lambda$ tending to $\ZZ^d$ is the same as the one of $T(n L_0)$ with $n$ tending to infinity. Therefore,
 \begin{align*}
  \underset{\Lambda \rightarrow \ZZ^d}{\lim} \alpha_\Lambda(\LL_\Lambda) & =  \underset{n \rightarrow \infty}{\lim} T(2^n L_0) \\
  & \geq   \left( \underset{n=1}{\overset{\infty}{\prod}} \left( 1+ \frac{C}{\sqrt{2^{n-1}L_0}} \right) \right)^{-3d} T(L_0) \\
 &  \geq \left( \exp \left[ \underset{n=0}{\overset{\infty}{\sum}} \, \frac{C}{2^n L_0} \right] \right)^{-3d} T(L_0) \\
 & = \exp \left[\frac{-3 d C}{L_0}(2+\sqrt{2}) \right] T(L_0),
 \end{align*}
where the constants $L_0$ and $C$ do not depend on the size of $\Lambda$. We conclude from the following simple observation that 
\begin{align*}
T(L_0)=\underset{\widetilde{T}\in\widetilde{\mathcal{T}}_{L_0}}{\inf}\beta_\ccc(\LL_{\widetilde{T}})\overset{(1)}{\ge}
\underset{\widetilde{T}\in\widetilde{\mathcal{T}}_{L_0}}{\inf} \alpha_\ccc(\LL_{\widetilde{T}})=U(L_0)\,.
\end{align*}
where $(1)$ follows from an obvious extension of \Cref{betatoalpha} to grained rectangles, namely 
	\begin{align*}
	\beta_\Lambda(\cL_{\widetilde{T}})\ge \alpha_{\ccc}(\cL_{\widetilde{T}})\,,
	\end{align*}
	for any grained rectangle $\widetilde{T} \subseteq \Lambda\ssubset \ZZ^d$.

\end{proof}

\end{document}